\documentclass[11pt]{amsart}
\pdfoutput=1
\usepackage[english]{babel}
\usepackage{graphicx,amsmath}
\usepackage{enumitem}
\usepackage{float}
\usepackage{floatflt,young,youngtab}
\usepackage{tikz}
\theoremstyle{plain}
\numberwithin{equation}{section}
 
\newtheorem{theorem}{Theorem}[section]
\newtheorem{remark}{Remark}[section]
\newtheorem{proposition}{Proposition}[section]

\newtheorem{lemma}{Lemma}[section]
\newtheorem{definition}{Definition}[section]
\newtheorem{example}{Example}[section]

\newtheorem{construction}{Construction}[section]

\def \Z {{\mathbb Z}}
\def \S {{\mathcal S}_{\mathcal M}^{\mbox{\tiny TNN}}}
\def \Grkn {Gr^{\mbox{\tiny TNN}}(k,n)}
\newcommand\mycom[2]{\genfrac{}{}{0pt}{}{#1}{#2}}
\def \DKP {{\mathcal D}_{\textup{\scriptsize KP},\Gamma}}
\def \DDG {{\mathcal D}_{\textup{\scriptsize dr},\Gamma}}
\def \DS {{\mathcal D}_{\textup{\scriptsize S},\Gamma_0}}
\def \DVG {{\mathcal D}_{\textup{\scriptsize vac},\Gamma}}
\def \DVN {{\mathcal D}_{\textup{\scriptsize vac},{\mathcal N}^{\prime}}}
\def \DDN {{\mathcal D}_{\textup{\scriptsize dr},{\mathcal N}^{\prime}}}
\def \gvac {\gamma^{\textup{\scriptsize (vac)}}}
\def \gdr {\gamma^{\textup{\scriptsize (dr)}}}
\def \gs {\gamma^{\textup{\scriptsize (S)}}}
\def \Pvac {P^{\textup{\scriptsize (vac)}}}
\def \Pdr {P^{\textup{\scriptsize (dr)}}}
\def \Ps {P^{\textup{\scriptsize (S)}}}
\def \sign {\mbox{sign\,}}

\title[$M$-curves for Le-networks and KP--II solitons]{Reducible $M$-curves for Le-networks in the totally-nonnegative Grassmannian and KP--II multiline solitons}
\author{Simonetta Abenda}
\address{Dipartimento di Matematica, Universit\`a di Bologna, P.zza di Porta San Donato 5, I-40126 Bologna BO, ITALY
}
\email{simonetta.abenda@unibo.it
}
\author{Petr G. Grinevich}
\address{L.D.Landau Institute for Theoretical Physics,
pr. Ak Semenova 1a, Chernogolovka, 142432, Russia,
{\footnotesize pgg@landau.ac.ru}\\
Lomonosov Moscow State University,
Faculty of Mechanics and Mathematics, 
Russia, 119991, Moscow, GSP-1, 1 Leninskiye Gory, Main Building,\\
Moscow Institute of Physics and Technology, 
9 Institutskiy per., Dolgoprudny,
Moscow Region, 141700, Russia.}

\thanks{
This research has been partially supported by GNFM-INDAM and RFO University of Bologna, by the Russian Foundation for Basic Research, grant 17-01-00366, 
by the program ``Fundamental problems of nonlinear dynamics'', Presidium of RAS. Partially this research was fulfilled during the visit of the second author (P.G.) to IHES, 
Université Paris-Saclay, France in November 2017.}

\begin{document}

\begin{abstract}
{
We associate real and regular algebraic--geometric data to each multi--line soliton solution of Kadomtsev-Petviashvili II (KP) equation. These solutions are known to be parametrized  by points of the totally non--negative part of real Grassmannians $\Grkn$. In \cite{AG1} we were able to construct real algebraic-geometric data for soliton data in the main cell $Gr^{\mbox{\tiny TP}} (k,n)$ only. Here we do not just extend that construction to all points in $\Grkn$, but we also considerably simplify it, since both the reducible rational $\mathtt M$--curve $\Gamma$ and the real regular KP divisor on $\Gamma$  are directly related to the parametrization of positroid cells in $\Grkn$ via the Le--networks introduced in \cite{Pos}. In particular, the direct relation of our construction to the Le--networks guarantees that the genus of the underlying smooth $\mathtt M$--curve is minimal and it coincides with the dimension of the positroid cell in $\Grkn$ to which the soliton data belong to. Finally, we apply our construction to soliton data in $Gr^{\mbox{\tiny TP}}(2,4)$ and we compare it with that in \cite{AG1}.} 


\end{abstract}
\maketitle

\tableofcontents
\section{Introduction}

The deep relation between the asymptotic behavior of real bounded multi--line KP\footnote{Throughout the paper, we always use the notation KP for KP II, with the heat conductivity operator in the Lax pair.} soliton solutions, asymptotic web networks and total positivity has been unveiled in a series of papers (see \cite{BC,BPPP,CK2,CK,DMH,Kod1,Kod2,KW1,KW2, Z} and references therein). On the other side, real regular KP finite-gap solutions are associated to  $\mathtt M$--curves \cite{DN}, and soliton solutions can be obtained as degenerations of complex finite-gap ones, when some cycles on the spectral curves shrink to double points. As pointed out by S.P. Novikov, it is natural to check whether \textbf{real regular} degenerate solutions may be obtained by degenerating \textbf{real regular} finite-gap solutions. In particular, in the case of real bounded multi-line KP solitons, this means to investigate whether they can be obtained by degenerating smooth  $\mathtt M$--curves to rational (reducible) ones. Therefore we have started to search for new relations between total positivity in Grassmannians \cite{Lus1,Lus2,Pos,PSW,Rie} and $\mathtt M$--curves \cite{Har,Gud,Nat,Vi}, by connecting two relevant approaches used in KP theory to classify solutions, the Sato Grassmannian \cite{S} and KP finite--gap theory \cite{Kr1,Kr2,DN}, and in \cite{AG1} we have succeeded in establishing a new connection between classical total positivity \cite{Kar,Pin} and rational degenerations of $\mathtt M$--curves when the KP soliton data belonging to the main cells, $Gr^{\mbox{\tiny TP}}(k,n)\subset\Grkn$. 

Here we construct an analogous connection between all positroid cells in $\Grkn$ and  $\mathtt M$--curves. More precisely, we provide a new interpretation of minimal parametrizations of $g$--dimensional positroid cells $\S\subset\Grkn$ via degree $g$ real and regular KP  divisors on reducible $\mathtt M$--curves which are rational degenerations of genus $g$ $\mathtt M$--curves, using the Le--networks introduced in \cite{Pos}. To the Le--graph representing $\S$ \cite{Pos}, we canonically associate an universal reducible curve $\Gamma =\Gamma(\S)$ with $g+1$ ovals which is a rational degeneration of a genus $g$ smooth $\mathtt M$--curve. Then we establish a canonical relation between points in $\S$ parametrized by Le--networks and degree $g$ divisors on $\Gamma$. Since such networks provide a minimal parametrization of positroid cells \cite{Pos} and the degree of the KP divisor coincides with the dimension of $\S$, our parametrization is optimal for generic soliton data.

The starting point is the fact that regular multi--line KP solitons are obtained in well--defined finite--dimensional reductions  of the Sato Grassmannian \cite{S}. More precisely,
each family of KP real regular multiline soliton solutions corresponds to soliton data in a uniquely identified $d$--dimensional irreducible positroid cell $\S$ in a totally non--negative Grassmannian $\Grkn$\footnote{Each family of soliton solutions is also realized in an infinite number of $d$--dimensional reducible positroid cells in $Gr^{\mbox{\tiny TNN}}(k^{\prime}, n^{\prime})$ with $k^{\prime}\ge k$ and $n^{\prime}> n$.} (see \cite{CK, KW1,KW2} and references therein).
In this setting, the soliton data are a set of ordered phases ${\mathcal K} =\{ \kappa_1 < \cdots < \kappa_n\}$ and a point $[A] \in \S$. Then the corresponding KP multiline soliton solution is real regular for real times and the direct spectral approach  provides a
$k$ point divisor on a rational curve $\Gamma_0$ with $n$ marked points corresponding to the phases $\mathcal K$ and a marked point corresponding to the essential singularity of the normalized KP wave function \cite{Mal}. However, these spectral data are, in general, insufficient to reconstruct soliton data varying in $\S$ since $\max\{k, n-k\}\le d \le k(n-k)$.

On the other side, in principle, soliton solutions can be obtained by degenerating finite-gap solutions on smooth curves in the limit when some gaps degenerate to double points, as it was first observed in \cite{Nov} for the case of the Korteweg--de Vries equation.\footnote{ Using other
degenerations analogous to those in \cite{DKMM}, one can construct other interesting classes of KP solutions, including the rational ones. For additional information about singular spectral curves in soliton theory see \cite{Taim}.} In particular, 
finite-gap KP solutions were constructed in \cite{Kr1,Kr2}: they are parametrized  by degree $g$ non-special divisors on genus $g$ Riemann surfaces with a marked point. Real regular finite-gap KP solutions 
correspond to algebraic data on genus $g$ ${\mathtt M}$-curves satisfying natural constraints \cite{DN}: by definition, the curve has $g+1$ real ovals, and 
one of them contains the marked point, while each other oval contains exactly one divisor point.

Therefore, in order to obtain a one--to--one correspondence between soliton data varying in a given positroid cell and KP divisors, it is 
natural to impose that $\Gamma_0$ is an irreducible component of a reducible spectral curve. This approach is fully justified analytically by the extension of finite--gap theory to degenerate solutions, like solitons, on reducible curves in \cite{Kr3}. 
However, as remarked in \cite{Kr3}, the finite--gap approach on reducible curves
is ill--defined, in the sense that, due to degeneracy, there is not a unique way to extend the Baker--Akhiezer function. It is then relevant to search for canonical constructions compatible with the degeneration from regular finite--gap solutions. Since real regular quasi--periodic KP solutions are parametrized by real and regular divisors on smooth $\mathtt M$--curves \cite{DN},
we construct reducible curves $\Gamma$, which are rational degenerations of smooth $\mathtt M$--curves, and KP divisors satisfying reality and regularity conditions compatible with those settled in \cite{DN}.
The relevance of the construction proposed in this paper relies on the fact that, on one side, it perfectly matches the reality problem for KP finite--gap theory and, on the other side, provides a canonical parametrization of positroid cells. 
More precisely, in this paper:
\begin{enumerate}
\item We associate a canonical reducible $\mathtt M$--curve $\Gamma$ to the Le--graph ${\mathcal G}$ describing the corresponding cell and we prove that it is a rational degeneration of a smooth $\mathtt M$--curve of genus equal to the dimension $g$ of the positroid cell. This curve $\Gamma$ contains the rational curve $\Gamma_0$ coming from the direct spectral analysis as one of its irreducible components;
\item We then provide a parametrization of each $g$--dimensional positroid cell by real regular degree $g$ non--special divisors $\DKP$. The Sato divisor coincides with $\DKP\cap\Gamma_0$.
\end{enumerate}
We have decided to treat the Le--network case separately both because the reducible curve for the Le--network is the rational degeneration of a smooth $\mathtt M$--curve of \textbf{minimal genus} equal to the dimension of the positroid cell $\S$ and because we get a parametrization of $\S$ via degree $g$ non--special KP divisors.
Moreover, the explicit use of Le--networks is directly connected
to the construction proposed in \cite{AG1} for the case $\S = Gr^{\mbox{\tiny TP}}(k,n)$ and it also considerably simplifies it. Finally, throughout this paper, we carry out explicitly the construction for the usual acyclic orientation of the Le--network  and we postpone to \cite{AG2} the proof of the invariance of the KP divisor with respect to changes of orientation of the network and the generalization of this construction to all Postnikov networks.

\smallskip

Before outlining our construction, we would like to point out that there are other well-known relations of different nature between networks, algebraic curves and integrable systems in literature. In dimer models with periodic boundary conditions (models on tori) the Riemann surfaces arise as the spectral curves for operators on networks on tori \cite{KSO}, and such spectral curves, which are generically regular, may be associated to classical or quantum integrable systems \cite{KG}. Another big area of activity is currently associated with the use of planar networks in the disk for the computation of scattering amplitudes in $N=4$ super Yang-Mills on-shell diagrams, see \cite{AGP1,AGP2, ADM} and references therein. We have noticed an analogy between the momentum--helicity conservation relations in the trivalent planar networks in the approach of \cite{AGP1,AGP2} and the relations satisfied by the vacuum and dressed edge wave functions in our approach. Consequently, a relevant open problem is whether our approach for KP may be interpreted as a scalar analog of a field theoretic model. Finally another relevant open problem is the connection of our construction to that in \cite{KW2} where the asymptotic behavior of the multiline KP soliton solutions in the $(x,y)$-plane for large time $t$ are shown to give rise to soliton webs interpreted in terms of real tropical geometry \cite{IMS} and cluster algebras \cite{FZ2}. In our approach the KP solution plays the role of a potential in the spectral problem and its asymptotic behavior should be put in relation to that of KP zero divisors on $\Gamma$.

Total positivity itself, since its appearance in \cite{Sch}, naturally arises in many applications in connection with some reality properties of the system. In particular, important connections between positivity and oscillatory properties of mechanical systems were found in \cite{GK,GK2}. The extension of the positivity property to integral kernels was investigated in \cite{Kar}. Extension of total positivity to split reductive connected algebraic groups and flag manifolds was developped in \cite{Lus1,Lus2}. Total positivity in classical and generalized sense is also one of the basic concepts in the theory of cluster manifolds and cluster algebras \cite{FZ1,FZ2}, see also the book \cite{GSV}. Applications of the theory of total positivity in Lie groups to the study of homomorphisms of the fundamental group of a closed surface into a Lie group is considered in \cite{FG}. Non-negativity of systems of modified Bessel functions of the first kind arising in the solutions to a model of overdamped Josephson junction was studied in \cite{BG,BG2}. Of course, this list of literature is far from being complete. 

Finally the positroid stratification of $\Grkn$ \cite{Pos} is naturally related to the Gelfand-Serganova stratification of the complex Grassmannian $Gr(k,n)$ \cite{GGMS,GS}. As observed in \cite{Pos}, the real positivity condition essentially simplifies the problem, whereas the geometrical structure of the strata for complex Grassmannians can be as complicated as essentially any algebraic variety \cite{Mn}. Let us point out that full Gelfand-Serganova stratification corresponds to the action of all complex KP flows on the Grassmannians. In particular, the factor-space of the Grassmannians by the action of compact tori corresponding to pure imaginary times has interesting topology studied in \cite{BT1,BT2}.

\smallskip

\paragraph{\textbf{Outline of the main construction}} 
In this paper, to any given ordered $n$-set of real phases $\mathcal K$ and $g$-dimensional positroid cell in $\S\subset\Grkn$,  we associate a canonical curve $\Gamma$ which is a rational degeneration of a smooth $\mathtt M$-curve of minimal genus $g$, and we parametrize these cells via real regular KP divisors on these curves. For any soliton data $({\mathcal K},[A])$, $[A]\in\S$, the divisor can be made non-special by a proper choice of the normalization time $\vec t_0$. In our construction an essential tool  is the parametrization of all positroid cells in $\Grkn$ by the Le--networks introduced in \cite{Pos}, which, in particular, guarantees both the non-specialty of the divisor and the minimality of the genus. 

Indeed, given $\mathcal K$ and the planar trivalent bipartite Le--graph  ${\mathcal G}$ in the disk representing $\S$, to such data we associate a unique reducible rational curve $\Gamma=\Gamma(\mathcal G,\mathcal K)$  using the following natural correspondence: 
\begin{enumerate}
\item The boundary of the disk corresponds to the rational component $\Gamma_0$ containing the essential singularity of the KP wave function and the Sato divisor, while the $n$ boundary vertices correspond to the $n$ marked phases $\kappa_j$ on $\Gamma_0$;
\item Each bivalent or trivalent internal vertex corresponds to a rational component of $\Gamma$. The black and white colors are related to the different analytic properties of the KP wave function on the corresponding rational components. The divisor points are associated to white trivalent vertices through linear relations;
\item Each edge corresponds to a double point of $\Gamma$ where different components are glued. Thanks to the trivalency assumption, each $\mathbb{CP}^1$ component carries three marked points and we avoid the introduction of parameters marking double points, and the curve $\Gamma$ is the same for all points in $\S$;
\item Faces of the graph correspond to ovals of $\Gamma$;
\item The canonical acyclic orientation of the Le--graph in \cite{Pos} is associated to a well-defined choice of coordinates on the rational components of $\Gamma$.
\end{enumerate}
We remark that the above correspondence is a minor modification of a special case in the representation of reducible curves by dual graphs (see, for example, \cite{ACG}, Section X). Non rational components are allowed in degenerate finite--gap theory on reducible curves as well; however, the rational ansatz for $\Gamma\backslash \Gamma_0$ considerably simplifies the overall construction.

By our construction, $\Gamma$ is a real curve with $g+1$ ovals and is the rational degeneration of an $\mathtt M$--curve of genus $g$. If the soliton data belong to the top cell $Gr^{\mbox{\tiny TP}}(k,n)$, then $g=k(n-k)$ and the curve $\Gamma(\xi)$ constructed in \cite{AG1} corresponds to a particular desingularization of $\Gamma({\mathcal G})$ which reduces the number of rational components to $k+1$. We thouroughly discuss such desingularization in the simplest non--trivial case of soliton data in $Gr^{\mbox{\tiny TP}} (2,4)$. 

Then, we fix a point $[A] \in \S$ and extend its KP wave function from $\Gamma_0$ to $\Gamma$. We must control that at each pair of double points the values of the normalized KP wave function coincide for all times. In particular, this requirement has to be satisfied at the double points connecting components to $\Gamma_0$. Moreover, if we have linear relations between the values of the normalized KP wave function at the marked points of any given component, then its meromorphic  extension to the whole curve is canonical. To define the wave function at all double points, we use the canonically oriented Le--network ${\mathcal N}$ representing $[A]$. 

On $\mathcal N$ we first construct a system of \textbf{edge} vectors satisfying linear relations at the vertices. Each component of a given edge vector coincides, up to a sign, with the sum of the weights of all paths starting at the given edge and ending at the same boundary sink vertex. We remark that, for soliton data in $Gr^{\mbox{\tiny TP}}(k,n)$, the recursive construction of such system of vectors generalizes the algebraic construction in \cite{AG1}. 

We then use this system of vectors to construct both a vacuum edge wave function $\Phi(\vec t)$ and its dressing $\Psi (\vec t)$ on the Le--network. At this step, we modify the original network adding an univalent internal vertex next to each boundary source vertex using Postnikov move (M2) in order that the vacuum edge wave function satisfies Sato boundary conditions on $\Gamma_0$ and an edge vector corresponds to each Darboux point in $\Gamma$. Then, using the linear relations satisfied by the vacuum edge wave function at the vertices, we associate a canonical real number, which we call vacuum network divisor number, to each trivalent white vertex of the modified network $\mathcal N^{\prime}$. Similarly,  using the same linear relations for the dressed edge wave function, we also associate another canonical real number, which we call dressed network divisor number, to any trivalent white vertex of $\mathcal N^{\prime}$ not containing a Darboux edge. 

Moreover, thanks to these linear relations, the normalized dressed wave function admits degree one meromorphic extension on each component corresponding to a trivalent white vertex of $\mathcal N$, and admits constant in the spectral parameter extension on each other component. The dressed network divisor numbers become the coordinates of the divisor points on the corresponding components. The full KP divisor is the sum of these points on $\Gamma\backslash\Gamma_0$ and the Sato divisor on $\Gamma_0$, it is effective, non-special and has degree $g$. Finally, we check that each finite oval of $\Gamma$ contains exactly one divisor point. 

We would like to remark that the divisors on networks introduced in our text are different from the commonly used divisors on graphs, see for example, \cite{BN}. To each trivalent white vertex we associate not only the multiplicity of divisor (it is always 1 in our setting), but also its position on the real part of the corresponding rational component, which is a real number.  

\smallskip

\paragraph{\textbf{Plan of the paper}:} We did our best to make the paper self--contained. In Section \ref{sec:soliton_theory} and in Appendix \ref{app:TNN}, we briefly present a review of the necessary results respectively for KP soliton theory and totally non--negative Grassmannians. In Section \ref{sec:3} we outline the main construction and state the principal theorems. In Section \ref{sec:le}, we link the main algebraic construction in \cite{AG1} to the Le--networks and extend it to any positroid cell. The construction of the system of edge vectors on the Le--networkn and the proof of the main Theorem on the characterization of the vacuum divisor is carried out in Section~\ref{sec:proof}.
In section \ref{sec:example} we apply our construction to soliton data in $Gr^{\mbox{\tiny TP}} (2,4)$ and compare it with \cite{AG1}.

{\bf Notations:} We use the following notations throughout the paper:
\begin{enumerate}
\item $k$ and $n$ are positive integers such that $k<n$;
\item  for $s\in {\mathbb N}$ let $[s] =\{ 1,2,\dots, s\}$; if $s,j \in {\mathbb N}$, $s<j$, then
$[s,j] =\{ s, s+1, s+2,\dots, j-1,j\}$;
\item  ${\vec t} = (t_1,t_2,t_3,\dots)$, where
$t_1=x$, $t_2=y$, $t_3=t$. Throughout the paper we assume that $\vec t$ always has only a finite number of non-zero components, but this number can be arbitrarily large;
\item $\theta(\zeta,\vec t)= \sum\limits_{s=1}^{\infty} \zeta^s t_s,$
\item {we} denote the real phases 
$\kappa_1< \kappa_2 < \cdots < \kappa_n$ and
$\theta_j \equiv \theta (\kappa_j, \vec t)$.
\end{enumerate}

\section{KP--II multi-line solitons}\label{sec:soliton_theory}

In this section, we review the characterization of real bounded regular multiline KP soliton solutions via Darboux transformations, Sato's dressing transformations and finite gap--theory.
The KP-II equation \cite{KP}
\begin{equation}\label{eq:KP}
(-4u_t+6uu_x+u_{xxx})_x+3u_{yy}=0,
\end{equation}
is the first non--trivial flow of an integrable hierarchy \cite{D,DKN,H,MJD,S}. In the following we denote $\vec t = (t_1=x,t_2=y,t_3=t, t_4,\dots)$. The family of solutions we consider belong to the class of real regular exact KP solutions used, in particular, to model the shallow water waves in the approximation where the surface tension is negligible.

\subsection{The heat hierarchy and the dressing transformation}\label{sec:Sato}

Multiline KP solitons may be realized starting from the soliton data $({\mathcal K}, [A])$, where 
${\mathcal K}$ is a set of real ordered phases $\kappa_1<\cdots<\kappa_n$,  $A =( A^i_j )$ is a $k\times n$ real matrix of rank $k$ and $[A]$ denotes
the point in the finite dimensional real Grassmannian $Gr (k,n)$ corresponding to $A$.
Following \cite{Mat}, see also \cite{FN}, multiline KP soliton solutions to the KP equation are defined as
\begin{equation}\label{eq:KPsol}
u( \vec t ) = 2\partial_{x}^2 \log(\tau ( {\vec t})),
\end{equation}
where 
\begin{equation}\label{eq:tau}
\tau (\vec t) = Wr (f^{(1)},\dots, f^{(k)}) \equiv 
\mbox{det } \left|
\begin{array}{ccc}
f^{(1)} & \cdots & f^{(k)}\\
\partial_x f^{(1)} & \cdots & \partial_x f^{(k)}\\
\vdots & \ddots & \vdots\\
\partial_x^{n-1} f^{(1)} & \cdots &\partial_x^{n-1} f^{(k)}\\
\end{array}
\right|
 = \sum\limits_{I} \Delta_I (A)\prod_{\mycom{i_1<i_2}{ i_1,i_2 \in I}} (\kappa_{i_2}-\kappa_{i_1} ) e^{ \sum\limits_{i\in I} \theta_i }
\end{equation}
is the Wronskian of $k$ linear independent solutions to the heat hierarchy
$\partial_{t_l} f = \partial_x^l f$, $l=2,3,\dots$, of the form
$f^{(i)}(\vec t) = \sum_{j=1}^n A^i_j e^{\theta_j}$, $i\in [k]$.
In (\ref{eq:tau}), the sum is over all $k$--element ordered subsets $I$ in $[n]$, {\it i.e.} $I=\{ 1\le i_1<i_2 < \cdots < i_k \le n\}$ and $\Delta_I (A)$ are the maximal minors of the matrix $A$. Since we obtain the same KP solution by linearly recombining the heat hierarchy solutions, $u(\vec t)$ is associated to the equivalence class $[A]$ of $A$, which is a point in the Grassmannian $Gr(k,n)$.

$u(x,y,t, \vec 0)$ is regular and bounded for all real $x,y,t$ if and only if $\Delta_I (A) \ge 0$, for all $I$ \cite{KW2}. In such case, let $Mat^{\mbox{\tiny TNN}}_{k,n}$ and $GL_k^+$, respectively denote the set of real $k\times n$ matrices of maximal rank $k$ with non--negative maximal minors $\Delta_I (A)$, and the group of $k\times k$ matrices with positive determinants.
Since left multiplication by elements in $GL_k^+$
preserves $u({\vec t})$ in (\ref{eq:KPsol}), we conclude that the soliton data $[A]$ is a point in the totally non--negative Grassmannian \cite{Pos}
$\Grkn = GL_k^+ \backslash Mat^{\mbox{\tiny TNN}}_{k,n}$.

Any given soliton solution is associated to an infinite set of soliton data $({\mathcal K}, [A])$, but there exists a unique minimal pair $(k,n)$, such that the soliton solution can be realized with $n$ phases $\kappa_1<\cdots<\kappa_n$ and $[A]\in \Grkn$, but not with $n-1$ phases and $[A^\prime]\in Gr^{\mbox{\tiny TNN}} (k^{\prime},n^{\prime})$, where $(k^{\prime}, n^{\prime})$ is either $(k, n-1)$ or $(k-1, n-1)$.

\begin{definition}\label{def:regsol}{\bf Regular and irreducible soliton data} \cite{CK2}.
We call $({\mathcal K}, [A])$ regular soliton data if ${\mathcal K} = \{ \kappa_1 < \cdots < \kappa_n \}$ and $[A]\in \Grkn$.
We call the regular soliton data $({\mathcal K}, [A])$ irreducible  if $[A]$ is a point in the irreducible part of the real Grassmannian, {\sl i.e.} if the reduced row echelon matrix $A$ has the following properties:
\begin{enumerate}
\item\label{it:col} Each column of $A$ contains at least a non--zero element;
\item\label{it:row} Each row of $A$ contains at least one nonzero element in addition to the pivot.
\end{enumerate}
If either (\ref{it:col}) or (\ref{it:row}) doesn't occur, we call the soliton data $({\mathcal K}, [A])$ reducible.
\end{definition}

\begin{remark}\label{rem:irred}\textbf{Reducible soliton data} \cite{CK2}.
If (\ref{it:col}) in  Definition \ref{def:regsol} is violated for column $l$, then the phase $\kappa_l$ does not appear in the solution (\ref{eq:KPsol}). Then, one may remove 
such phase from ${\mathcal K}$, remove the zero column from $A$ (see also Remark \ref{rem:redLe}) and realize the soliton in $Gr^{\mbox{\tiny TNN}} (k, n-1)$.

If (\ref{it:row}) in Definition \ref{def:regsol} is violated for the row $l$ corresponding to the pivot index $i_l$, then the heat hierarchy solution $f^{(l)}(\vec t)$ 
contains only the phase $\kappa_{i_l}$, and such phase is missing in all other heat hierarchy solutions associated to RREF
(reduced row echelon form) matrix. $f^{(l)}(\vec t)$ is factored out in (\ref{eq:tau}), and again, $\kappa_{i_l}$ is missing in (\ref{eq:KPsol}). So one may eliminate such phase 
from ${\mathcal K}$, remove the corresponding row and pivot column from $A$, change all signs in the new matrix to the right of the removed column and above the removed 
row and realize the soliton in $Gr^{\mbox{\tiny TNN}} (k-1, n-1)$ (see also Remark \ref{rem:redLe}).
\end{remark}
For generic choices of the phases ${\mathcal K}$, the combinatorial classification of the irreducible part $\Grkn$ rules the classification of the asymptotic properties of multi--soliton solutions both in the $(x,y)$ plane at fixed time $t$ and in the tropical limit ($t\to \pm \infty)$ (see \cite{BC,BPPP,CK2,CK,DMH,Kod1,Kod2,KW1,KW2,Z} and references therein).

\smallskip

The following spectral data are associated to each soliton data $({\mathcal K},[A])$, $[A]\in \Grkn$: an irreducible rational curve, which we denote $\Gamma_0$, a marked point $P_0\in \Gamma_0$, a degree $k$ real divisor, which we call Sato divisor, and a KP wave function meromorphic on $\Gamma_0\backslash \{ P_0\}$, which we call the Sato KP wave function.  
The unnormalized Sato wave function can be obtained from the dressing (inverse gauge) transformation \cite{S} of the vacuum (zero--potential) eigenfunction $\displaystyle \phi^{(0)} (\zeta, \vec t) =\exp ( \theta(\zeta, {\vec t}))$, which solves
\begin{equation}
\label{eq:vacuum_eq}
\partial_x \phi^{(0)} (\zeta, \vec t)=\zeta \phi^{(0)} (\zeta, \vec t), \quad\quad
\partial_{t_l}\phi^{(0)} (\zeta, \vec t) = \zeta^l \phi^{(0)} (\zeta, \vec t),\quad l\ge 2.
\end{equation}
The operator $
W = 1 -{\mathfrak w}_1({\vec t})\partial_x^{-1} -\cdots - {\mathfrak w}_k({\vec t})\partial_x^{-k}$,
where ${\mathfrak w}_1({\vec t}),\dots,{\mathfrak w}_k({\vec t})$ are the solutions to the following linear system of
equations
$\partial_x^k f^{(i)} = {\mathfrak w}_1 \partial_x^{k-1} f^{(i)}+\cdots + {\mathfrak w}_k f^{(i)}$,$i\in [k]$,
is the dressing ({\it i.e.} gauge) operator for the soliton data $({\mathcal K},[A])$. Indeed $W$ satisfies Sato equations $\partial_{t_l} W = B_l W - W \partial_x^l$, $l\ge 1$,
with $B_l = (W \partial_x^l W^{-1} )_+$ (the symbol $(H)_+$ denotes the differential part of the operator $H$). Therefore
$L= W \partial_x W^{-1} = \partial_x + \frac{u(\vec t)}{2}\partial_x^{-1} +\cdots$,  $u(\vec t) = 2\partial_x {\mathfrak w}_1 (\vec t)$ and $\psi^{(0)} (\zeta; \vec t)= W\phi^{(0)} (\zeta; \vec t)$ are,
respectively, the KP-Lax operator, the KP--potential (KP solution) and the KP-eigenfunction, {\sl i.e.}
$L \psi^{(0)} (\zeta; \vec t) =\zeta \psi^{(0)} (\zeta; \vec t)$, $\partial_{t_l}\psi^{(0)} (\zeta; \vec t)= B_l \psi^{(0)} (\zeta; \vec t)$, forall $l\ge 2$.

The Darboux dressing operator ${\mathfrak D}$ is defined as 
\begin{equation}\label{eq:D}
{\mathfrak D}\equiv W \partial_x^k = \partial_x^k - {\mathfrak w}_1 (\vec t) \partial_x^{k-1} -\cdots - {\mathfrak w}_k(\vec t)
\end{equation}
and the KP-eigenfunction  may be also represented by
\begin{equation}\label{eq:Satowf}
{\mathfrak D}\phi^{(0)} (\zeta; \vec t)  = W\partial_x^k \phi^{(0)} (\zeta; \vec t)
= \left(\zeta^k -{\mathfrak w}_1 (\vec t)\zeta^{k-1} -\cdots - {\mathfrak w}_k(\vec t)\right)
\phi^{(0)} (\zeta; \vec t) = \zeta^k \psi^{(0)} (\zeta; \vec t).
\end{equation}

\begin{definition}\label{def:Satodiv}{\bf Sato divisor}
Let the regular soliton data be $({\mathcal K}, [A])$, ${\mathcal K} = \{ \kappa_1 < \cdots < \kappa_n \}$, $[A]\in \Grkn$. We call Sato divisor at time $\vec t_0$, $\DS (\vec t_0)$, the set of the roots of the characteristic equation associated to the Dressing transformation
\begin{equation}\label{eq:Satodiv}
\DS (\vec t_0) = \{ \gs_j (\vec t_0), \; j\in [k]\, : \;\; (\gs_j(\vec t_0))^k - {\mathfrak w}_1 (\vec t)(\gs_j(\vec t_0))^{k-1}-\cdots  -  {\mathfrak w}_k(\vec t_0) = 0 	\,	\}. 
\end{equation}
\end{definition}

In \cite{Mal} it is proven the following proposition
\begin{proposition} \textbf{The Sato divisor}\label{prop:malanyuk} \cite{Mal}.
Let the regular soliton data be $({\mathcal K}, [A])$, ${\mathcal K} = \{ \kappa_1 < \cdots < \kappa_n \}$, $[A]\in \Grkn$. Then for all real $\vec t_0$ the Sato divisor  $\DS (\vec t_0)$ is real and satisfies $\gs_{j}(\vec t)\in [\kappa_1,\kappa_n]$, $j\in [k]$. 
Moreover for almost all $\vec t_0$ the Sato divisor points are distinct.
\end{proposition}

\begin{remark}\label{rem:reddiv}{\bf Sato divisor for reducible regular soliton data}
In the case of reducible regular soliton data $({\mathcal K}, [A])$, ${\mathcal K} = \{ \kappa_1<\cdots <\kappa_n\}$, $[A]\in \Grkn$ (see Remark \ref{rem:irred}), we use the reduced Sato divisor $\DS^{\prime}(\vec t_0)$ of the corresponding maximally reduced positroid cell $Gr^{\mbox{\tiny TNN}}(k^{\prime},n^{\prime})$.

More precisely, if the representative RREF matrix $A$ in $\Grkn$, contains a zero column in position $l$, then $k^{\prime} = k$ and $\DS(\vec t_0)= \DS^{\prime} (\vec t_0)$ for any $\vec t_0$, since the reducible and the reduced Darboux transformations coincide ${\mathfrak D}^{(k)} = {\mathfrak D}^{(k^{\prime})}$.

Instead, if for some $r\in [k]$ and $i_r\in [r,n]$, the $r$--th row of the RREF matrix $A$ contains only the pivot element: $A^r_{j} = \delta_{j,i_r}$, then, for all $\vec t_0$,
$\kappa_{i_r} \in \DS (\vec t_0)$, $k^{\prime} =k-1$ and $\DS^{\prime} (\vec t_0)= \DS (\vec t_0)\backslash \{ \kappa_{i_r}\}$. Indeed, the characteristic polynomial associated to the Darboux differential operator ${\mathfrak D}$ satisfies
$\zeta^k -\zeta^{k-1} {\mathfrak w}_1 (\vec t)-\cdots - {\mathfrak w}_k(\vec t) = (\zeta-\kappa_{i_r}) \left(  \zeta^{k-1} - {\mathfrak w}_1^{\prime} (\vec t)\zeta^{k-2}-\cdots - {\mathfrak w}_{k-1}^{\prime} (\vec t) \right)$. Then
${\mathfrak D}^{\prime} =W^{\prime} \partial_x^{k-1}= \partial_x^{k-1} -{\mathfrak w}_1^{\prime} (\vec t)\partial_x^{k-2} -\cdots - {\mathfrak w}_{k-1}^{\prime} (\vec t)$
is the Darboux transformation associated to the reduced soliton data $({\mathcal K}^{\prime}, [A^{\prime}])$, with ${\mathcal K}^{\prime} = {\mathcal K} \backslash \{ \kappa_{i_r} \}$, $[A^{\prime}]\in Gr^{\mbox{\tiny TNN}} (k-1,n-1)$ and $A^{\prime}$ related to $A$ as in Remark \ref{rem:irred}.
\end{remark}

\begin{definition}\label{def:Sato_data}\textbf{Sato algebraic--geometric data} Let $({\mathcal K}, [A])$ be given regular soliton data with $[A]$ belonging to a $d$ dimensional positroid cell in $\Grkn$. Let $\vec t_0$ such that the Sato divisor consists of $k$ simple poles. Let $\Gamma_0$ be a copy of $\mathbb{CP}^1$ with marked points $P_0$, local coordinate $\zeta$ such that $\zeta^{-1} (P_0)=0$ and $\zeta(\kappa_1)<\zeta(\kappa_2)<\cdots<\zeta(\kappa_n)$.

Then to the data $({\mathcal K}, [A], \Gamma_0\backslash \{ P_0	\} ,\vec t_0)$ we associate the Sato divisor $\DS=\DS(\vec t_0)$
as in Definition (\ref{def:Satodiv}) and the normalized Sato wave function
\begin{equation}\label{eq:SatoDN}
{\hat \psi } (P, \vec t) = \frac{{\mathfrak D}\phi^{(0)} (P; \vec t)}{{\mathfrak D}\phi^{(0)} (P; \vec t_0)} = \frac{\psi^{(0)} (P; \vec t)}{\psi^{(0)} (P; \vec t_0)}, \quad\quad \forall P\in \Gamma_0\backslash \{ P_0\},
\end{equation}
with ${\mathfrak D}\phi^{(0)} (\zeta; \vec t)$ as in (\ref{eq:Satowf}). 

By definition $({\hat \psi}_0 (P,\vec t)) + \DS(\vec t_0) \ge 0$, for all $\vec t$.
\end{definition}

\begin{remark}\label{rem:fundam}\textbf{Incompleteness of Sato algebraic--geometric data}
Let $1\le k<n$ and let $\vec t_0$ be fixed. Given the phases $\kappa_1<\cdots <\kappa_n$ and the spectral data $( \Gamma_0\backslash \{ P_0	\} , \DS) $, where $\DS=\DS(\vec t_0) $ is a $k$ point divisor satisfying Proposition \ref{prop:malanyuk}, it is, in general, impossible to identify uniquely the point $[A]\in \Grkn$ corresponding to such spectral data.
Indeed, if we assume that the soliton data belong to an irreducible positroid cell
of dimension $d$, then $\max\{k, n-k\} \le d \le k(n-k)$. Otherwise, an analogous inequality holds for the reduced Sato divisor.
\end{remark}

\subsection{Finite-gap KP solutions and their multi-line soliton limits}

Soliton KP solutions can be obtained from the finite-gap ones by proper degenerations of the spectral curve \cite{Kr2,DKN}. 

The spectral data for periodic and quasiperiodic solutions of the KP equation (\ref{eq:KP}) in the finite-gap approach \cite{Kr1,Kr2} are: a finite genus $g$ compact Riemann surface $\Gamma$ with a marked point $P_0$, a local parameter $1/\zeta$ near $P_0$ and a non-special 
divisor $\mathcal D=\gamma_1+\ldots+\gamma_g$ of degree $g$ in $\Gamma$.
The Baker-Akhiezer function $\hat\psi (P, \vec t)$, $P\in\Gamma$, is defined by the following analytic properties:
\begin{enumerate}
\item For any fixed $\vec t$ the function $\hat\psi (P, \vec t)$ is meromorphic in $P$ on $\Gamma\backslash P_0$.
\item On  $\Gamma\backslash P_0$ the function $\hat\psi (P, \vec t)$ is regular outside the divisor points $\gamma_j$ and has at most first order poles 
at the divisor points. Equivalently, if we consider the line bundle $\mathcal L(\mathcal D)$  associated to $\mathcal D$, then
for each fixed $\vec t$ the function $\hat\psi (P, \vec t)$ is a holomorphic section of $\mathcal L(\mathcal D)$ outside $P_0$.
\item $\hat\psi (P, \vec t)$ has an essential singularity at the point $P_0$ with the following asymptotic:
\[
{\hat \psi} (\zeta, \vec t) = e^{ \zeta x +\zeta^2 y +\zeta^3 t +\cdots} \left( 1 - \chi_1({\vec t})\zeta^{-1} - \cdots
-\chi_k({\vec t})\zeta^{-k}  - \cdots\right). 
\]
\end{enumerate}
For generic data these properties define an unique function, which is a common eigenfunction to all KP hierarchy auxiliary linear operators 
$-\partial_{t_j} + B_j$, where $B_j =(L^j)_+$, and the Lax operator is $L=\partial_x+\frac{u(\vec t)}{2}\partial_x^{-1}+ u_2(\vec t)\partial_x^{-2}+\ldots$. All these operators commute and the potential $u(\vec t)$ satisfies the KP hierarchy. In particular, the KP equation arises in the 
Dryuma-Zakharov-Shabat commutation representation \cite{Druma}, \cite{ZS} as the compatibility for the second and the third operator: $[ -\partial_y + B_2, -\partial_t +B_3] =0$,
with $B_2 \equiv (L^2)_+ = \partial_x^2 + u$, $B_3 = (L^3)_+ = \partial_x^3 +\frac{3}{4} (u\partial_x +\partial_x u) + \tilde u$
and $\partial_x\tilde u =\frac{3}{4} \partial_y u$.

The Its-Matveev formula represents the KP hierarchy solution $u(\vec t)$ in terms of the Riemann theta-functions associated with $\Gamma$ (see, for example, \cite{Dub}). 
After fixing a canonical basis of cycles $a_1,\dots,a_g,b_1,\dots,b_g$ and a basis of normalized holomorphic differentials $\omega_1,\dots,\omega_g$ on $\Gamma$ such that
$\oint_{a_j} \omega_l =2\pi i \delta_{jl}$, $\oint_{b_j} \omega_l = B_{lj}$, $j,l \in [g]$,
the KP solution takes the form $u(\vec t) = 2\partial_x^2 \log \theta (\sum_j t_j U^{(j)} + z_0)+c_1,$
where $\theta$ is the Riemann theta function and $U^{(j)}$ are the vectors of the $b$--periods of the following normalized meromorphic differentials, holomorphic on $\Gamma\backslash \{ P_0\}$ and with principal parts $\hat \omega^{(j)} = d (\zeta^j ) +O(1)$, at $P_0$ (see \cite{Kr1,DN}).

The real regular solutions are the most relevant in physical applications. In \cite{DN} the necessary and sufficient conditions on spectral data to generate real regular KP hierarchy solutions for all real $\vec t$ were established, under the assumption that $\Gamma$ is smooth and has genus $g$: 
\begin{enumerate}
\item $\Gamma$ possesses an antiholomorphic involution ${\sigma}:\Gamma\rightarrow\Gamma$, ${\sigma}^2=\mbox{id}$, which has the maximal possible number of fixed components (real ovals). This number 
is equal to $g+1$, therefore $(\Gamma,\sigma)$ is an $\mathtt M$-curve.
\item $P_0$ lies in one of the ovals, and each other oval contains exactly one divisor points. The oval containing $P_0$ is called ``infinite'' and all 
other ovals are called ``finite''.
\end{enumerate}

The set of real ovals divides $\Gamma$ into two connected components. Each of these components is homeomorphic to a sphere with $g+1$ holes. 
In Figure \ref{fig:regM_g2}(left) we show an example for $g=2$.

\begin{figure}
  \centering
  {\includegraphics[scale=0.15,angle=90,width=0.44\textwidth]{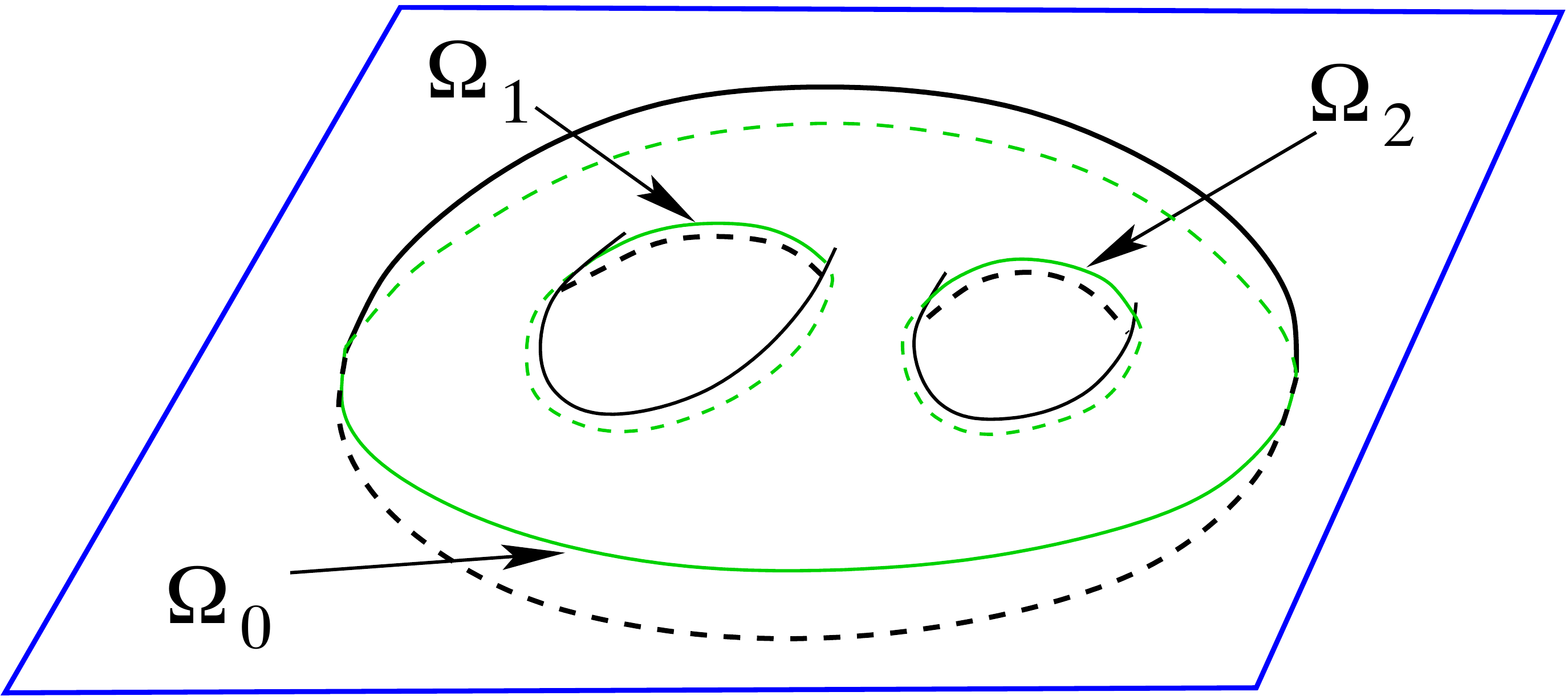}}
  \hfill
  {\includegraphics[scale=0.15,angle=90,width=0.44\textwidth]{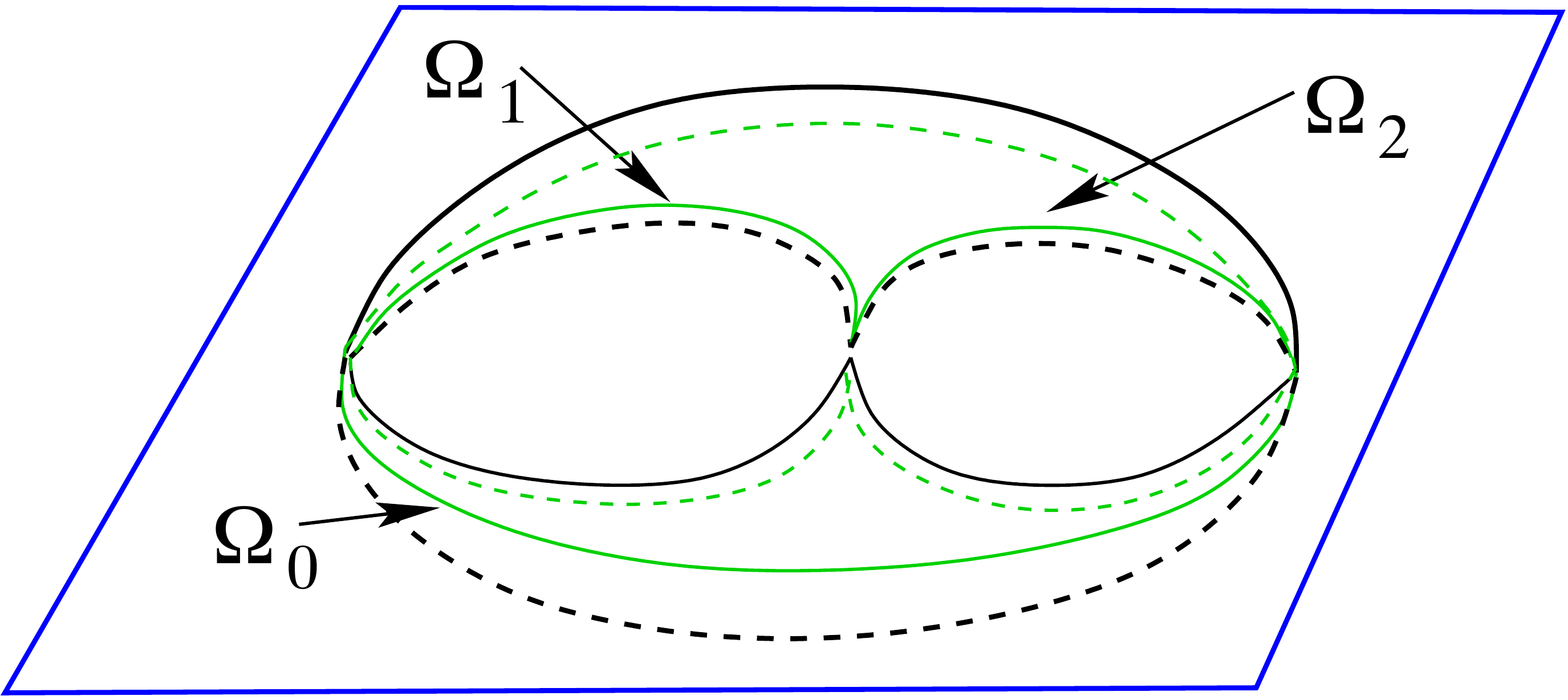}}
	\vspace{-1. truecm}
\caption{\footnotesize{\sl Left: a genus 2 regular $\mathtt M$-curve  with 3 real ovals invariant w.r.t. the 
involution {$\sigma$} (orthogonal reflection with respect to the horizontal plane). Right: its degeneration is a reducible $\mathtt M$-curve still possessing 3 real ovals.}}\label{fig:regM_g2}	
\end{figure}

The sufficient condition of the Theorem in \cite{DN} still holds true if the spectral curve $\Gamma$ degenerates in such a way that the divisor remains in the finite ovals 
at a finite distance from the essential singularity \cite{DN}. Of course, this condition is not necessary for degenerate curves, but the properties of the Sato divisor established in \cite{Mal} are compatible with such an ansatz. Moreover, in \cite{Kr3}, it has been proven that the algebraic--geometric approach goes through also for degenerate finite--gap solutions on \textbf{reducible} curves. Such inverse spectral problem is ill--posed, since there is not a unique reducible curve associated to the given soliton data. Finally there is also no a priori reason why, given one such reducible curve, the divisor on it should satisfy any reality condition. 

In \cite{AG1}, we have proven that the multiline soliton solutions corresponding to points in $Gr^{\mbox{\tiny TP}} (k,n)$ may indeed be obtained as limits of real regular finite-gap solutions on smooth $\mathtt M$--curves: to any soliton datum in $Gr^{\mbox{\tiny TP}} (k,n)$ and any $\xi\gg 1$, we have associated a curve $\Gamma_{\xi}$, which is the rational degeneration of 
a smooth $\mathtt M$--curve of minimal genus $k(n-k)$ and a degree $k(n-k)$ divisor satisfying the reality conditions of Dubrovin and Natanzon's theorem.
In Figure~\ref{fig:regM_g2}(right) we show the rational degeneration of the genus $g=2$ curve associated to soliton data in 
$Gr^{\mbox{\tiny TP}}(1,3)$ and $Gr^{\mbox{\tiny TP}}(2,3)$ in \cite{A,AG1}.

The main objective of this paper is therefore twofold: \textbf{provide a canonical construction of a reducible rational $\mathtt M$--curve of minimal genus $g=d$ for any fixed positroid cell in $\Grkn$ and show that real and regular divisors on such curve provide a parametrization of the cell}. We therefore give the following definition:

\begin{definition} 
\label{def:rrss}
\textbf{Real regular algebraic-geometrical data associated with a given soliton solution.}
Assume that we have fixed soliton data $({\mathcal K},[A])$, where ${\mathcal K}$ is a 
collection of real phases $\kappa_1<\kappa_2<\ldots <\kappa_n$, $[A]\in \Grkn$. Let $d$ be the dimension of the positroid cell to which $[A]$ belongs. 

Assume that we have a reducible connected curve $\Gamma$ with a marked point $P_0$, a local parameter $1/\zeta$ near $P_0$. In addition, assume that the curve $\Gamma$  may be obtained from a rational degeneration of a smooth ${\mathtt M}$-curve of genus $g$, with $g\ge d$, and that the antiholomorphic involution preserves the maximum number of the ovals in the limit, so that $\Gamma$ possesses $g+1$ real ovals. 

Assume that ${\mathcal D}$ is a degree $g$ non-special divisor on $\Gamma\backslash P_0$, and that $\hat\psi$ is the normalized Baker-Ahkiezer function associated to such data, i.e. for any $\vec t$ its pole divisor is contained in ${\mathcal D}$: $({\hat \psi} (P , \vec t))+{\mathcal D} \ge 0$ on $\Gamma\backslash P_0$, where $(f)$ denotes the divisor of $f$.

We say that  \textbf{the algebraic-geometrical data $\Gamma$, ${\mathcal D}$ are associated to the soliton data $({\mathcal K},[A])$,} if the irreducible component $\Gamma_0$ of $\Gamma$ containing $P_0$ is $\mathbb{CP}^1$, and the restriction of $\hat\psi$ to $\Gamma_0$ coincides with Sato normalized dressed wave function for the soliton data $({\mathcal K},[A])$. In particular, for such data the restriction of ${\mathcal D}$ to $\Gamma_0$ coincides with the Sato divisor.
 
We say that the \textbf{divisor 
${\mathcal D}$ satisfies the reality and regularity conditions} if $P_0$ belongs to one of the fixed ovals and the boundary of each other finite oval contains exactly one divisor point. 
\end{definition}

We remark that the simplicity and reality of the Sato divisor points proven in \cite{Mal} is compatible with the reality and regularity of the algebraic-geometrical data associated with a given soliton solution in the Definition above, provided that the reducible curve $\Gamma$ possesses $k$ distinct ovals containing the Sato divisor points.

\section{Algebraic-geometric approach for KP soliton data in $\Grkn$: the main construction}\label{sec:3}

Since $\Grkn$ is topologically the closure of $Gr^{\mbox{\tiny TP}}(k,n)$, one can try to extend \textbf{indirectly} the construction of \cite{AG1} to soliton data in $\Grkn$ considering the latter as the limit of a sequence of soliton data in $Gr^{\mbox{\tiny TP}}(k,n)$. But this limiting procedure is very non-trivial, and it provides only an upper bound for the genus: $g\le k(n-k)$.

In the following we present a \textbf{direct} construction of algebraic geometric data associated to points in $\Grkn$ which is naturally related to the characterization of positroid cells in \cite{Pos} and provides optimal genus spectral curves. Moreover, the present construction unveils the relation of the algebraic construction in \cite{AG1} with Le-networks in $Gr^{\mbox{\tiny TP}}(k,n)$. The starting point are the algebraic geometric data associated to $({\mathcal K}, [A])$ via Sato dressing (see Definition \ref{def:Sato_data}):
\begin{enumerate}
\item A rational curve $\Gamma_0$ equipped with a finite number of marked points: the ordered real phases $\kappa_1<\cdots<\kappa_n$, and the essential singularity 
$P_0$ of the wave function;
\item The Sato divisor $\DS(\vec t_0)$ for the soliton data defined in Definition \ref{def:Satodiv};
\item The normalized wave function $\hat \psi(P,\vec t) $ on $\Gamma_0 \backslash \{ P_0\}$ defined in (\ref{eq:SatoDN}). 
\end{enumerate}
As pointed out in Section \ref{sec:Sato}, in general, the Sato divisor does not parametrize the whole positroid cell to which $[A]$ belongs to. Therefore, in general, we cannot reconstruct the soliton data $[A]$ just from the Sato divisor at $\vec t_0$.
 
Below, to any regular soliton data $({\mathcal K}, [A])$, we associate a well--defined curve $\Gamma$ containing $\Gamma_0$ as a connected component, and a unique KP divisor on it using the Le--network  ${\mathcal N}$ representing $[A]$. In particular, we show that $\Gamma$ is reducible and a rational degeneration of a $\mathtt M$-curve having the minimal possible genus, $g=d=\dim \S$, under genericity assumption on the soliton data. Finally, the set of poles of $\hat\psi$ exactly coincides with $\DKP$ for generic $\vec t$.  

\smallskip

\textbf{Main construction} {\sl Assume we are given a real regular bounded multiline KP soliton solution generated by the following soliton data:
\begin{enumerate}
\item A set of $n$ real ordered phases ${\mathcal K} =\{ \kappa_1<\kappa_2<\dots<\kappa_n\}$;
\item A point $[A]\in \S \subset \Grkn$, where $\S$ is a positroid stratum of dimension $d$. 
\end{enumerate}
We represent $[A]$ with its canonically oriented bipartite trivalent Le--network ${\mathcal N}$. We recall that ${\mathcal N}$ provides a representation of the points of the cell depending exactly by $d$ parameters. Let us also denote ${\mathcal G}$ the Le--graph representing $\S$.
Then, we associate the following algebraic-geometric objects to the soliton data $({\mathcal K}, [A])$:
\begin{enumerate}
\item A reducible $\mathtt M$--curve $\Gamma=\Gamma({\mathcal G})$ with $g+1$ ovals which is the rational degeneration of a smooth $\mathtt M$--curve of genus $g = d$. In our approach, $\Gamma_0$ is one of the irreducible components of $\Gamma$. The marked point $P_0$ belongs to the intersection of $\Gamma_0$ with an oval (infinite oval); 
\item An unique real and regular degree $g$ non--special KP divisor $\DKP= \DKP({\mathcal K}, [A])\subset \Gamma\backslash \{P_0\}$ such that any finite oval contains exactly one divisor point and $\DKP \cap \Gamma_0$ coincides with Sato divisor $\DS(\vec t_0)$ for some initial time $\vec t_0$;
\item An unique KP wave--function $\hat \psi(P, \vec t)$ on $\Gamma \backslash \{ P_0 \}$ as in Definition \ref{def:rrss} such that
\begin{enumerate}
\item Its restriction to $\Gamma_0\backslash \{P_0\}$ coincides with the normalized Sato wave function (\ref{eq:Satowf});
\item Its pole divisor has degree $\mathfrak d \le g$ and is contained in $\DKP$.
\end{enumerate}
\end{enumerate}
}

\begin{remark}
Here and in the following, when we refer to the Sato divisor for reducible real and regular soliton data, we mean the reduced Sato divisor defined in Remark \ref{rem:reddiv}. In particular, the KP divisor $\DKP$ restricted to $\Gamma_0$ is the reduced Sato divisor.
\end{remark}

\begin{remark}
In \cite{AG2} we extend the main construction using any oriented bipartite trivalent network in the disk representing $[A]$ in Postnikov equivalence class \cite{Pos}. 
\end{remark}

\subsection{The reducible rational curve $\Gamma$}
\label{sec:gamma}

Given the oriented graph ${\mathcal G}$ representing a given positroid cell $\S\subset \Grkn$, the curve $\Gamma=\Gamma({\mathcal G})$ is obtained gluing a finite number of copies of $\mathbb{CP}^1$, corresponding to the internal vertices in ${\mathcal G}$, and one copy of $\mathbb{CP}^1=\Gamma_0$, corresponding to the boundary of the disk. We glue these components at pairs of points corresponding to its edges. We also fix a local affine coordinate $\zeta$ on each component (see Definition~\ref{def:loccoor}), therefore we have complex conjugation $\zeta\rightarrow\bar\zeta$ at each component. The points with real $\zeta$ form the real part of the given component. By construction (see Definition \ref{def:loccoor}), the coordinates at each pair of glued points $P$, $Q$, are real. We then topologically represent the real part of $\Gamma$ as a union of circles (ovals), where the latter correspond to the faces of ${\mathcal G}$.

\begin{figure}
  \centering
  {\includegraphics[width=0.46\textwidth]{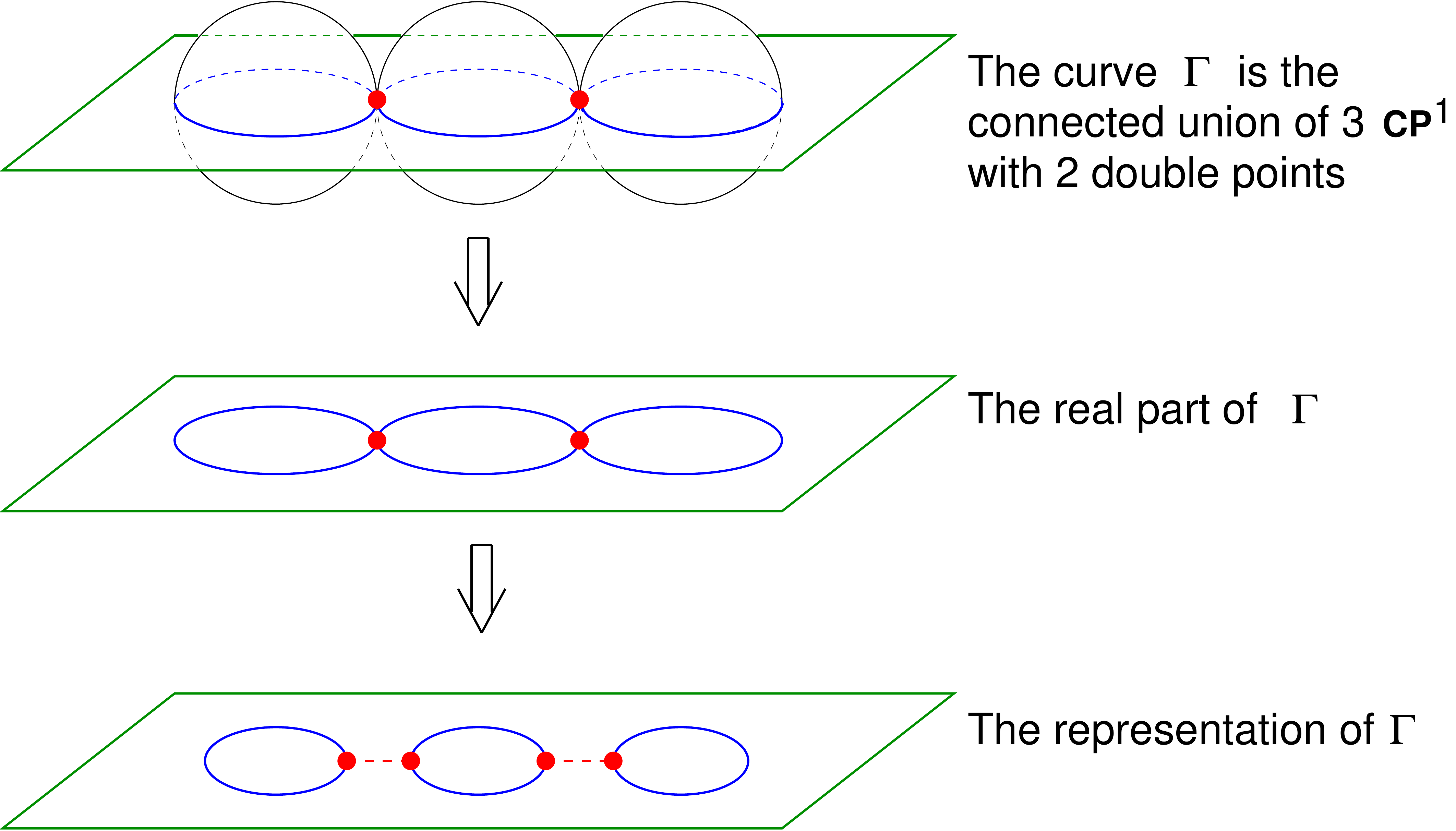}}
  \caption{\label{fig:curve_model}\footnotesize{\sl The model of a reducible rational curve $\Gamma$ with three components and one oval.}}
\end{figure}

We use the same representation for real rational curves as in \cite{AG1} (see Fig.~\ref{fig:curve_model}). We draw only the real part of each component and we represent it with a circle. Then we schematically represent the real part of $\Gamma$ by drawing these circles separately and connecting the glued points by dashed lines. The planarity of the Le--graph implies that $\Gamma$ is a reducible rational $\mathtt M$--curve. 

\begin{construction}\label{def:gamma}\textbf{The curve $\Gamma=\Gamma({\mathcal G})$.}
Let ${\mathcal K} =\{\kappa_1 < \cdots < \kappa_n\}$ and let $\S\subset \Grkn$ be the positroid cell corresponding to the realizable matroid ${\mathcal M}$. Let ${\mathcal G}$ be the planar connected acyclically oriented trivalent bipartite Le--graph in the disk of Definition \ref{def:can_Le} representing $\S$. 
Let $I=\{1 \le i_1 < \cdots < i_k\le n\}$  be the set of the pivot indexes ({\sl i.e } the lexicographically minimal base of ${\mathcal M}$) and, for any $r\in [k]$, let $N_r$ be the number of filled boxes in the $r$-th row of the corresponding Le--diagram (see (\ref{eq:Nij}) in Appendix \ref{app:TNN}). Finally, let $1\le j_1 < j_2 < \cdots < j_{N_r} \le n$ be the non--pivot indexes of the boxes $B_{i_r, j_s}$ in the Le--diagram with index $\chi^{i_r}_{j_s}=1$, $s\in [N_r]$, where notations are consistent with (\ref{eq:Nij}) and (\ref{eq:chiindex}) in Appendix \ref{app:TNN}.

\begin{table}[H]
\caption{The correspondence between the Le--graph ${\mathcal G}$  and the reducible rational curve $\Gamma$} 
\centering
\begin{tabular}{|c|c|}
\hline\hline
${\mathcal G}$ & $\Gamma$\\[0.5ex]
\hline
Boundary of disk & Copy of $\mathbb{CP}^1$ denoted $\Gamma_0$ \\
Boundary vertex $b_l$             & Marked point $\kappa_l$ on  $\Gamma_0$\\
Internal black vertex   $V^{\prime}_{ij}$  & Copy of $\mathbb{CP}^1$ denoted $\Sigma_{ij}$\\
Internal white vertex   $V_{ij}$           & Copy of $\mathbb{CP}^1$ denoted $\Gamma_{ij}$\\
Edge                     & Double point\\
Face                              & Oval\\ [1ex]
\hline
\end{tabular}
\label{table:LeG}
\end{table}

The curve
$\Gamma=\Gamma({\mathcal G})$ is associated to ${\mathcal G}$ according to 
Table \ref{table:LeG}, after reflecting the graph w.r.t. a line orthogonal to the one containing the boundary vertices 
(we reflect the graph to have the natural increasing order of the phases on $\Gamma_0\subset\Gamma$).
More precisely, $\Gamma$ is the connected union of
$2n+k+1$ copies of $\mathbb{ CP}^1$ denoted as $\Gamma_0$, $\Sigma_{i_rj_s}$, $\Gamma_{i_rj_s}$, $\Gamma_{i_r}$, for $r\in [k]$, $s\in [N_r]$
\[
\displaystyle \Gamma=\Gamma_0 \bigsqcup\limits_{r=1}^k \left(\Gamma_{i_r} \sqcup\Bigg(\bigsqcup\limits_{s=1}^{N_r} \Gamma_{i_rj_s} \sqcup\Sigma_{i_rj_s}\Bigg)\right) 
\]
according to the following rules (see also Figures \ref{fig:markedpoints} and \ref{fig:doublepoints}):   
\begin{enumerate}
\item $\Gamma_0$ is the copy of $\mathbb{CP}^1$ corresponding to the boundary of the disk. It has $n+1$ marked points: $P_0$ such that $\zeta^{-1} (P_0)=0$ and the points $\kappa_1<\cdots <\kappa_n$ corresponding to the boundary vertices $b_1,\dots, b_n$ on ${\mathcal G}$; 
\item A copy of $\mathbb{CP}^1$ corresponds to any internal vertex of ${\mathcal G}$. For any fixed $r\in [k]$ and $s\in [N_r]$ we denote $\Gamma_{i_rj_s}$ (resp. $\Sigma_{i_rj_s}$) the copy of $\mathbb{CP}^1$ corresponding to the white vertex $V_{i_rj_s}$ (resp. black vertex $V^{\prime}_{i_rj_s}$);
\item We denote $\Gamma_{i_r}$ the copy of $\mathbb{CP}^1$ corresponding to the internal white vertex $V_{i_r}$ joined by an edge to the source $b_{i_r}$, $r\in [k]$;

\begin{figure}
  \centering
  \includegraphics[width=0.46\textwidth]{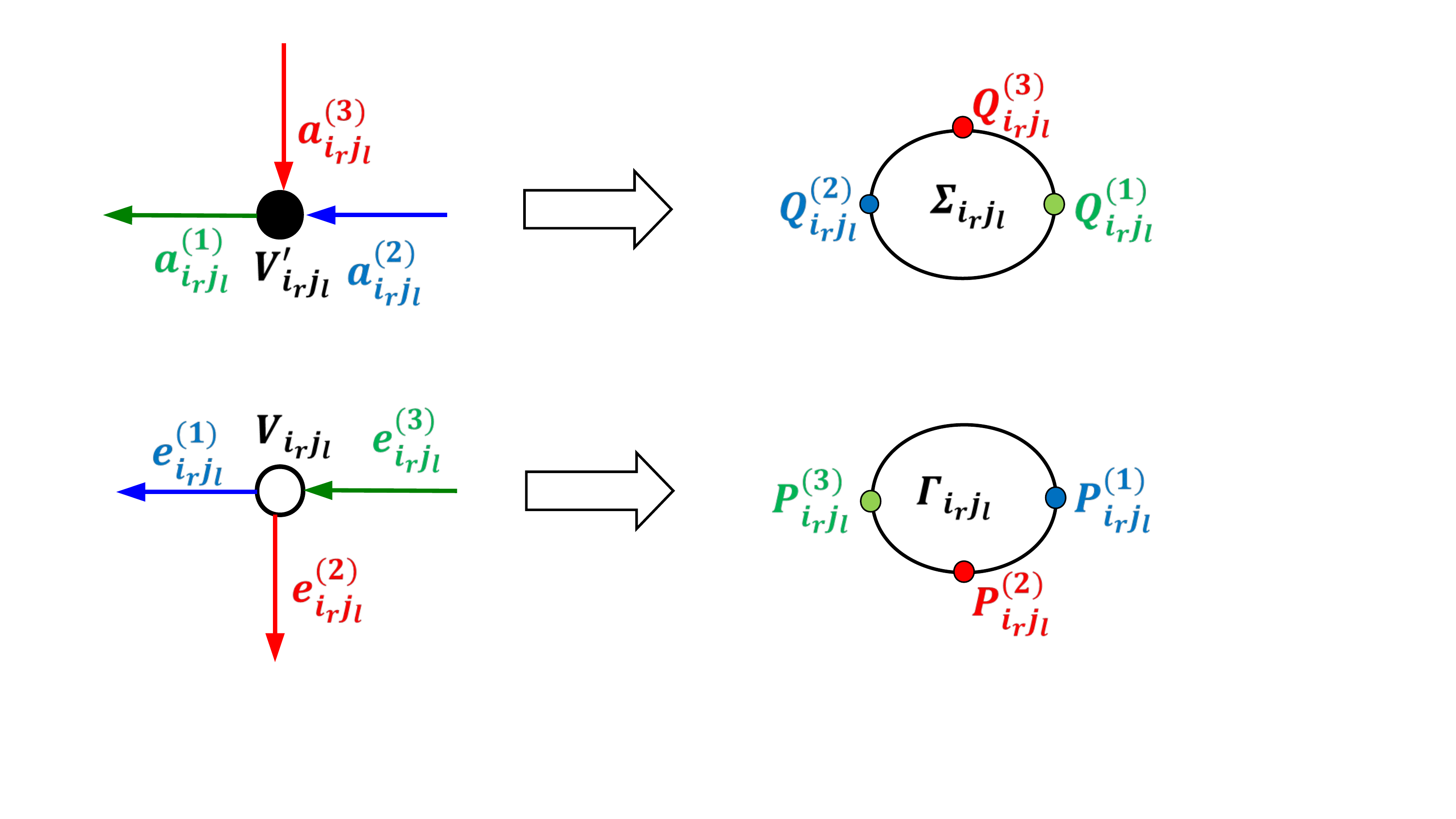}
	\hspace{.5 truecm}
	 \includegraphics[width=0.46\textwidth]{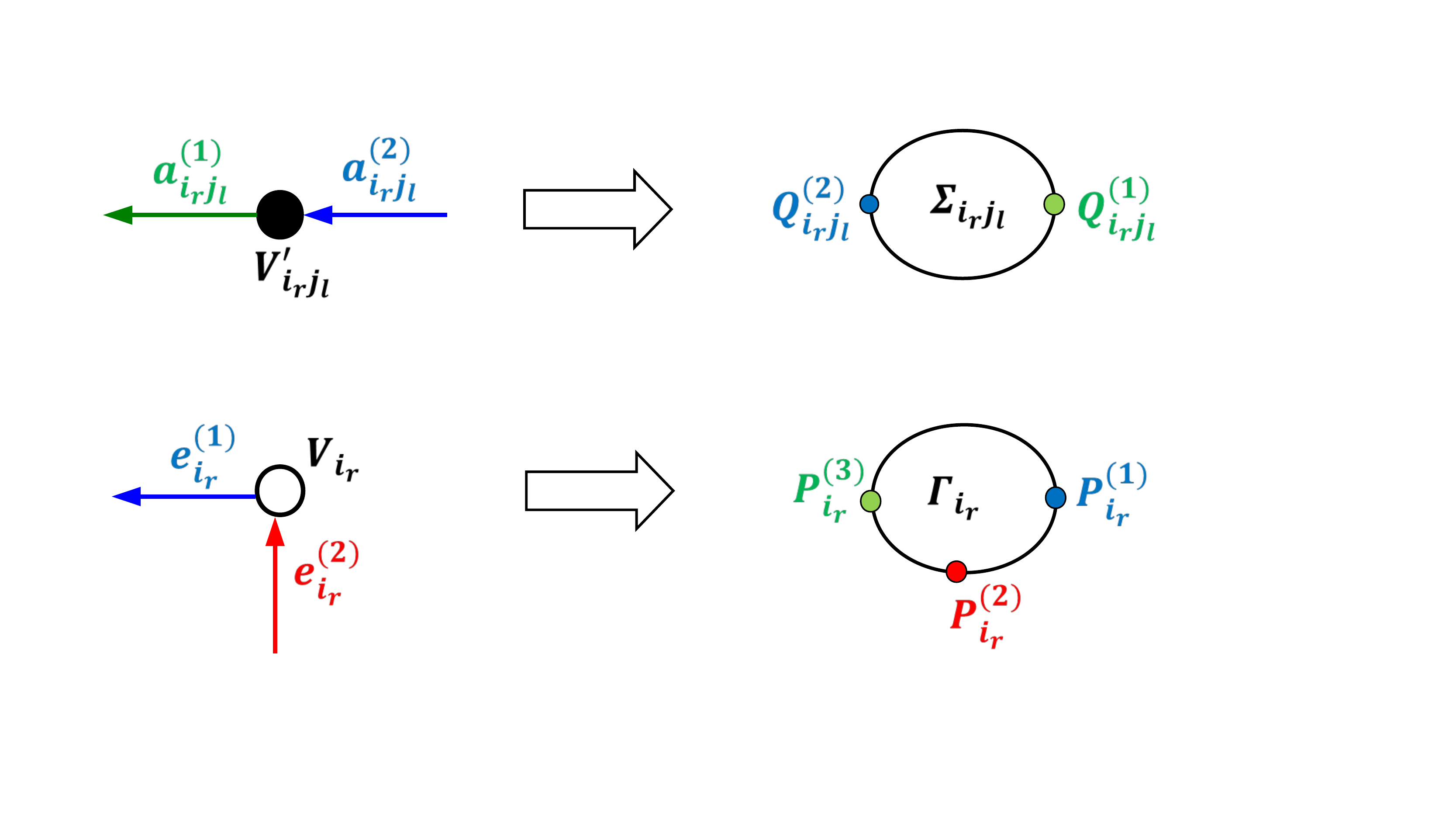}
  \vspace{-.7 truecm}
  \caption{\footnotesize{\sl The correspondence between marked points on copies $\Gamma_{i_rj_l}$ and $\Sigma_{i_rj_l}$ and edges of white and black vertices. The rule at the marked points corresponding to the edges of a bivalent white vertex at a boundary source $b_{i_r}$ is justified by the necessity of adding a third marked point (Darboux point) on $\Gamma_{i_r}$.}}
	\label{fig:markedpoints}
\end{figure}

\item On each copy of $\mathbb{CP}^1$ corresponding to an internal vertex $V$, we mark as many points as edges at $V$. 
We number the edges at $V$ anticlockwise in increasing order, so that, on the corresponding copy of $\mathbb{CP}^1$, the marked points are numbered clockwise because of the mirror rule (see Figure \ref{fig:markedpoints}). We use the following numbering rule:
\begin{enumerate}
\item The unique horizontal edge pointing inward at the white vertex $V_{i_r i_s}$ is numbered 3, for any $r\in [k]$, $s\in [N_r]$. Therefore $\Gamma_{i_rj_s}$, $r\in [k]$, $s\in [N_r-1]$, has 3 real ordered marked points which we denote $P_{i_rj_s}^{(1)},P_{i_rj_s}^{(2)},P_{i_rj_s}^{(3)}$ (see Figure \ref{fig:markedpoints}[bottom, left]) and $\Gamma_{i_r j_{N_r}}$ has two marked points $P_{i_rj_s}^{(2)}, P_{i_rj_s}^{(3)}$;
\item At each white vertex $V_{i_r}$ we have a horizontal edge marked $e^{(1)}_{i_r}$ and a vertical edge marked $e^{(2)}_{i_r}$
which correspond to the marked points $P^{(1)}_{i_r}, P^{(2)}_{i_r}\in \Gamma_{i_r}$.
 On each $\Gamma_{i_r}$, $r\in [k]$, we add an extra point, the Darboux point $P^{(3)}_{i_r}$, which we use to rule the position of the vacuum divisor;
\item The unique edge pointing outward at a black vertex $V^{\prime}_{i_r j_s}$, $r\in [k]$, $s\in [N_r-1]$, is always numbered $1$. We denote $Q_{i_rj_s}^{(m)}$, $m\in [3]$ (resp. $m\in [2]$) the marked points on $\Sigma_{irj_s}$ corresponding to  the trivalent (resp. bivalent) black vertex $V^{\prime}_{i_rj_s}$.
\end{enumerate}
\item We glue copies of $\mathbb{CP}^1$ in pairs at the marked points corresponding to the end points of the corresponding edge on ${\mathcal G}$ (see Figure \ref{fig:doublepoints}). More precisely:

\begin{figure}
  \centering
  {\includegraphics[width=0.43\textwidth]{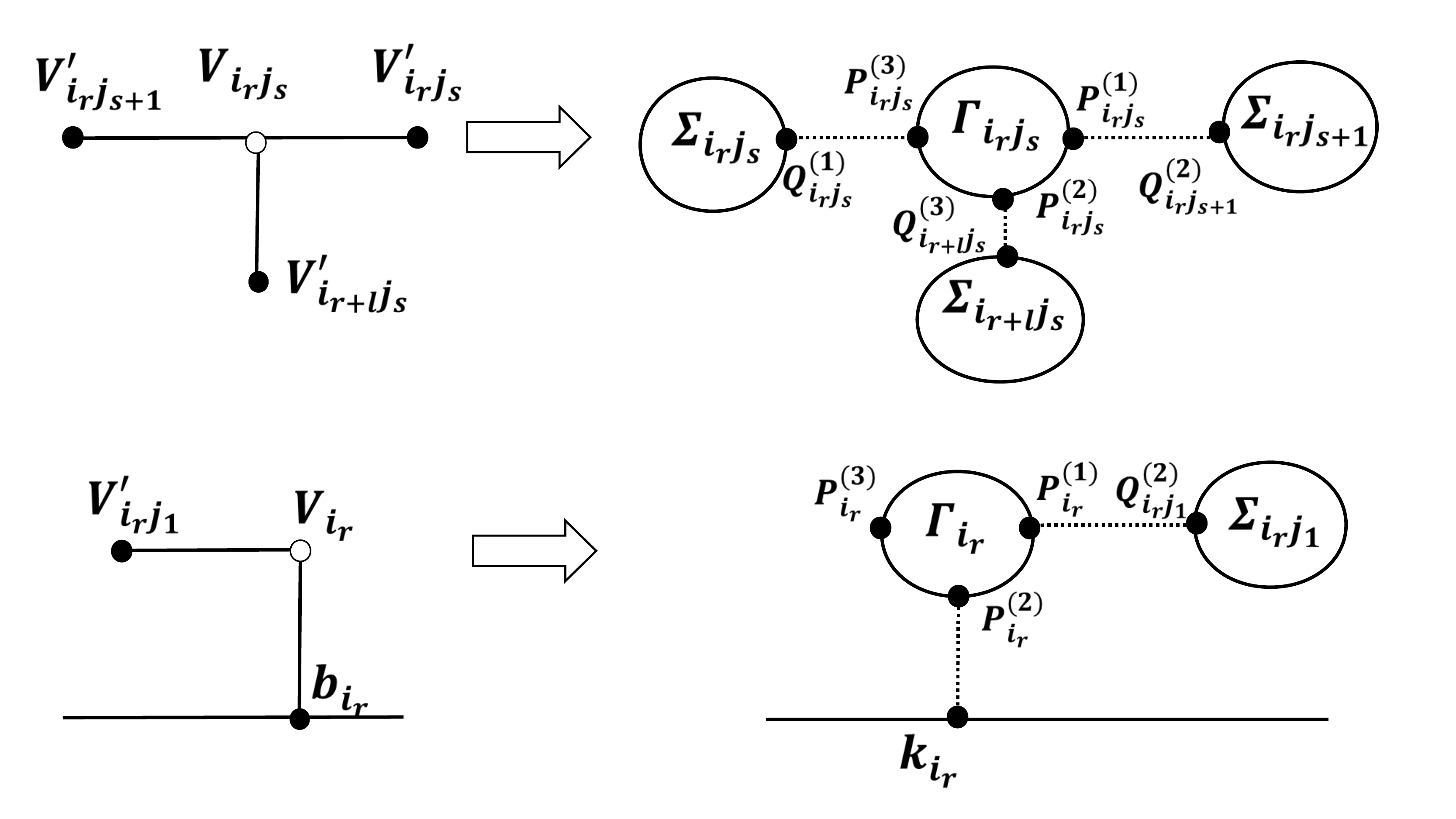}}
	\hspace{1.2 truecm}
	{\includegraphics[width=0.43\textwidth]{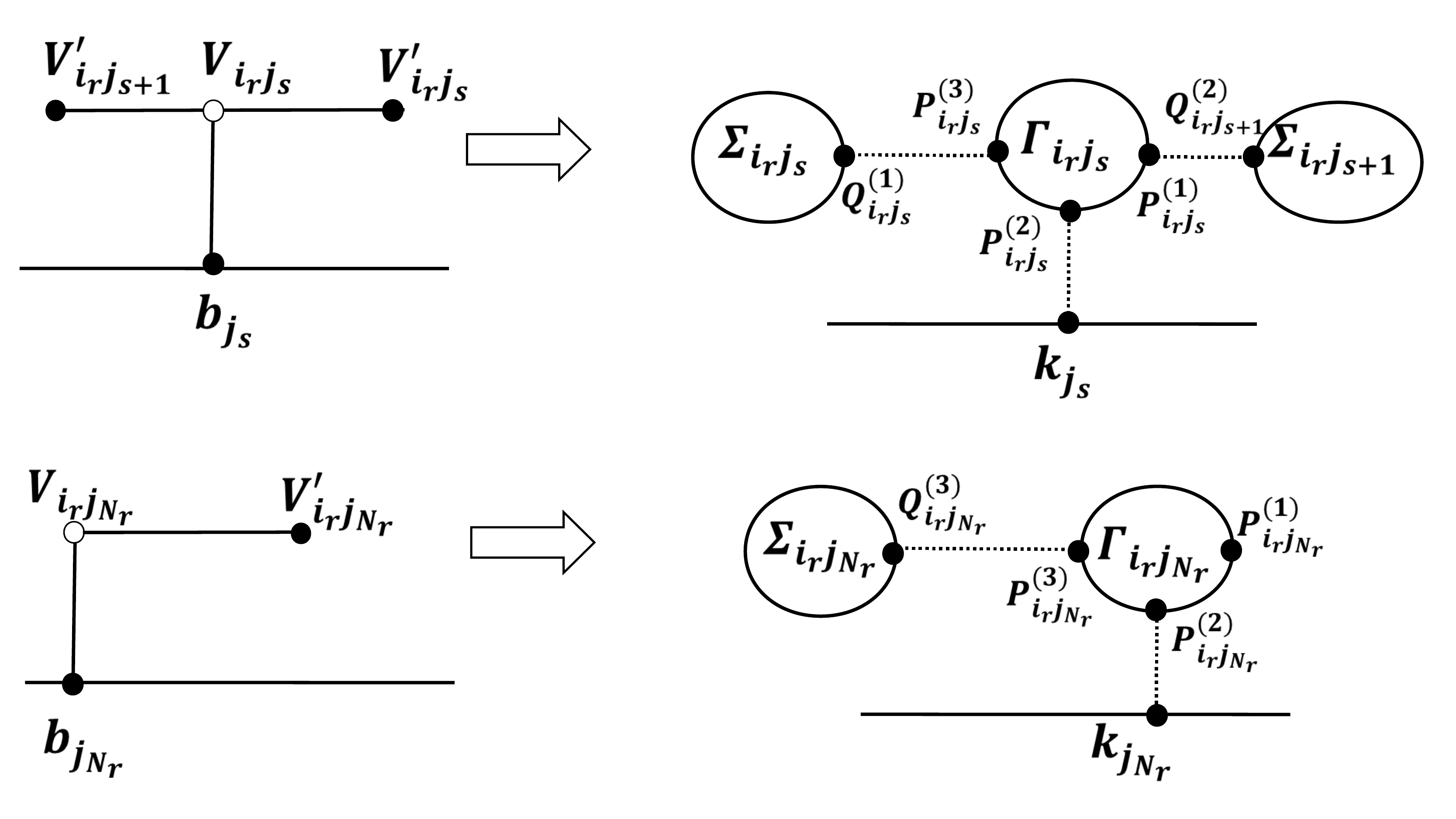}}
	\vspace{.8 truecm}
	{\includegraphics[width=0.43\textwidth]{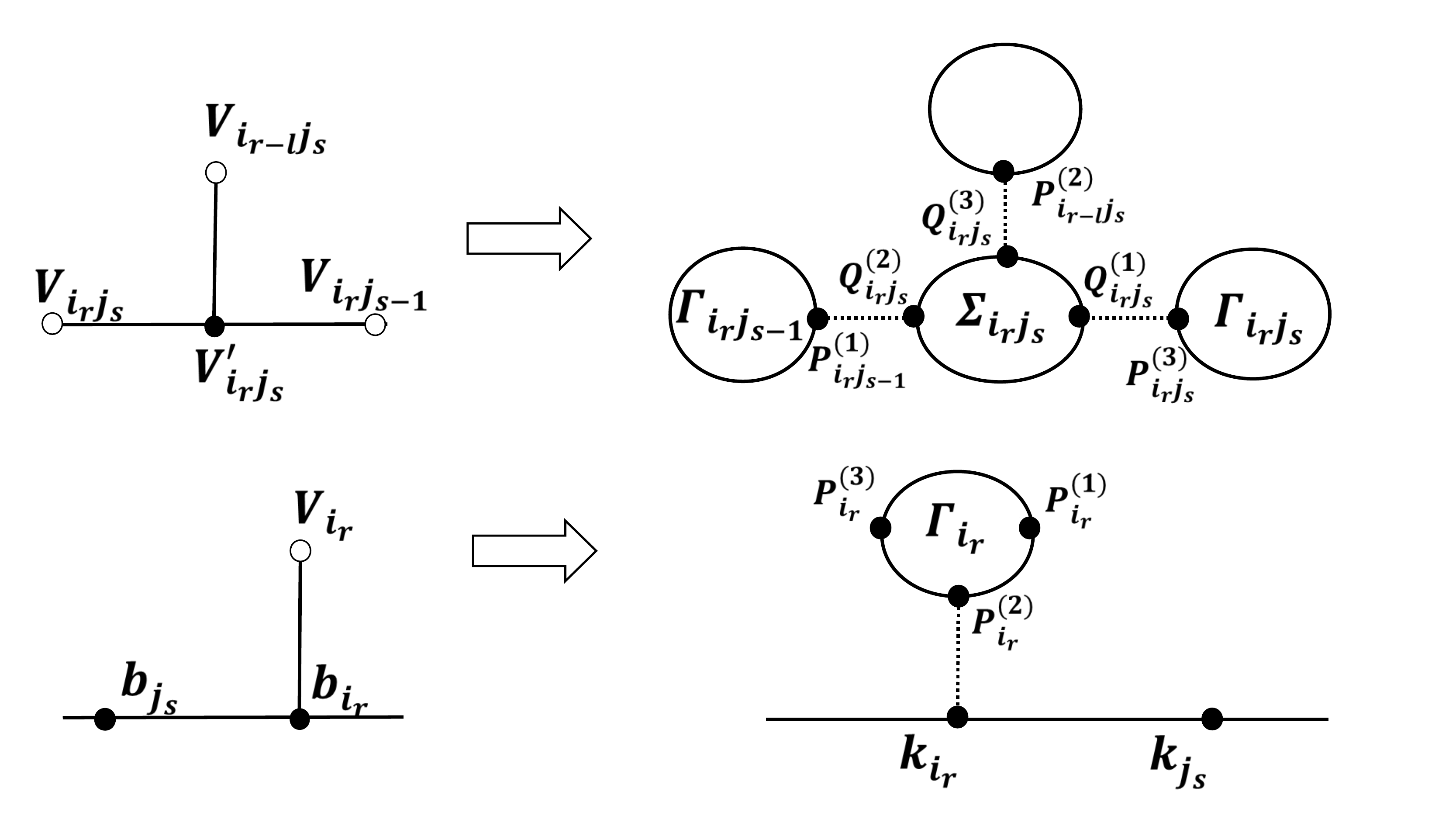}}
	\hspace{1.2 truecm}
	{\includegraphics[width=0.43\textwidth]{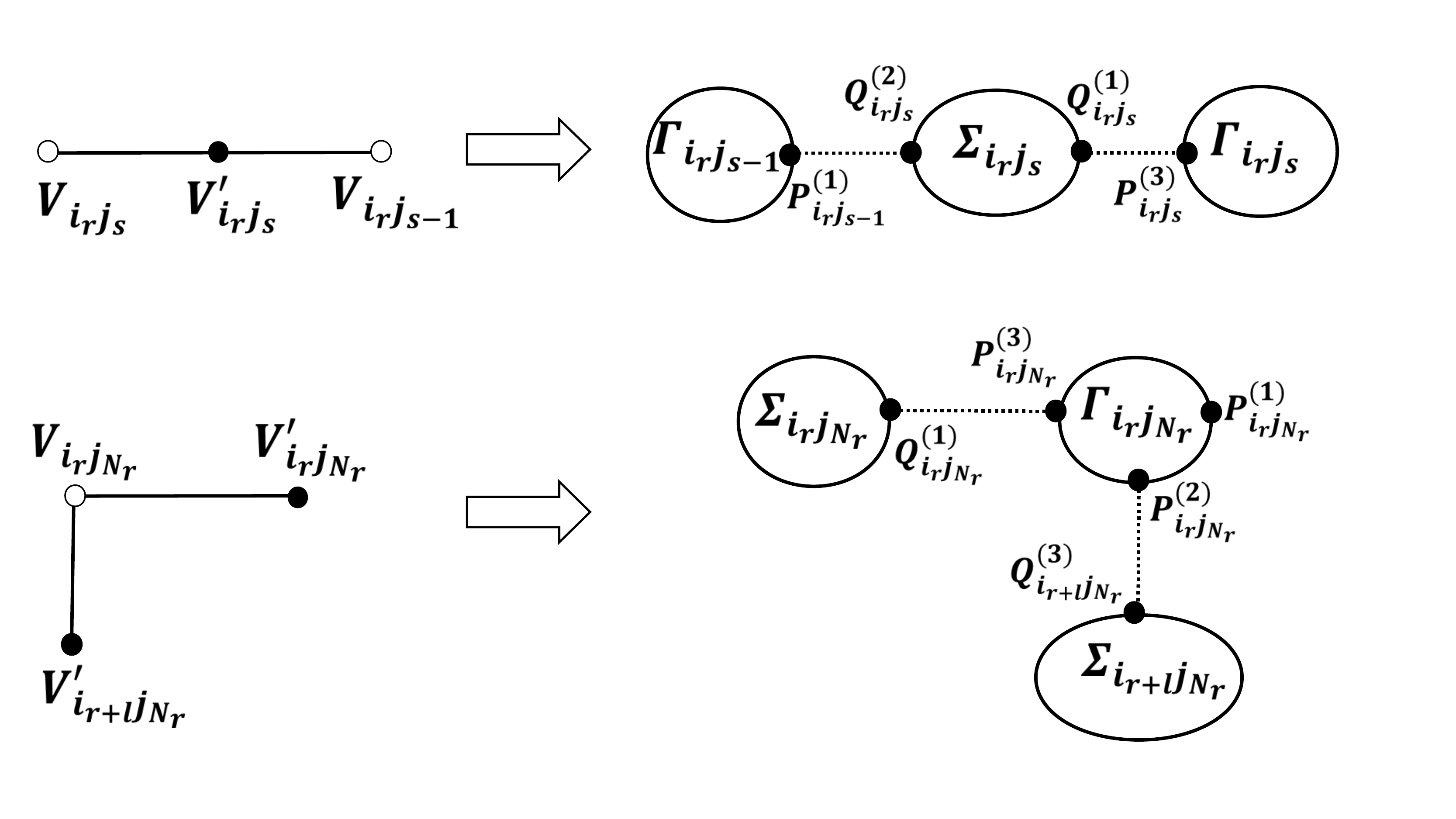}}
	\vspace{-.9 truecm}
  \caption{\footnotesize{\sl The gluing rules on $\Gamma$ are modeled on the  bipartite Le--graph ${\mathcal G}$ reflected w.r.t. the vertical axis. The dotted lines mark the points where we glue different copies of $\mathbb{CP}^1$.}}
	\label{fig:doublepoints}
\end{figure}

\item {\bf Horizontal gluing rules for fixed $r\in [k]$}:
\begin{enumerate} 
\item If, for some $r\in [k]$, $N_r=0$, then $P_{i_r}^{(1)}\in \Gamma_{i_r}$ is not glued to any other marked point;
\item If, for some $r\in [k]$, $N_r>0$, then $P_{i_r}^{(1)}\in \Gamma_{i_r}$ is glued to $Q_{i_rj_1}^{(2)}\in\Sigma_{i_rj_1}$;
\item For any $s\in [N_r-1]$, $P_{i_rj_s}^{(1)}\in \Gamma_{i_rj_s}$ is glued to $Q_{i_rj_{s+1}}^{(2)}\in\Sigma_{i_rj_{s+1}}$;
\item For any $s\in [N_r]$, $P_{i_rj_s}^{(3)}\in \Gamma_{i_rj_s}$ is glued to $Q_{i_rj_s}^{(1)}\in\Sigma_{i_rj_s}$;
\item $P_{i_rj}^{(3)}\in \Gamma_{i_r}$ is not glued to any other marked point.
\end{enumerate}
\item {\bf Vertical gluing rules}:
\begin{enumerate}
\item For any $r\in [k]$, $\kappa_{i_r}\in\Gamma_0$ is glued to $P_{i_r}^{(2)}\in \Gamma_{i_r}$;
\item For any  $j\in {\bar I}$ such that ${\bar r} = \max \{ r \in [k] \, : \, \chi^{i_r}_j=1\}>0$, $\kappa_{j}\in\Gamma_0$ is glued to $P_{i_{\bar r}j}^{(2)}\in \Gamma_{i_{\bar r}j}$;
\item If, for some $j\in {\bar I}$, $\chi^{i_r}_j=0$ for all $r\in [k]$, then $\kappa_{j}\in\Gamma_0$ is not glued to any other marked point;
\item For any fixed $r\in [2,k]$ and any fixed $s\in [N_r]$, let
${\bar r } =\max \{ l\in [1,r-1] \, : \chi^{i_l}_{j_s}=1\}$. Then $Q_{i_rj_s}^{(3)}\in\Sigma_{i_rj_3}$ is glued to $P_{i_{\bar r}j_s}^{(2)}\in \Gamma_{i_{\bar r}j_s}$.
\end{enumerate}
\item The faces of ${\mathcal G}$ correspond to the ovals of $\Gamma$. We label the ovals $\Omega_0$, $\Omega_{i_rj_s}$, $s\in [N_r]$, $r\in [k]$, as the corresponding faces of ${\mathcal N}$. 
\end{enumerate}
\end{construction}

\begin{remark} \textbf{Universality of the reducible rational curve $\Gamma$.}
Let us point out that, for any fixed positroid cell ${\mathcal S} =\S$, the construction of $\Gamma$ does \textbf{not} require the introduction of any parameter. Therefore it provides an \textbf{universal} curve $\Gamma$ for the whole positroid cell. In the next section we introduce the parametrization of $\S$ via KP divisors on $\Gamma$.
\end{remark}

\begin{figure}
  \centering
  {\includegraphics[width=0.46\textwidth]{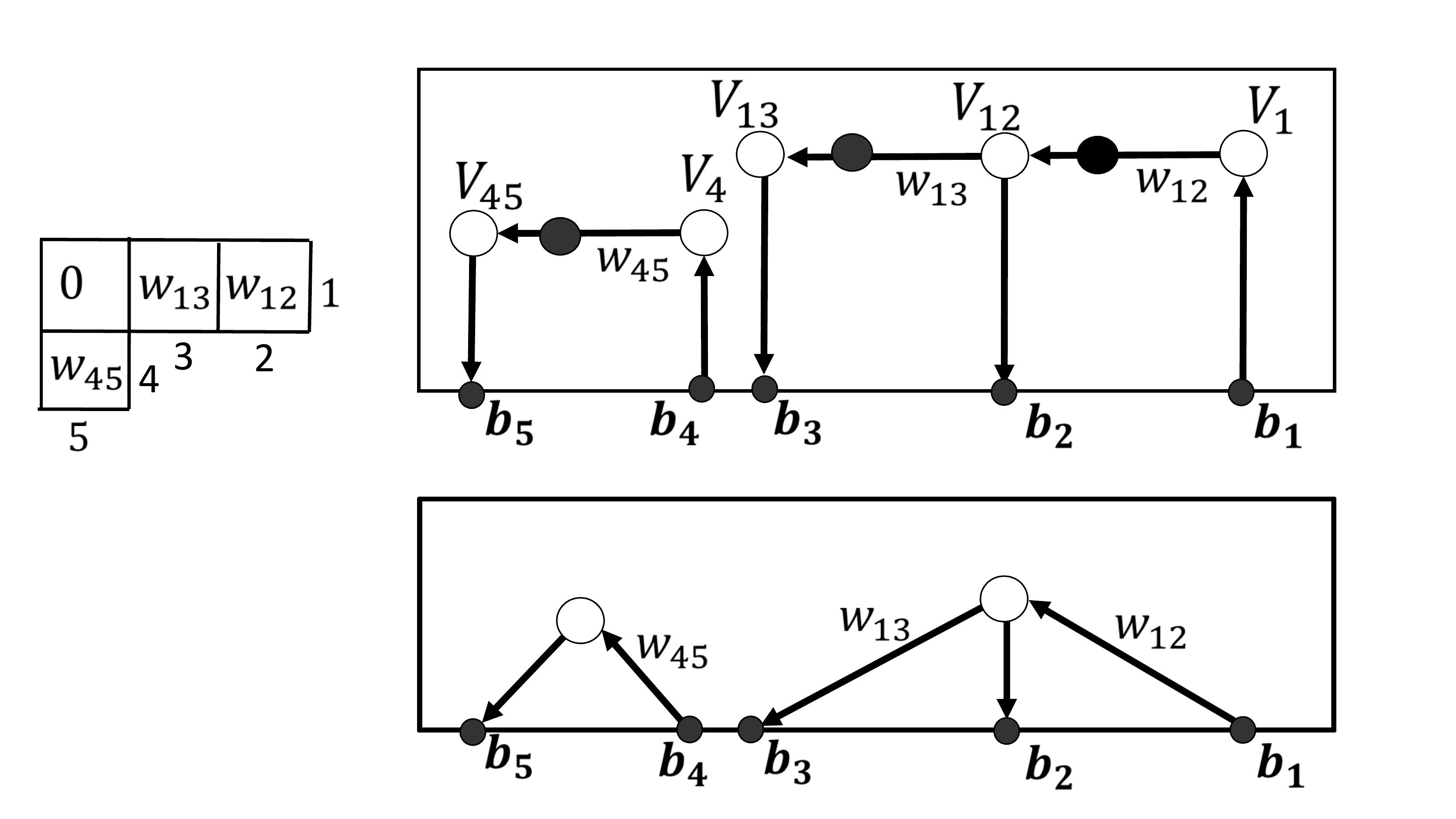}}
  \caption{\label{fig:biv_comp}\footnotesize{\sl The Le--networks corresponding to the Le--graph ${\mathcal G}$ [top] and its reduction ${\mathcal G}_{\mbox{\scriptsize red}}$ [bottom] for the same Le--tableau of a positroid cell of dimension $d=3$ in $Gr^{\mbox{\tiny TNN}} (2,5)$ [left]. The weight is one for any edge not marked in the Figure.}}
\end{figure}

\begin{remark}\label{rem:red_1}\textbf{The role of bivalent vertices, the reduced graph ${\mathcal G}_{\mbox{\scriptsize red}}$ and the reduced $\mathtt M$--curve $\Gamma({\mathcal G}_{\mbox{\scriptsize red}})$}
The number of copies of $\mathbb{CP}^1$ used to construct $\Gamma$ above is \textbf{excessive} in the sense that both the number of ovals and the \textbf{KP divisor} are invariant if we eliminate from $\mathcal G$ all copies of $\mathbb{CP}^1$ corresponding to bivalent vertices and change edge weights following \cite{Pos}. In this procedure we maintain a bivalent vertex in ${\mathcal G}_{\mbox{\scriptsize red}}$ for each component which disconnects from the graph upon removing the boundary of the disk and consists of a single boundary source connected to a single boundary sink. We show a simple example in Figure \ref{fig:biv_comp}. In the following, we denote ${\mathcal G}_{{\mbox{\scriptsize red}}}$ the reduced trivalent graph and $\Gamma({\mathcal G}_{\mbox{\scriptsize red}})$ the reducible rational curve associated to it.

For the construction of the \textbf{reducible rational} ${\mathtt M}$-curve we can use both $\mathcal G$ and ${\mathcal G}_{{\mbox{\scriptsize red}}}$ graphs. In Sections~\ref{sec:le} and \ref{sec:proof} we use the Le--graph ${\mathcal G}$ to evidence the recursive construction in the proof. 
However, it is also possible to directly construct the KP wave function and its divisor on $\Gamma({\mathcal G}_{\mbox{\scriptsize red}})$, since, by our construction, the KP wave function is constant with respect to the spectral parameter on each $\mathbb{CP}^1$ corresponding to a bivalent vertex.

For constructing a regular perturbed ${\mathtt M}$-curve of genus equal to $d$ it is convenient to start from $\Gamma({\mathcal G}_{\mbox{\scriptsize red}})$ since it corresponds to a nodal plane curve of degree lesser than that for $\Gamma({\mathcal G})$. In Section \ref{sec:example}, we use $\Gamma({\mathcal G}_{{\mbox{\scriptsize red}}})$ in the construction of the plane curve and of the KP divisor for soliton data in $Gr^{\mbox{\tiny TP}}(2,4)$.
\end{remark}

\begin{remark}\label{rem:comp2}\textbf{Comparison with the construction in  \cite{AG1}.}
In \cite{AG1}, to any given soliton data $({\mathcal K}, [A])$, $[A]\in Gr^{\mbox{\tiny TP}}(k,n)$, we associate a curve obtained gluing $k+1$ copies of $\mathbb{CP}^1$ at double points whose position is ruled by a parameter $\xi\gg 1$. We then control the asymptotic leading behavior in $\xi$ of the vacuum wave function via an algebraic construction using the positivity properties of a specific representative matrix of $[A]$. In this approach the number of $\mathbb{CP}^1$ components is 
much smaller, but we have to introduce extra parameters marking the positions of the glued points. In practice one can obtain such curve from the universal one by a proper desingularization of some double point (see also Section \ref{sec:example}, where we desingularize explicitly $\Gamma({\mathcal G}_{{\mbox{\scriptsize red}}})$ to $\Gamma(\xi)$ when $[A]\in Gr^{\mbox{\tiny TP}}(2,4)$).
\end{remark}

\begin{proposition}\textbf{The oval structure of $\Gamma$.}
Let ${\mathcal K} = \{ \kappa_1 < \cdots < \kappa_n\}$ and $\S\subset \Grkn$ be a positroid cell of dimension $d$. Let $\Gamma$ be as in Construction~\ref{def:gamma}. Then $\Gamma$ possesses $d+1$ ovals which we label $\Omega_0$, $\Omega_{i_r j_s}$, $s\in [N_r]$, $r\in [k]$, $N_r\ge 1$. Moreover the ovals are uniquely identified by the following properties:
\begin{enumerate}
\item  $\Omega_0$ is the unique oval whose boundary contains both $\kappa_1$ and $\kappa_n$; 
\item  For any $r\in [k]$, $s\in [2,N_r]$, $\Omega_{i_r j_s}$ is the unique oval  whose boundary contains both $P^{(2)}_{i_r j_s}\in \Gamma_{i_r j_s}$ and $P^{(2)}_{i_r j_{s-1}}\in \Gamma_{i_r j_{s-1}}$;
\item  For any $r\in [k]$, $\Omega_{i_r j_1}$ is the unique oval  whose boundary contains both $P^{(2)}_{i_r j_1}\in \Gamma_{i_r j_1}$ and $P^{(2)}_{i_r}\in \Gamma_{i_r}$.
\end{enumerate}
\end{proposition}

The proof is straightforward and we omit it. We remark that $\Gamma({\mathcal G}_{\mbox{\scriptsize red}})$ has the same number of
ovals as $\Gamma({\mathcal G})$.

Let $d$ be the dimension of the \textbf{irreducible} positroid cell $\S\subset \Grkn$. Let ${\mathcal G}_{\mbox{\scriptsize red}}$ be its reduced graph as in Remark \ref{rem:red_1} and suppose that it has $n_b$ bivalent vertices after the reduction. Then $\Gamma({\mathcal G}_{\mbox{\scriptsize red}})$ is a partial normalization \cite{ACG} of a connected reducible nodal plane curve with $d+1$ ovals obtained by gluing $2d-n+n_b+1$ copies of $\mathbb{CP}^1$. The curve $\Gamma({\mathcal G}_{{\mbox{\scriptsize red}}})$ is a rational degeneration of a genus $d$ smooth $\mathtt M$--curve. The total number of edges of ${\mathcal G}_{\mbox{\scriptsize red}}$ is $3d-n+n_b$, and each of them corresponds to an handle of the desingularized $\mathtt M$--curve. In the next Proposition we verify that the genus of the latter coincides with the dimension of $\S$.

\begin{proposition}\textbf{ $\Gamma({\mathcal G}_{\mbox{\scriptsize red}})$ is the rational degeneration of a smooth $\mathtt M$-curve of genus $d$.}
Let ${\mathcal K} = \{ \kappa_1 < \cdots < \kappa_n\}$ and $\S\subset \Grkn$ be an irreducible positroid cell of dimension $d$. Let $\Gamma$ be as in Construction~\ref{def:gamma} and $\Gamma({\mathcal G}_{\mbox{\scriptsize red}})$ be the its reduction obtained by eliminating the components corresponding to bivalent vertices eliminated in ${\mathcal G}_{\mbox{\scriptsize red}}$. Then $\Gamma({\mathcal G}_{\mbox{\scriptsize red}})$ is a rational degeneration of a regular $\mathtt M$--curve of genus $d$ equal to the dimension of the positroid cell, possessing $d+1$ ovals.
\end{proposition}

\begin{proof}
The only untrivial statement is the one concerning the genus of the perturbed curve. Let $n_b$ be the number of bivalent vertices survived the reduction of the graph ${\mathcal G}$ according to Remark~\ref{rem:red_1}. By definition, $\Gamma({\mathcal G}_{\mbox{\scriptsize red}})$ is represented by $2d-n+n_b+1$ copies of $\mathbb{CP}^1$ connected at $3d-n+n_b$ pairs of double points. The regular curve is obtained opening a gap at each pair of these double points. We perform this desingularization respecting the real structure and keeping the number of real ovals fixed.

By construction, the desingularized curve has genus $g= \#\mbox{handles}-\#\mathbb{CP}^1+1= 3d-n+n_b - (2d-n+n_b+1) +1 =d$, and it possesses $d+1$ real ovals, therefore it is an  ${\mathtt M}$-curve. 
\end{proof}

\subsection{The planar representation of the desingularized curve.}
\label{sec:planar}

Generic Riemann surfaces cannot be holomorphically mapped into $\mathbb{CP}^2$ without self-intersections \cite{GrH}, therefore partial normalization is necessary if the number of $\mathbb{CP}^1$ copies is sufficiently high. In our construction we have $2d-n+n_b+1$ copies of $\mathbb{CP}^1$, which may be lines, quadrics or rational cubics in $\mathbb{CP}^2$. Denote the numbers of lines, quadrics and cubics by $n_1$, $n_2$, $n_3$ respectively. Clearly $n_1+n_2+n_3=2d-n+n_b+1$, the total degree of the rational reducible curve $\Gamma({\mathcal G}_{\mbox{\scriptsize red}})$ is $n_1+2n_2+3n_3$. The total number of singularities before normalization is 
$$
n_s=\frac{n_1(n_1-1)}{2} + 2 n_1 n_2 + 3 n_1 n_3 + 2 n_2(n_2-1) + 6 n_2 n_3 + \frac{9 n_3(n_3-1)}{2}+n_3
$$
The last term in the above sum takes into account that all cubics are rational and each has one cusp. When we desingularize $\Gamma({\mathcal G}_{\mbox{\scriptsize red}})$ to the 
genus $d$ curve, $n_s-3d+n-n_b$ intersections have to remain intersections for its plane curve model, and they are resolved after normalization.

Let us provide evidence that we have enough parameters. Let us assume that $\Gamma_0$ is defined by $\mu=0$, and we have 3 systems of linear functions $l_j=a_j \lambda + b_j \mu$, $j\in[n_1]$,  $l'_j=a'_j \lambda + b'_j \mu$, $j\in[n_2]$, $l''_j=a''_j \lambda + b''_j \mu$, $j\in[n_3]$ such that
\begin{enumerate}
\item All slopes are pairwise distinct and all $a_j$, $a'_j$, $a''_j$ are non-zero; 
\item The system of lines $\Gamma_0$, ${\mathcal L}_j:\{l_j-\alpha_j=0 \}$ intersect only in pairs.  
\end{enumerate}
The quadrics and cubics are represented by ${\mathcal Q}_j=0$, $j\in[n_2]$ and ${\mathcal C}_i=0$, $i\in[n_3]$ respectively, where ${\mathcal Q}_j= y-\alpha'_{j,2}(l'_j)^2-\alpha'_{j,1}l'_j-\alpha'_{j,0}$ and ${\mathcal C}_i=  y-\alpha''_{i,3}(l''_i)^3-\alpha''_{i,2}(l''_i)^2-\alpha''_{i,1}l''_i-\alpha''_{i,0}$. 

The coefficients $\alpha$, $\alpha'$, $\alpha''$ have to be chosen so that all lines, quadrics and cubics intersect $\Gamma_0$ at proper points. We also assume that  all $\alpha'_{j,2}$, $\alpha''_{i,3}$ are sufficiently large, so that all quadrics and cubics are small perturbations of pairs or triples of parallel lines respectively, and all intersections of components are real. 

The unperturbed curve has then the following form
$$
{\Pi}_0(\lambda,\mu)=0, \ \ \mbox{where} \ \  {\Pi}_0(\lambda,\mu) =\mu \prod\limits_{i_1\in[n_1]} {\mathcal L}_{i_1} \prod\limits_{i_2\in[n_2]} {\mathcal Q}_{i_2} \prod\limits_{i_3\in[n_3]} {\mathcal C}_{i_3}.
$$
We use the following collection of perturbative terms: $\{ {\Pi}^{[0,1]}_{i_1};{\Pi}^{[0,2],k}_{i_2},k\in[2];{\Pi}^{[0,3],k}_{i_3},k\in[3]; {\Pi}^{[1,1]}_{i_1,j_1};{\Pi}^{[1,2],k}_{i_1,i_2},k\in[2];{\Pi}^{[1,3],k}_{i_1,i_3},k\in[3];{\Pi}^{[2,2],k_1,k_2}_{i_2,j_2},k_1,k_2\in[2];{\Pi}^{[2,3],k_1,k_2}_{i_2,i_3},k_1\in[2],k_2\in[3]; {\Pi}^{[3,3],k_1,k_2}_{i_3,j_3},k_1,k_2\in[3]\}$, where $i_l,j_l\in[n_l]$, $l\in[3]$, $i_1<j_1$,  $i_2<j_2$,  $i_3<j_3$, and
$$
{\Pi}^{[0,1]}_{i_1}=\frac{{\Pi}_0}{\mu {\mathcal L}_{i_1}},  \ \ {\Pi}^{[0,2],k}_{i_2}=\frac{{\Pi}_0 (l'_{i_2})^{k-1}}{\mu {\mathcal Q}_{i_2}},  \ \ {\Pi}^{[0,3],k}_{i_3}=\frac{{\Pi}_0 (l''_{i_3})^{k-1}}{\mu {\mathcal C}_{i_3}},
$$
$$
{\Pi}^{[1,1]}_{i_1,j_1}=\frac{{\Pi}_0}{{\mathcal L}_{i_1} {\mathcal L}_{j_1}},  \ \ {\Pi}^{[1,2],k}_{i_1,i_2}=\frac{{\Pi}_0 (l'_{i_2})^{k-1}}{{\mathcal L}_{i_1} {\mathcal Q}_{i_2}},  \ \ {\Pi}^{[1,3],k}_{i_1,i_3}=\frac{{\Pi}_0 (l''_{i_3})^{k-1}}{{\mathcal L}_{i_1} {\mathcal C}_{i_3}},
$$
$$
{\Pi}^{[2,2],k_1,k_2}_{i_2,j_2}=\frac{{\Pi}_0  (l'_{i_2})^{k_1-1} (l'_{j_2})^{k_2-1} }{{\mathcal Q}_{i_2} {\mathcal Q}_{j_2}},  \ \ {\Pi}^{[2,3],k_1,k_2}_{i_2,i_3}=\frac{{\Pi}_0 (l'_{i_2})^{k_1-1} (l''_{i_3})^{k_2-1}}{{\mathcal Q}_{i_2} {\mathcal C}_{i_3}},  \ \ {\Pi}^{[3,3],k_1,k_2}_{i_3,j_3}=\frac{{\Pi}_0 (l''_{i_3})^{k_1-1}(l''_{j_3})^{k_2-1}  }{{\mathcal C}_{i_3} {\mathcal C}_{j_3}}.
$$
We then consider the following perturbation of our rational curve ${\Pi}_0(\lambda,\mu)=0$:
\begin{equation}
\label{eq:perturbed_curve}
{\Pi}(\lambda,\mu)=0, \ \ \mbox{where} \ \  {\Pi}(\lambda,\mu)={\Pi}_0+\sum {\epsilon}^{r}_{s} {\Pi}^{r}_{s},
\end{equation}
where the sum runs over all perturbation terms described above. The perturbed curve in $\mathbb{CP}^2$ has the same structure at the infinite line as the original rational curve. The number of perturbation parameters ${\epsilon}^{r}_{s}$  in (\ref{eq:perturbed_curve}) coincides with the number of intersections in ${\Pi}_0(\lambda,\mu)=0$. For sufficiently small   ${\epsilon}^{r}_{s}$ we have the following map
\begin{equation}
\label{eq:analytic_map}
\{{\epsilon}^{r}_{s}\}\rightarrow  {\Pi}({\mathcal R}^{r}_{s}),
\end{equation}
where ${\mathcal R}^{r}_{s}$ are the solutions of the system
\begin{equation}
\partial_{\lambda}{\Pi}(\lambda,\mu)=0, \ \ \partial_{\mu}{\Pi}(\lambda,\mu)=0.
\end{equation}
For the unperturbed curve, the set $\{ {\mathcal R}^{r}_{s} \}$ coincides with the intersection points, therefore for small ${\epsilon}^{r}_{s}$ we have a natural enumeration. The map (\ref{eq:analytic_map}) is analytic for $|\epsilon^r_s|\ll 1$ and its Jacobian is non-zero, therefore it is locally invertible, and at each double point we can open a gap in the desired direction, or keep the point double.  

Let us remark that these arguments are analogous to arguments used in \cite{Kr4}.

\subsection{The KP divisor on $\Gamma$}
\label{sec:KPdiv}

Throughout this section we fix a set of phases ${\mathcal K} =\{ \kappa_1 < \cdots <\kappa_n\}$ and a positroid cell $\S \subset \Grkn$of dimension $g$. $\Gamma$ is the curve of Construction~\ref{def:gamma} associated to such data with marked point $P_0\in \Gamma_0$.

\textbf{In this section we state the main results of our paper: for any soliton data $({\mathcal K}, [A])$, $[A]\in \S$, we construct a unique real and regular degree $g$ KP divisor $\DKP$ on $\Gamma$ as follows:}
\begin{enumerate}
\item We first construct a unique degree $g$ effective real and regular \textbf{vacuum} divisor $\DVG$ and a unique real and regular \textbf{vacuum} 
wave function $\hat \phi(P,\vec t)$ on $\Gamma$ satisfying appropriate boundary conditions (Theorem~\ref{theo:exist});
\item We then apply the Darboux dressing to such vacuum wave function and define the normalized \textbf{dressed divisor} $\DKP$, which, by construction is effective and of degree $g$;
\item $\DKP$ has minimal degree $g$ and is the \textbf{KP divisor} for the given soliton data on the spectral curve $\Gamma$, since its restriction to $\Gamma_0$ coincides with the Sato divisor defined in Section~\ref{sec:Sato}, and, by construction, it satisfies the reality and regularity conditions of Definition \ref{def:rrss} (see Theorem~\ref{theo:KPdiv0}).
\end{enumerate}

\begin{remark} \textbf{Parametrization of positroid cells.}
In our construction we associate a system of edge vectors to each Le--network in $\S$. The properties of such edge vectors guarantee that the non--normalized KP wave function has untrivial dependence on $\vec t$ at all double points of $\Gamma$. Therefore, for any $[A_0]\in S$ it is possible to find an initial time $\vec t_0$ such that the degree $g$ KP divisor $\DKP = \DKP (\vec t_0)$ is non--special for any point $[A]\in S$ sufficiently near to $[A_0]$ in the natural metric associated to the reduced row echelon matrix representation of such points. It is in this sense that we obtain a parametrization of $\S$ via degree $g$ non--special KP divisors.
\end{remark}

We start introducing local affine coordinates on each copy of $\mathbb{CP}^1$ and we use the same symbol $\zeta$  for any such affine coordinate to simplify notations (see also Figure~\ref{fig:lcoord}).

\begin{figure}
  \centering
  {\includegraphics[width=0.60\textwidth]{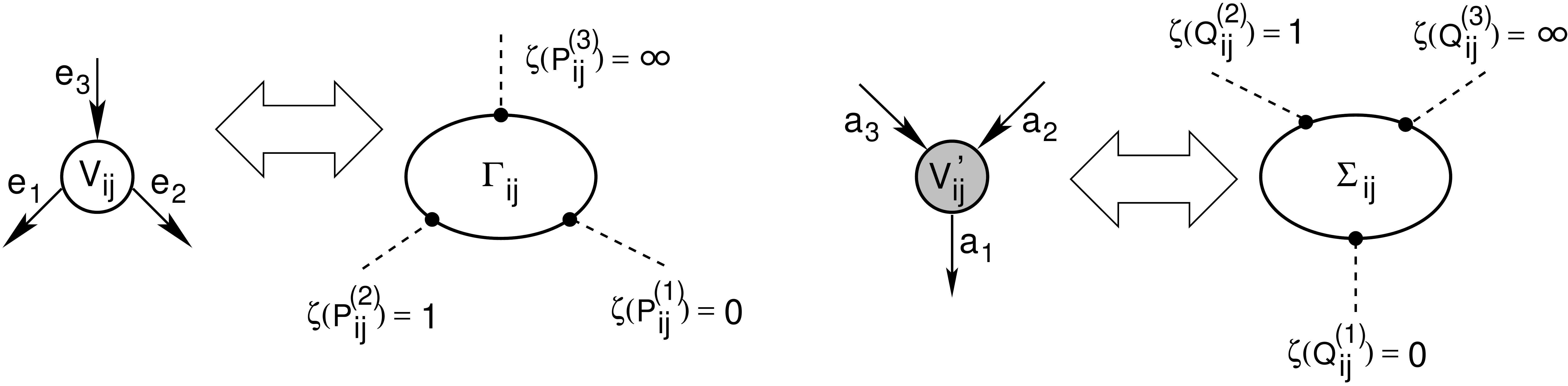}}
  \caption{\footnotesize{\sl Local coordinates on the components $\Gamma_{ij}$ and $\Sigma_{ij}$: the marked point $P^{(s)}_{ij}\in \Gamma_{ij}$ corresponds to the edge $e_s$ at the white vertex $V_{ij}$ and the marked point $Q_s\in \Gamma_{ij}$ corresponds to the edge $a_s$ at the black vertex $V_{ij}^{\prime}$.}\label{fig:lcoord} }
\end{figure}

\begin{definition}\label{def:loccoor}{\bf Local affine coordinate on $\Gamma$}
On each copy of $\mathbb{ CP}^{1}$ the local coordinate $\zeta$ is uniquely identified by the following properties: 
\begin{enumerate}
\item On $\Gamma_0$, $\zeta^{-1} (P_0)=0$ and $\zeta(\kappa_1)< \cdots < \zeta(\kappa_n)$. To abridge notations, we identify the $\zeta$--coordinate with the marked points $\kappa_j=\zeta(\kappa_j)$, $j\in [n]$;
\item On the component $\Gamma_{ij}$ corresponding to the internal white vertex $V_{ij}$:
\[
\zeta (P^{(1)}_{ij}) =0, \;\;\zeta (P^{(2)}_{ij}) =1, \;\;\zeta (P^{(3)}_{ij}) =\infty,
\]
while on the component $\Sigma_{ij}$ corresponding to the internal black vertex $V^{\prime}_{ij}$:
\[
\zeta (Q^{(1)}_{ij}) =0, \;\;\zeta (Q^{(2)}_{ij}) =1, \;\;\zeta (Q^{(3)}_{ij}) =\infty.
\]
\end{enumerate}
\end{definition}

In the following Definitions we state the desired properties for both the vacuum divisor and the vacuum wave function on $\Gamma$.

\begin{definition}\label{def:real_vac_div}\textbf{Real and regular vacuuum divisor compatible with $\S$.}  Let $\Omega_0$ be the infinite oval containing the marked point $P_0\in\Gamma_0$ and let $\Omega_j$, $j\in[g]$ be the finite ovals of $\Gamma$. Let $P^{(i_r)}_3$, $r\in [k]$ be the Darboux points in $\Gamma$.

We call a degree $g$ divisor $\DVG\in\Gamma\backslash\Gamma_0$ a real and 
regular vacuum divisor compatible with $\S$ if:
\begin{enumerate}
\item $\DVG$ is contained in the union of all the ovals of $\Gamma$; 
\item There is exactly one divisor point on each component of $\Gamma$ corresponding to a trivalent white vertex or a bivalent white vertex containing a Darboux point; 
\item\label{item:defodd} In any $\Omega_j$, $j\in [g]$, the total number of vacuum divisor poles plus the number of Darboux points is 1 mod 2;
\item\label{item:defeven} In $\Omega_0$, the total number of vacuum divisor poles plus the number of Darboux points plus $k$ is 0 mod 2.
\end{enumerate}
\end{definition}

\begin{definition}\label{def:vacuumKP}{\bf A real and regular vacuum wave function on $\Gamma$ corresponding to $\DVG$:}  
Let $\DVG$ be a degree $g$ real regular divisor on $\Gamma$ as in Definition~\ref{def:real_vac_div}.
A function ${\hat \phi} (P, \vec t)$, where $P\in\Gamma\backslash \{ P_0\}$ and $\vec t$ are the KP times, is called a real and regular vacuum wave function on $\Gamma$ corresponding to $\DVG$ if:
\begin{enumerate}
\item\label{it:second} There exists $\vec t_0$ such that $\hat \phi (P, \vec t_0)=1$ at all points $P\in \Gamma\backslash \{P_0\} $;
\item The restriction of $\hat \phi$ to $\Gamma_0\backslash \{P_0\}$ coincides with the following normalization of Sato's wave function,  ${\hat \phi} (\zeta(P), \vec t)=e^{\theta (\zeta, \vec t-\vec t_0)}$, where $\theta (\zeta, \vec t) = \sum_{l\ge 1} t_l \zeta^l$;
\item\label{it:first} For all $P\in\Gamma\backslash\Gamma_0$ the function ${\hat \phi} (P, \vec t)$ satisfies all equations (\ref{eq:vacuum_eq}) of the vacuum hierarchy;
\item If both $\vec t$ and $\zeta(P)$ are real, then ${\hat \phi} (\zeta(P), \vec t)$ is real. Here $\zeta(P)$ is the local affine coordinate on the corresponding component of $\Gamma$ as in Definition~\ref{def:loccoor};
\item $\hat \phi$ takes equal values at pairs of glued points $P,Q\in \Gamma$, for all $\vec t$:  $\hat \phi(P, \vec t) = \hat \phi(Q, \vec t)$;
\item\label{it:last} For each fixed $\vec t$ the function $\hat \phi(P, \vec t)$ is meromorphic of degree $\le g$ in $P$ on $\Gamma\backslash \{ P_0\}$: for any fixed $\vec t$ we have $({\hat \phi} (P, \vec t))+\DVG\ge 0$ on $\Gamma\backslash P_0$, where $(f)$ denotes the divisor 
of $f$. Equivalently, for any fixed $\vec t$ on $\Gamma\backslash \{ P_0\}$ the function $\hat \phi(\zeta, \vec t)$ is regular outside the points of 
$\DVG$ and at each of these points either it has a first order pole or it is regular;
\item For each $P\in\Gamma\backslash \{ P_0\}$ outside  $\DVG$  the function $\hat \phi(P, \vec t)$ is regular in $\vec t$ for all times.
\end{enumerate}
\end{definition}

\begin{definition}\label{def:vacuumKP2}{\bf A real and regular vacuum wave function on $\Gamma$ for the soliton data $(\mathcal K,[A])$:} Let ${\mathcal K}$, $\Gamma$ and  ${\hat \phi} (P, \vec t)$ be as in Construction~\ref{def:gamma} and in Definition~\ref{def:vacuumKP}. Let $[A]\in\S$.
The function  ${\hat \phi} (P, \vec t)$ is a real and regular vacuum wave function for the soliton data  $(\mathcal K,[A])$ if, at each Darboux point $P^{(3)}_{i_r}$, $r\in [k]$ and for all $\vec t$,
\begin{equation}\label{item:dar}
\hat \phi(P^{(3)}_{i_r}, \vec t)\equiv f^{(r)} (\vec t),
\end{equation}
 where $f^{(r)} (\vec t)$, $r\in[k]$, generate the Darboux transformation for the soliton data.
\end{definition}

\begin{theorem}\label{theo:exist}\textbf{Existence and uniqueness of a real and regular divisor and vacuum wave function 
on $\Gamma$ satisfying appropriate boundary conditions.}
Let $({\mathcal K}, [A])$ be given soliton data with $[A]\in\S$ of dimension $g$, and let $\Gamma$ be as in Construction~\ref{def:gamma} with Darboux points $\{ P^{(3)}_{i_r}, r\in [k] \}$. Then, we can fix an initial time $\vec t_0$ such that to the following data $({\mathcal K}, [A]; \Gamma, P_0, P^{(3)}_{i_1},\dots, P^{(3)}_{i_k}; \vec t_0 )$ we associate
\begin{enumerate}
\item A \textbf{unique} real and regular degree $g$ vacuum divisor  
$\DVG$ as in Definition \ref{def:real_vac_div},
\item A \textbf{unique} real and regular vacuum wave function $\hat \phi(P, \vec t)$ corresponding to this divisor satisfying Definitions~\ref{def:vacuumKP} and \ref{def:vacuumKP2}.
\end{enumerate}
Moreover, at the Darboux points, $\hat \phi(P, \vec t)$ satisfies 
\begin{equation}\label{eq:darboux_cond}
\hat \phi(P^{(3)}_{i_r}, \vec t)\equiv \frac{\sum_{l=1}^n A^{r}_l \exp (\theta_l (\vec t))}{\sum_{l=1}^n A^{r}_l \exp (\theta_l (\vec t_0))}, \ \ r\in[k], \quad\quad \forall \vec t,
\end{equation}
where $f^{(r)}(\vec t)$ are the generators of the Darboux transformation associated to the RRE representative matrix $A$.
\end{theorem}

We prove Theorem \ref{theo:exist} in Section~\ref{sec:proof}. More precisely, we construct a unique vacuum wave function on $\Gamma$ using the algebraic recursion settled in Section \ref{app:mainalgtheo}. We first modify the Le-network moving the boundary sources to convenient inner vertices, added in correspondence of the Darboux points in $\Gamma$. Then we assign a vector constructed in Section~\ref{sec:rowvectors} to each vertical edge of this modified network and use the linear relations at the inner vertices to extend this system of vectors to all its edges. 
We use this system of vectors to define a unique vacuum edge wave function (v.e.w.) satisfying the necessary boundary conditions. By construction, the linear relations at trivalent white vertices define a degree $g$ divisor with the required reality and regularity conditions (see Lemma \ref{lemma:vacvertexwf}).
Finally, we construct the degree $g$ real and regular vacuum wave function on $\Gamma$ imposing that it takes the value of the normalized v.e.w. at the marked points (double points and Darboux points) which correspond to the edges (see Theorem \ref{theo:existvac}).

\begin{definition}\label{def:dressKP}\textbf{The dressing of the vacuum wave function on $\Gamma$.}
Let $\Gamma$ and $\hat\phi$ be as in Theorem~\ref{theo:exist} for given soliton data $({\mathcal K}, [A])$. 
Then the corresponding \textbf{Darboux transformed wave function} is defined by:
\begin{equation}\label{eq:dressKP}
\psi(P, \vec t) = {\mathfrak D}{\hat \phi} (P,\vec t), \quad\quad P\in \Gamma\backslash \{ P_0\}, \; \forall \vec t,
\end{equation}
where ${\mathfrak D}$ is the Darboux dressing differential operator for the soliton data $({\mathcal K}, [A])$ defined
in (\ref{eq:D}). Let the initial condition $\vec t_0$ be such that ${\mathfrak D}{\hat \phi} (P,\vec t_0) 	\not = 0$ at all double points $P\in  \Gamma$. We define the \textbf{normalized dressed wave function} ${\hat \psi} (P,\vec t)$ as
\begin{equation}\label{eq:normKP}
{\hat \psi}(P, \vec t) = \frac{\psi (P,\vec t)}{\psi (P,\vec t_0)}.
\end{equation}
We define the \textbf{normalized dressed divisor} as
\begin{equation}
\label{eq:d_dressed2}
\DKP = \DDG + (\psi (P,\vec t_0)).
\end{equation}
where the non-effective divisor $\DDG$ is defined by
\begin{equation}
\label{eq:d_dressed}
\DDG = \DVG + k P_0 - \sum\limits_{r=1}^k P_{i_r}.
\end{equation}
\end{definition}

\begin{remark}
For reducible soliton data $[A]$ in $\Grkn$ with $s$ isolated boundary sources we have two Darboux dressings: the reducible $k$-th
order dressing operator and the reduced $(k-s)$-order dressing operator (see Remarks \ref{rem:irred} and \ref{rem:reddiv}). The normalized dressed wave function is the same for both dressings, while the $\DKP$ divisor associated to the reducible dressing operator contains $s$ extra points in the intersection of the finite ovals with $\Gamma_0$, so that we have more than one divisor point in some of the finite ovals. Therefore, in such case, in the Definition above we use the $(k-s)$-th order reduced dressing operator. The extra $s$ divisor points may be interpreted as being originated from real and regular divisor data on regular $\mathtt M$--curves of genus $g+s$ under the assumption that $s$ ovals degenerate to points in the solitonic limit.
\end{remark}

\begin{lemma} 
\label{lem:d_dressed} 
\begin{enumerate}
\item For any $\vec t$ we have the following inequality on $\Gamma\backslash P_0$:
\begin{equation}
({\psi} (P, \vec t))+\DDG\ge 0.
\end{equation}
\item The number of poles minus the number of zeroes for $\DDG$ (counted with multiplicities, 
if necessary) at each finite oval is odd and at the infinite oval it is even. 
\end{enumerate}
\end{lemma}
The first property follows directly from the definition of $\psi(P,\vec t)$. The second statement simply follows from properties of the vacuum wave 
function constructed in Theorem~\ref{theo:exist}, namely equation (\ref{item:dar}) in Definition \ref{def:vacuumKP2}, properties (\ref{item:defodd})-(\ref{item:defeven}) in Definition \ref{def:real_vac_div} and formula~(\ref{eq:d_dressed}). 

\begin{lemma} 
\label{lem:d_dressed2} 
\begin{enumerate}
\item For any $\vec t$ we have the following inequality on $\Gamma\backslash P_0$:
\begin{equation}
({\hat \psi} (P, \vec t))+\DKP \ge 0.
\end{equation}
\item The divisor $\DKP$ is effective and has degree $g$. 
\item All poles of  $\DKP$ lie at the finite ovals, and each finite oval contains exactly one pole 
of $\DKP$.
\end{enumerate}
\end{lemma}
The first statement follows immediately from the definition of ${\hat \psi} (P, \vec t)$ and  $\DKP $.
The second statement follows immediately from (\ref{eq:d_dressed2}). The third statement follows from the fact that 
$\psi (P,\vec t_0)$ is real at all real ovals and from Lemma~\ref{lem:d_dressed}.

\begin{theorem}\label{theo:KPdiv0}\textbf{The effective divisor on $\Gamma$.}
Assume that $\hat \phi$ is the real and regular vacuum wave function on $\Gamma$ of Theorem~\ref{theo:exist} for the given soliton 
data $({\mathcal K}, [A])$. Let $\hat \psi$ be the normalized dressed wave function from Definition \ref{def:dressKP}.

Then the divisor $\DKP$ is the KP divisor on $\Gamma$ for the soliton data $({\mathcal K}, [A])$, and it satisfies the reality and regularity conditions of Definition \ref{def:rrss}, whereas ${\hat \psi}$ is the KP wave function on $\Gamma$ for the soliton data $({\mathcal K},[A])$.
Moreover, the degree of the effective pole divisor of ${\hat \psi}$ coincides with $g$, the dimension of the positroid cell of $[A]$. 
\end{theorem}
The proof of the Theorem easily follows from  Lemmas~\ref{lem:d_dressed}, \ref{lem:d_dressed2} and Theorem~\ref{theo:exist}. 
We remark that the Darboux transformation automatically creates the Sato divisor on $\Gamma_0$.

\begin{remark}\textbf{$\DKP$ is the KP divisor on $\Gamma({\mathcal G}_{\mbox{\scriptsize red}}).$}
\label{rem:red_div} In our construction each KP divisor point either belongs to $\Gamma_0$ or to a copy of $\mathbb{CP}^1$ represented by a trivalent white vertex. Below we prove that the value of the normalized KP wave function is constant with respect to the spectral parameter on each component corresponding to a bivalent vertex; therefore the elimination of bivalent vertices doesn't affect either the value of the normalized KP wave function on the $\mathbb{CP}^1$ components corresponding to trivalent white vertices or the position of the KP divisor points. 
\end{remark}

\begin{remark}\textbf{Invariance of the KP divisor}
In \cite{AG2} we prove that $\hat \psi$ and $\DKP$ do not depend neither on the orientation of the network nor on the choice of the position of the Darboux points.
\end{remark}

\section{A system of vectors on the Le--network}
\label{sec:le}

In the previous Section we constructed the spectral curve $\Gamma$ associated with the given positroid cell $\S$ and the ordered set ${\mathcal K}$. The final goal of the main construction is the computation of the divisor corresponding to a given point $[A]\in S$. At this aim, in this and in the next sections, we 
first introduce a system of edge vectors on the Le--network and we use it to compute the values of the vacuum and dressed wave functions at all marked points of the curve.

The construction of a system of vectors at the edges of the Le--network is based on an algebraic procedure analogous to the one introduced in \cite{AG1} on the main cell. In \cite{AG1}, we related a specific representation of the rows of the banded matrix in $[A]\in Gr^{\mbox{\tiny TP}}(k,n)$ to the leading order behavior in the parameter $\xi$ of the vacuum wave function at double points and Darboux points. Here we use an excessive number of copies of $\mathbb{CP}^1$ to relate the system of edge vectors to the \textbf{exact behavior} of both the vacuum and the dressed wave functions at the marked points of $\Gamma$. Moreover, both the vectors and the wave functions satisfy linear relations at the vertices of the Le-network which are used to construct the vacuum and the dressed divisors.

In this section we first use the Le--network ${\mathcal N}$ to express each row of the RREF representative matrix as a linear combination with positive coefficients of a minimal number of some basic row vectors. This construction is a generalization of the Principal Algebraic Lemma in \cite{AG1} and easily follows from \cite{Pos}. Then, in Section~\ref{app:mainalgtheo} we present a recursive construction of these basic vectors, generalizing the corresponding theorem in  \cite{AG1}. In \cite{AG2}, we generalize the construction of edge vectors to any
planar trivalent bipartite network ${\mathcal N}$ in the disk representing a given point of $\Grkn$.

\subsection{Representation of the rows of the RREF matrix using the Le--tableau}
\label{sec:rowvectors}

The notations used in this section are 
the same as in the Appendix. We fix the positroid cell $\S\subset \Grkn$ and the planar bipartite trivalent acyclically oriented Le--graph $\mathcal G$ representing it, whereas ${\mathcal N}$ is the Le--network on $\mathcal G$ representing $[A]\in \S$. In the following $A$ is the representative matrix in RREF of $[A]$, $I$ is the lexicographically minimal base in the matroid $\mathcal M$ and $\bar I=[n]\backslash I$.
Each row $r$ of $A$ may be expressed as a linear combination 
with \textbf{positive coefficients} $c^r_{l_s}$ of $N_r+1$ vectors $E^{(r)}[l_s]$, $s\in[0,N_r]$, computed using ${\mathcal N}$. 
Each coefficient $c^r_{l_s}$ is the weight of the directed path from the boundary source $b_{i_r}$ to the internal white vertex $V_{i_r l_s}$. 
The absolute value of the $j$-th component of the vector $E^{(r)}[l_s]$ is the sum of the weights of all possible paths starting downwards at $V_{i_r l_s}$ and zig--zagging 
to the destination $b_j$, while its sign depends on the number of boundary sources passed before reaching the destination (see Lemma \ref{lemma:Asum}).

\begin{figure}
\includegraphics[scale=0.18]{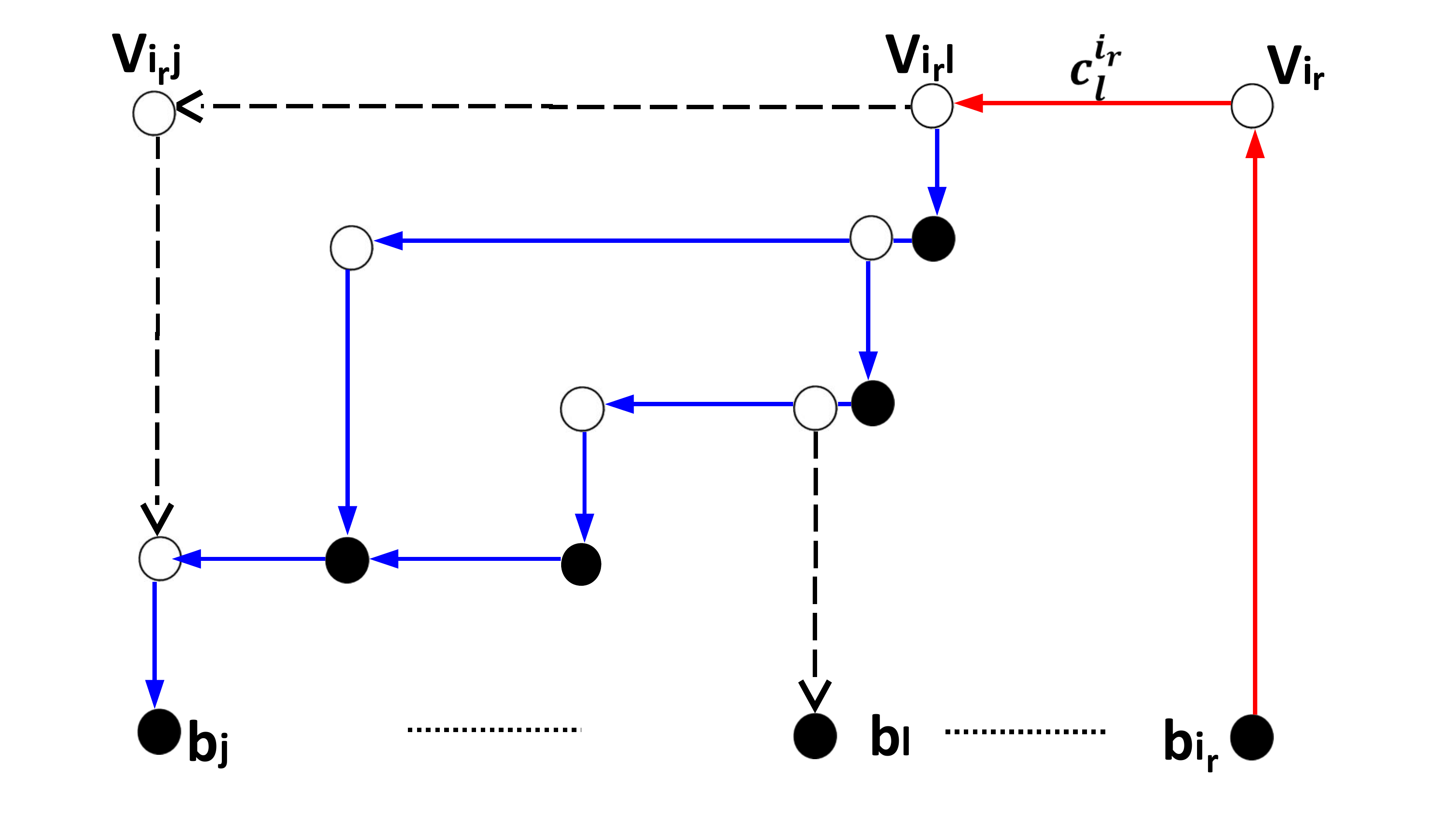}
\vspace{-0.5 truecm}
\caption{\footnotesize{\sl The coefficient $c^{i_r}_l$ is the weight of the directed path from the source $b_{i_r}$ to the white internal vertex $V_{i_r l}$, while $E^{(r)}_j [l] $, the $j$--th component of $E^{(r)} [l]$, has absolute value equal to the sum of the weights of all the directed paths which start downwards at the internal vertex $V_{i_r l}$ and zig--zag to the boundary sink $b_j$. The sign of $E^{(r)}_j [l] $ is equal to the number of boundary sources $b_{i_s}$ with $i_s \in ]i_r, j[$.}}
\label{fig:Eijvector}
\end{figure}

Let us fix $r\in [k]$ and let $i_r\in I$ be the corresponding pivot index. Let $l\in {\bar I}$ such that the box $B_{i_r l}$ has index $\chi^{i_r}_l= 1$ and let $V_{i_r l}$ be the corresponding white vertex in $\mathcal N_t$. Let us consider all directed paths $P$ starting at $b_{i_r}$, moving horizontally from $V_{i_r}$ to $V_{i_r l}$, moving downwards at $V_{i_r l}$ and then zigzagging towards any possible destination $j\in {\bar I}$. Necessarily $j\ge l$; moreover, for $l=j$, there is  exactly one such path from $b_{i_r}$ to $b_l$. For $j\in {\bar I}$ and $j\ge l$, let us define
\begin{equation}\label{eq:path}
\begin{array}{l}
\displaystyle {\mathcal P}^{(r)}_{l j}  \; = \; \{ \, P^{(r)}_{l j} : b_{i_r} \mapsto b_j \; : \; P^{(r)}_{l j} \mbox{ moves downwards at the white vertex } V_{i_r l} \, \},\\
{\bar {\mathcal P}}^{(r)}_{l j}  \; = \; \{ \, {\bar P}^{(r)}_{l j} : V_{i_rl} \mapsto b_j \; : \; {\bar P}^{(r)}_{l j} \mbox{ moves downwards at the white vertex } V_{i_r l} \, \},\\
\end{array}
\end{equation}
and denote by  $P_{i_r l}$ the path from the source  $b_{i_r}$ to the white vertex $V_{i_r l}$. Then
for any $P^{(r)}_{l j}\in {\mathcal P}^{(r)}_{l j}$ there exists a unique path ${\bar P}^{(r)}_{l j}\in {\bar {\mathcal P}}^{(r)}_{l j}$ such that
$P^{(r)}_{l j} = P_{i_r l} \sqcup {\bar P}^{(r)}_{l j}$.
Define
\begin{equation}\label{eq:coeff}
c^i_l \equiv \prod\limits_{e \in P_{i_rl}} w_e.
\end{equation}
Then the weight of the path $P^{(r)}_{l j}$ is
$w(P^{(r)}_{l j} )= \prod_{e\in P^{(r)}_{l j}} w_e = c^i_l \,\left( \prod_{e\in {\bar P}^{(r)}_{l j}} w_e \right)= c^i_l \,w({\bar P}^{(r)}_{l j})$
and, for any $j\in {\bar I}$, $j>i_r$, the matrix entry in reduced row echelon form as in (\ref{eq:ARREF}), may be re--expressed as
$$
A^{r_i}_j \, = \, (-1)^{\sigma_{i_r j}} \sum_{s=1}^{N_r} c^i_{l_s} \sum_{{\bar P}^{(r)}_{l_sj}\in {\bar {\mathcal P}}^{(r)}_{l_s j}} w({\bar P}^{(r)}_{l_s j})
= (-1)^{\sigma_{i_r j}} \sum_{l\le j} c^i_l \sum_{{\bar P}^{(r)}_{lj}\in {\bar {\mathcal P}}^{(r)}_{l j}} w({\bar P}^{(r)}_{l j}),
$$
where $\sigma_{i_r j}$ is the number of boundary sources in $]i_r,j[$. In \cite{Pos} the matrix elements $A^{r_i}_j$ are computed using columns instead of rows.

Finally, for any $r\in [k]$, $i_r\in I$, $l\in {\bar I}$, $l>i_r$, and such that $\chi^{i_r}_l=1$, let us define the row vector
$E^{(r)}[l] = ( E^{(i)}_1 [l],\dots, E^{(i)}_n [l])$, with
\begin{equation}\label{eq:rowentry}
E^{(r)}_j [l] \; = \; \left\{ \begin{array}{ll}
\displaystyle 0 & \quad \mbox{ if } j<l \mbox{ or } j\in I, \\
 (-1)^{\sigma_{i_r j}} \sum\limits_{{\bar P}^{(r)}_{l j}\in {\bar {\mathcal P}}^{(r)}_{lj}} w({\bar P}^{(r)}_{l j})
&\quad \mbox{ if } l,j\in {\bar I} \mbox{ and } j\ge l.
\end{array}
\right.
\end{equation}
In the expression above the entry $E^{(r)}_j [l]=0$ if there is no path moving downwards at the vertex $V_{i_r l}$ and reaching destination $b_j$.
Finally we associate to the boundary source a vector 
\begin{equation}\label{eq:pivec}
E^{(r)}[i_r] = (0,\dots, 0 ,1 ,0,\dots,0),
\end{equation}
with non zero entry in the $i_r$-th column. Then the following Lemma holds true. 

\begin{lemma}\label{lemma:Asum}
Let $A$ be the reduced row echelon form matrix representing a given point in the matroid stratum $\S\subset \Grkn$ and let $D$ be the corresponding Le--diagram. Let $E^{(r)}[i]$, $E^{(r)} [l]$, $c^i_l$ as in (\ref{eq:coeff}), (\ref{eq:rowentry}) and (\ref{eq:pivec}), with $r\in [k]$, $i_r\in I$, $l\in {\bar I}$, such that the box $B_{i_r l}$ is filled by 1. Then, the $r$--th row of $A$ is
\begin{equation}\label{eq:Avec}
A[r] = E^{(r)}[i_r] + \sum\limits_{s=1}^{N_r} c^r_{l_s} E^{(r)} [l_s],
\end{equation}
where the sum runs over the indexes $l_s \in {\bar I}$ such that the index in (\ref{eq:chiindex}) is $\chi^{i_r}_{l_s}=1$.
\end{lemma}

The proof is trivial and is omitted. Let us remark that the vectors $E^{(r)}[i_r]$, $E^{(r)} [l_s]$, $s\in N_r$, form a minimal system of vectors to represent the $r$-th row of the reduced row echelon matrix. 

\smallskip

\subsection{Recursive construction of the row vectors $E^{(r)} [l]$ using the Le--diagram}
\label{app:mainalgtheo}

In Theorem \ref{theo:mainalg} we provide a recursive representation
for the above system of vectors using the Le--tableau starting from the last row ($r=k$) and moving upwards till the first row ($r=1$). 

For any fixed $r\in [k]$, we first complete the system of vectors introduced in the previous section to a convenient basis in $\mathbb{R}^n$, given by the rows of $n\times n$ matrix ${\hat E}^{(r)}$. For $r=k$, we use the canonical basis in $\mathbb{R}^n$. We pass from the basis associated to the $r$-th row to the one associated to the $(r-1)$--th row applying a transition matrix $C^{[r-1, r]}$: ${\hat E}^{(r-1)} = C^{[r-1,r]} {\hat E}^{(r)}$.
Each transition matrix $C^{[r-1, r]}$ keeps track of empty and non--empty boxes of the $r$--th row of the Le-tableau, and is upper triangular 
by definition. The choice of signs in $C^{[r-1, r]}$ in (\ref{eq:Ckk-1}) keeps track of the sign changes when passing a pivot column. We also complete each set of coefficients $c^r_l$ in (\ref{eq:coeff}) to a row vector ${\hat c}^{(r)}$, with $l$ indexing the column position.

This construction is a corollary to the Lindstr\"om lemma in the case of acyclic graphs and, for points $[A]\in Gr^{\mbox{\tiny TP}} (k,n)$, is also the combinatorial version of the recursive algebraic construction in \cite{AG1} for a different choice of the basic vectors and therefore of the representative matrix $A$. In \cite{AG2} we give general rules to construct well defined systems of vectors on any network and for any orientation. 

Theorem~\ref{theo:mainalg} is used in section \ref{sec:vvw} to define a vacuum edge wave function and its dressing on ${\mathcal N}$. Such vacuum edge wave function (respectively its dressing)
\begin{enumerate}
\item rules the behavior of the vacuum wave function (respectively the dressed wave function) on $\Gamma$ at all marked points;
\item satisfies linear relations at the inner vertices of ${\mathcal N}$ which are used to detect the position of the vacuum 
(resp. dressed) divisor points.
\end{enumerate}

\begin{theorem}{\bf (The recursive algebraic construction)}\label{theo:mainalg}
Let $[A] \in \S\subset \Grkn$ with pivot set $I = \{ 1\le i_1 < i_2 < \cdots < i_k \le n\}$. Let $T$ and ${\mathcal N}$ respectively be the Le--tableau and its acyclically oriented bipartite Le--network. For any $r\in [k]$, $j\in \bar I$,  let the index $\chi^{i_r}_j$ be as in (\ref{eq:chiindex}).
Let us define the following collections of $n\times n$ invertible matrices ${\hat E}^{(r)}$, $n\times n$ transition matrices $C^{[r-1,r]}$ and row vectors ${\hat c}^{(r)}$, $r\in [k]$, associated to $\mathcal N$:
\begin{enumerate}
\item ${\hat E}^{(k)}$ is the $n\times n$ identity matrix  and we denote its row vectors as ${\hat E}^{(k)} [l]$, $l\in [n]$;
\item For $r\in [k]$, define the $n\times n$ transition matrix $C^{[r-1,r]}$ as follows:
\begin{equation}\label{eq:Ckk-1}
C^{[r-1,r],l}_{j} = \left\{ \begin{array}{cl}
\delta^l_j,              & \quad \mbox{ if } l\in [1, i_r[, \quad j\in [n],\\
-1                       & \quad \mbox{ if } l=j=i_r,\\
-{\hat w}^{(r)}_{i_r j}, & \quad \mbox{ if } l=i_r,\quad j\in ]i_r, n] \mbox{ and }  \quad  \chi^{i_r}_j= 1,\\
-{\hat w}^{(r)}_{l j},   & \quad \mbox{ if } l,j\in {\bar I}\cap ]i_r, n], \quad j\ge l,  \quad  \chi^{i_r}_j\chi^{i_r}_l = 1,\\
-\delta^l_j ,            & \quad \mbox{ if } l\in {\bar I}\cap ]i_r, n], \quad  \chi^{i_r}_l = 0, \quad j\in [n],\\
-\delta^l_j ,            & \quad \mbox{ if } l\in I\cap ]i_r, n],  \quad j\in [n],\\
 0                       & \quad\mbox{ otherwise, }            
\end{array}
\right.
\end{equation}
where 
\begin{enumerate}
\item for $l,j\in {\bar I}\cap ]i_r, n]$, $j\ge l$, ${\hat w}^{(r)}_{l j}$ is the weight of the directed horizontal path from the black vertex
$V'_{i_r l}$ to the white vertex $V_{i_r j}$. In particular ${\hat w}^{(r)}_{l l}=1$;
\item for $l=i_r$ and $j\in {\bar I}$, ${\hat w}^{(r)}_{i_r j}$ is the weight of the directed horizontal path from the  vertex $V_{i_r}$ to the white vertex $V_{i_r j}$.
\end{enumerate}
\item The matrices ${\hat E}^{(r-1)}$ are recursively computed as $r$ decreases from $k$ to 2, by
\begin{equation}\label{eq:ECE}
{\hat E}^{(r-1)} = C^{[r-1,r]} {\hat E}^{(r)}.
\end{equation} 
\item For any $r\in [k]$, the $l$--element of the row vector ${\hat c}^{(r)}$, $l\in [n]$, is
\begin{equation}\label{eq:coelastrow}
{\hat c}^{(r)}_l = \left\{ \begin{array}{cl} 
1                      &\quad \mbox{ if } l=i_r,\\
\displaystyle c^r_l    &\quad \mbox{ if } \chi^{i_r}_ l =1,\\
0                      &\quad \mbox{ otherwise } ,
\end{array}\right.
\end{equation}
with $c^r_l$ as in (\ref{eq:coeff}).
\end{enumerate}

Then
\begin{enumerate}
\item For $r\in [k-1]$ ${\hat E}^{(r)}$ is the $n\times n$ matrix such that
\begin{equation}\label{eq:Er}
{\hat E}^{(r)} [l] = \left\{ \begin{array}{ll}
{\hat E}^{(k)} [l]        & \quad \mbox{ if } l\le i_r, \\
(-1)^{\sigma_{sr}} A[s]   & \quad \mbox{ if } l=i_s, \mbox{ and } s\in ]r,k],\\
E^{(r)} [l]               & \quad \mbox{ if } l\in {\bar I}\cap ]i_r,n], \mbox{ and }   \chi^r_l = 1, \\
\end{array}
\right.
\end{equation} 
with $E^{(r)} [l]$ as in (\ref{eq:rowentry}) and $\sigma_{sr} = \# \{ i_t\in I ,  r\le t<s\}$;
\item For any $r\in [k]$, 
\begin{equation}\label{eq:Asum}
A[ r] =\sum\limits_{l=1}^n {\hat c}^r_l {\hat E}^{(r)} [l] \equiv E^{(r)} [i_r] + \sum\limits_{s=1}^{N_r} c^r_{l_s}  E^{(r)} [l_s],
\end{equation}
where the second sum is over the indexes such that $\chi^{i_r}_{l_s} =1$.
\end{enumerate}
\end{theorem}

\begin{remark}{\bf Changing the orientation of the graph}
Any change of orientation in ${\mathcal N}$ corresponds to the composition of elementary changes of orientation \cite{Pos}, either corresponding to a change of the base in the matroid ${\mathcal M}$ of $[A]$ or to a change of orientation in a cycle of the graph. We postpone to \cite{AG2} a thorough explanation of how the system of vectors on any given network representing $[A]$ is effected by such elementary changes and the proof that both the normalized dressed wave function and the effective KP divisor are not affected by them.
\end{remark}

\begin{proof}
(\ref{eq:Asum}) follows immediately from Lemma~\ref{lemma:Asum}, (\ref{eq:coelastrow}) and (\ref{eq:Er}). 

The first statement in (\ref{eq:Er}), ${\hat E}^{(r)} [l] = {\hat E}^{(k)} [l]$ for $l\le i_r$, $r\in [k]$ is trivial by definition of the transition matrix (\ref{eq:Ckk-1}). The remaining statements in (\ref{eq:Er}) follow by induction in the index $r$ as it decreases from $k$ to 1. Indeed, for $r=k$, the transition matrix $C^{[k-1,k]}$
\begin{enumerate}
\item Leaves invariant all canonical basis vectors ${\hat E}^{(k)} [l]$ for all $l< i_k$;
\item Transforms the canonical basis vector ${\hat E}^{(k)}[i_k]$ to $-A[k]$, if $l=i_k$;
\item Changes the sign of the canonical basis vector,  ${\hat E}^{(k-1)} [l] = -{\hat E}^{(k)} [l]$, if $l\in {\bar I} \cap ]i_k, n]$ and $\chi^{i_k}_l=0$;
\item Acts untrivially only if $l\in {\bar I} \cap ]i_k, n]$ and $\chi^{i_k}_l=1$. In such case, the components of ${\hat E}^{(k-1)}[l]$ are transformed to
\[
{\hat E}^{(k-1), l}_j = \left\{ \begin{array}{cl} 
0                        & \quad \mbox{ if } j<  l,\\
-1                       & \quad \mbox{ if } j=l,\\
0                        & \quad \mbox{ if } j> l \mbox{ and } \chi^{i_k}_j=0,\\
-{\hat w}^{(k)}_{l j}    & \quad \mbox{ if } j> l \mbox{ and } \chi^{i_k}_j=1.
\end{array}
\right.
\]
\end{enumerate}
It is straightforward to verify that ${\hat E}^{(k-1)} [l] = E^{(k-1)} [l]$, if $l=i_{k-1}$ or $\chi^{i_{k-1}}_l=1$, for $l\ge i_{k-1}$.
Indeed if both  $\chi^{i_{k-1}}_l =1$ and  $\chi^{i_{k}}_l =1$, the white vertex $V_{i_{k-1} l}$ is joined to the black vertex $V'_{i_k l}$ by an edge of 
weight 1 so that by definition ${\hat E}^{(k-1)} [l] = E^{(k-1)} [l]$. If $\chi^{i_{k-1}}_l =1$ and  $\chi^{i_{k}}_l =0$, the white vertex $V_{i_{k-1} l}$ is joined to the boundary vertex $b_{l}$ by an edge of unit weight and
${\hat E}^{(k-1)} [l] = E^{(k-1)} [l]$.

Let us now suppose to have verified (\ref{eq:Er}), for any $s\in [r, k]$ and let's verify it for $s=r-1$. By definition
\begin{equation}\label{eq:rechatE}
{\hat E}^{(r-1)} = C^{[r-1, r]} {\hat E}^{(r)} =C^{[r-1, r]} C^{[r, r+1]} {\hat E}^{(r+1)} = \cdots = \left(\prod\limits_{s=r}^k C^{[s-1, s]}\right) {\hat E}^{(k)}.
\end{equation}
Let $l\in {\bar I} \cap ]i_{r-1},n]$ be fixed. If $\chi^{i_{r-1}}_l=0$, then there is no contribution to $A[r-1]$ from the vector ${\hat E}^{(r-1)} [l]$ since the coefficient ${\hat c}^{r-1}_l=0$.
Suppose now that $\chi^{i_{r-1}}_l=1$ and consider the set $S= \{ s\in [r,k] \, : \, \chi^{i_s}_l=1\, \}$. If $S= \emptyset$, then the white vertex $V_{i_{r-1} l}$ is joined directly to the boundary sink $b_{l}$ by an edge of unit weight and 
\[
{\hat E}^{(r-1)} [l]= (-1)^{k - r+1} {\hat E}^{(k)} [l]=(-1)^{k - r+1} E^{(k)}[l] =E^{(r-1)} [l].
\]
Otherwise, let ${\bar s} = \mbox{ min } S$. In this case moving downwards from the white vertex $V_{i_{r-1} l}$ the first black vertex that we meet is 
$V'_{i_{{\bar s}} l}$ and such edge has unit weight. Then, using (\ref{eq:rechatE}), we have
\[
{\hat E}^{(r-1)} [l]= (-1)^{{\bar s} - r+1} {\hat E}^{({\bar s})} [l]=(-1)^{{\bar s} - r+1} E^{({\bar s})}[l]=E^{(r-1)} [l].
\]
Finally let $l=i_{r}$. In such case
\[
{\hat E}^{(r-1)} [i_r] =  C^{[r-1, r]} [i_r] {\hat E}^{(r)} = -\sum_{j=1}^n {\hat c}^r_j {\hat E}^{(r)}[j] =-A[r].
\]
since, by definition $C^{[r-1, r],i_r}_j =- {\hat c}^r_j$ and ${\hat c}^r_j$ is non zero if and only if $j=i_r$ or $j\in {\bar I} \cap ]i_r,n]$ is such that $\chi^{i_r}_j=1$. 
\end{proof}

\begin{remark}\label{rem:comp1}\textbf{Comparison with the algebraic construction in \cite{AG1}} 
Lemma \ref{lemma:Asum} and Theorem \ref{theo:mainalg} generalize the algebraic construction in \cite{AG1} for points in $Gr^{\mbox{\tiny TP}} (k,n)$  to any point in $\Grkn$. The main difference is in the choice of the representative matrix: here $A$ is the RREF matrix while in \cite{AG1} we use a matrix in banded form.
\end{remark}

\begin{example}\label{ex:rec}
Let us apply Theorem \ref{theo:mainalg} to the Le--network in Figure \ref{fig:bipex1} representing Example \ref{ex:gr492}. ${\hat E}^{(4)}= Id_{9\times 9}$, the only non--zero coefficients associated to the forth row of $A$ are ${\hat c}^4_7=1$, ${\hat c}^4_8=w_{78}$ and clearly $A[4] = {\hat c}^4{\hat E}^{(4)}$. Then
\[
{\hat E}^{(3)}\equiv C^{[3,4]}= \left( \begin{array}{cc}
I_{6\times 6 } & 0_{6\times 3} \\
0_{3\times 6}               & \begin{array}{ccc}
-1 & -w_{78} & 0 \\
0  & -1      & 0\\
0  & 0 & -1
\end{array}
\end{array}
\right),
\]
where $0_{i\times j}$ is the $(i\times j)$ null matrix, while $I_{l\times l}$ is the $(l\times l)$ identity matrix. The third row coefficient vector is
${\hat c}^3 = (\, 0\, ,0\, ,0\, ,1\, ,w_{45}\, ,w_{45}w_{46}\, ,0\, ,w_{45}w_{46}w_{48}\, ,w_{45}w_{46}w_{48}w_{49}\, )$,
and $A[3] = {\hat c}^3 {\hat E}^{(3)}.$
The transition matrix from the third to the second row $C^{[2,3]}$, the new basis vectors ${\hat E}^{(2)}$ and the coefficient vector ${\hat c}^2$ of the second row respectively are
{\small\[
C^{[2,3]} = \left( \begin{array}{cc}
I_{3\times 3 } & 0_{3\times 6} \\
0_{6\times 3}               & \begin{array}{ccccccc}
-1 & -w_{45} & -w_{45}w_{46} & 0 & -w_{45}w_{46}w_{48} & -w_{45}w_{46}w_{48}w_{49}\\
0  & -1      & -w_{46}       & 0 & -w_{46}w_{48}       & -w_{46}w_{48}w_{49}\\
0  & 0       & -1            & 0 & -w_{48}             & -w_{48}w_{49}\\
0  &  0      & 0             &-1 & 0                   &  0 \\
0  & 0       & 0             & 0 & -1                  & -w_{49}\\
0  &  0      & 0             &0  & 0                   &  -1 
\end{array}
\end{array}
\right),
\]
\[
{\hat E}^{(2)} = C^{[2,3]}{\hat E}^{(3)} =\left( \begin{array}{cc}
Id_{3\times 3 } & 0_{3\times 6} \\
0_{6\times 3}               & \begin{array}{ccccccc}
-1 & -w_{45} & -w_{45}w_{46} & 0 & w_{45}w_{46}w_{48} & w_{45}w_{46}w_{48}w_{49}\\
0  & -1      & -w_{46}       & 0 & w_{46}w_{48}       & w_{46}w_{48}w_{49}\\
0  & 0       & -1            & 0 & w_{48}             & w_{48}w_{49}\\
0  &  0      & 0             & 1 & w_{78}             &  0 \\
0  & 0       & 0             & 0 & 1                  & w_{49}\\
0  &  0      & 0             &0  & 0                  &  1 
\end{array}
\end{array}
\right),
\]}
${\hat c}^2 = (\, 0\, ,1\, ,w_{23}\, ,0\, ,w_{23}w_{25}\, ,0\, ,0\, ,0\, ,0\, )$,
and $A[2] = {\hat c}^2 {\hat E}^{(2)}.$
Finally, the transition matrix from the second to the first row $C^{[1,2]}$, the new basis vectors ${\hat E}^{(1)}$ and the coefficient vector ${\hat c}^1$
of the first row, respectively, are
{\small\[
C^{[1,2]} = \left( \begin{array}{cc}
\begin{array}{ccccc}
1  & 0  & 0       & 0 & 0\\
0  & -1 & -w_{23} & 0 & -w_{23}w_{25} \\
0  & 0  & -1      & 0 & -w_{25}       \\
0  & 0  & 0       & -1&  0            \\
0  & 0  & 0       &  0&  -1            \\
\end{array}
& 0_{5\times 4}\\
0_{5\times 1} & -Id_{4\times 4}
\end{array}
\right),
\]
\[\resizebox{\textwidth}{!}{$
{\hat E}^{(1)} = C^{[1,2]}{\hat E}^{(2)} =\left( \begin{array}{ccccccccc}
1 & 0 & 0       & 0 & 0            & 0 & 0 & 0 & 0\\
0 &-1 & -w_{23} & 0 & w_{23}w_{25} & w_{23}w_{25}w_{46} & 0 & -w_{23}w_{25}w_{46}w_{48} & -w_{23}w_{25}w_{46}w_{48}w_{49}\\
0 &0  & -1 & 0 & w_{25} & w_{25}w_{46} & 0 & -w_{25}w_{46}w_{48} & -w_{25}w_{46}w_{48}w_{49}\\
0 &0  & 0  & 1 & w_{45} & w_{45}w_{46} & 0 & -w_{45}w_{46}w_{48} & -w_{45}w_{46}w_{48}w_{49}\\
0 &0  & 0  & 0 & 1      & w_{46}       & 0 & -w_{46}w_{48}       & -w_{46}w_{48}w_{49}\\
0 &0  & 0  & 0 & 0      & 1            & 0 & -w_{48}             & -w_{48}w_{49}\\
0 &0  & 0  & 0 & 0      & 0            & -1& -w_{78}             & 0\\
0 &0  & 0  & 0 & 0      & 0            & 0 & -1             & -w_{49}\\
0 &0  & 0  & 0 & 0      & 0            & 0 & 0             & -1
\end{array}
\right),$}
\]}
${\hat c}^1 = (\, 1\, ,0\, ,0\, ,0\, ,w_{15}\, ,w_{15}w_{16}\, ,0\, ,0\, ,w_{15}w_{16}w_{19}\, )$,
and $A[1] = {\hat c}^1 {\hat E}^{(1)}.$
\end{example}

\section{Proof of Theorem \ref{theo:exist} on $\Gamma$}
\label{sec:proof}

Throughout this Section we fix the KP regular soliton data, {\sl i.e.} $n$ ordered real phases ${\mathcal K} = \{\kappa_1 < \cdots < \kappa_n\}$ and a point $[A]\in\S\subset \Grkn$, $\dim\S=g$. $I=\{i_1,\ldots,i_k\}$ is the lexicographically minimal base of ${\mathcal M}$ (the pivot set for any matrix $A$ representing $[A]$); $T$ and ${\mathcal N}$ are the Le--tableau and the acyclically oriented bipartite  Le--network representing $[A]$, see \cite{Pos} and, also, Appendix \ref{app:TNN}. Finally, $\Gamma$ is the curve associated to the Le--graph as in Construction~\ref{def:gamma}. 

The idea of the proof of Theorem~\ref{theo:exist} is to extend the vacuum wave function from $\Gamma_0$ to $\Gamma$ using the network to define the values of the wave function at the double points and at the Darboux points so that it admits a good meromorphic extension, and the divisor has the desired properties. At this aim, we first modify the Le-network by adding an edge and an internal vertex in correspondence of each Darboux point in $\Gamma$, and transform the boundary sources into boundary sinks (see Figure~\ref{fig:vertex_added}). 
Then, in Section \ref{sec:vvw} we use the algebraic construction of the previous section to define a system of edge vectors and a vacuum and a dressed edge wave functions on the modified Le--network. In Section~\ref{sec:vacuumdiv} we use the linear relations satisfied by the edge wave functions at the internal vertices to assign to each trivalent white vertex a pair of real numbers, which we call vacuum and dressed network divisor numbers respectively. Finally, in Section \ref{sec:vacdiv}, we extend the normalized vacuum and dressed edge wave functions to $\Gamma$ in such a way that the network divisor numbers become the local coordinates of the pole divisors of such wave functions, and we prove that the vacuum divisor satisfies the reality and regularity conditions of Theorem \ref{theo:exist}.

\begin{figure}
  \centering
  \includegraphics[width=0.6\textwidth]{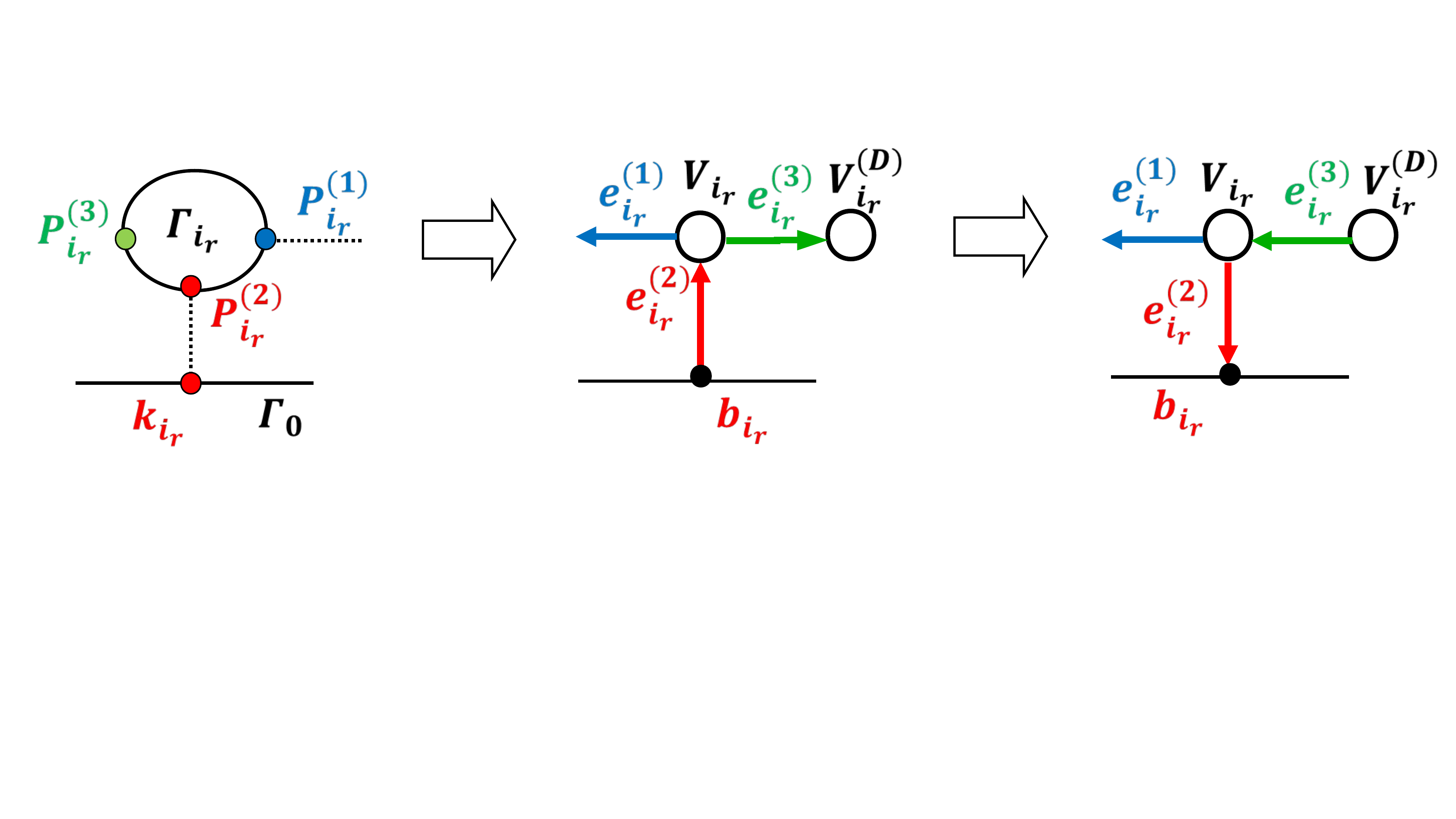}
  \vspace{-2.4 truecm}
  \caption{\footnotesize{\sl The modified graph ${\mathcal N}^{\prime}$ is obtained adding an edge $e^{(3)}_{i_r}$ at each pivot vertex $V_{i_r}$and a white vertex $V_{i_r}^{(D)}$. The orientation in ${\mathcal N}^{\prime}$ is the same as in the initial network ${\mathcal N}$ except at the edge $e^{(2)}_{i_r}$.}}
	\label{fig:vertex_added}
\end{figure}

\begin{definition}\label{def:NTprime}\textbf{The planar oriented trivalent bipartite network ${\mathcal N}^{\prime}$:} 
Denote ${\mathcal N}^{\prime}$ the network obtained from ${\mathcal N}$ adding a unit edge $e^{(3)}_{i_r}$ at each pivot vertex $V_{i_r}$ in the position 
corresponding to the Darboux point $P^{(3)}_{i_r}\in \Gamma_{i_r}$ and let $V^{(D)}_{i_r}$ be the white vertex at the other end of $e^{(3)}_{i_r}$. 
Such move corresponds to Move (M2) - unicolored edge contraction/uncontraction and still represents the same point in the Grassmannian \cite{Pos}. Then orient all edges in ${\mathcal N}^{\prime}$ as in ${\mathcal N}$ except the edges $e^{(2)}_{i_r}$, which point from $V_{i_r}$ to $b_{i_r}$,  and $e^{(3)}_{i_r}$, $r\in [k]$, which point from $V^{(D)}_{i_r}$ to $V_{i_r}$, for any $r\in [k]$ (see Figure~\ref{fig:vertex_added}). 
\end{definition}
For an example of transformation from $\mathcal N$ to $\mathcal N'$ see Figure~\ref{fig:NtoNprime1}.

\begin{figure}
\includegraphics[width=0.45\textwidth]{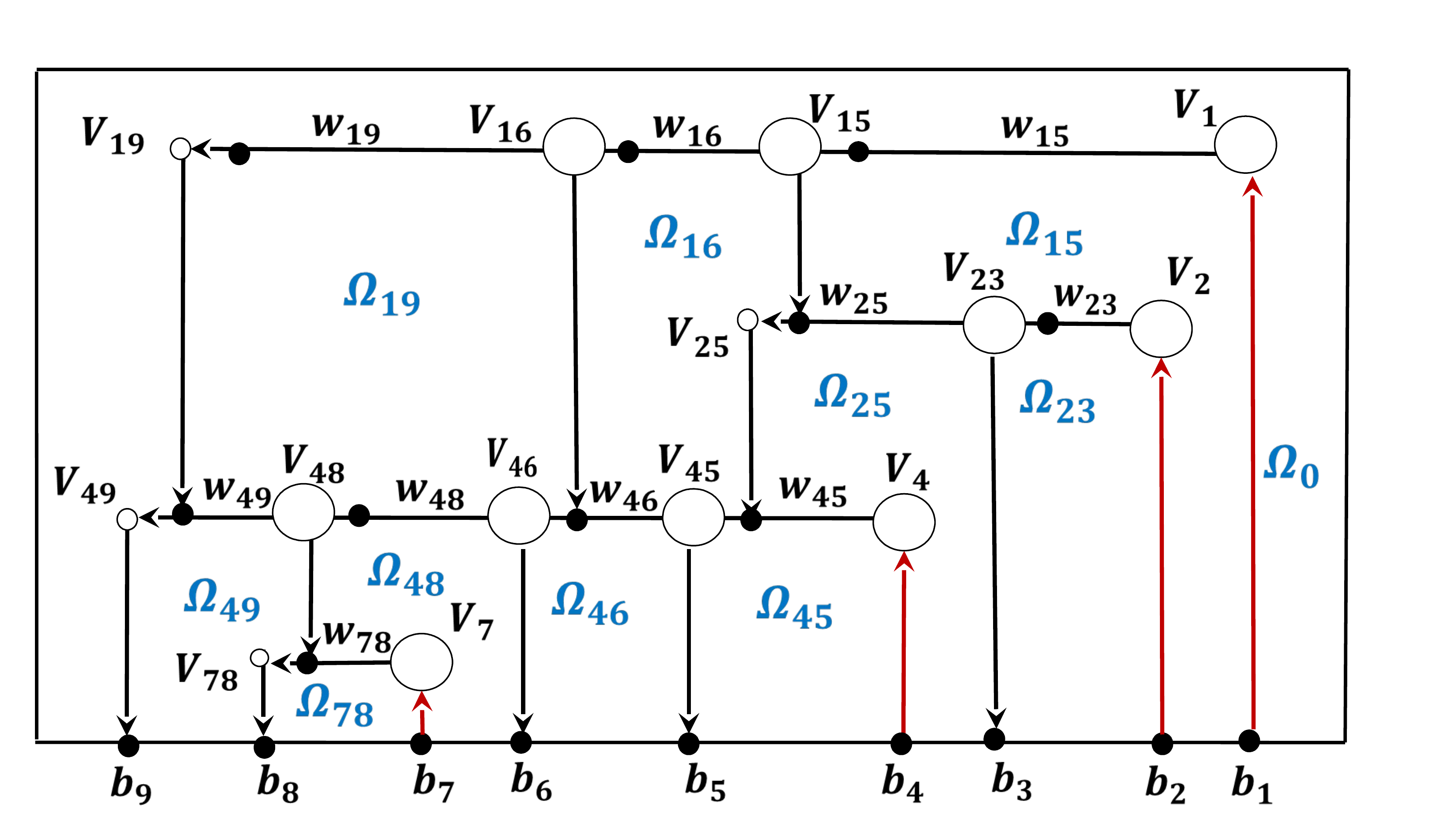}
\includegraphics[width=0.45\textwidth]{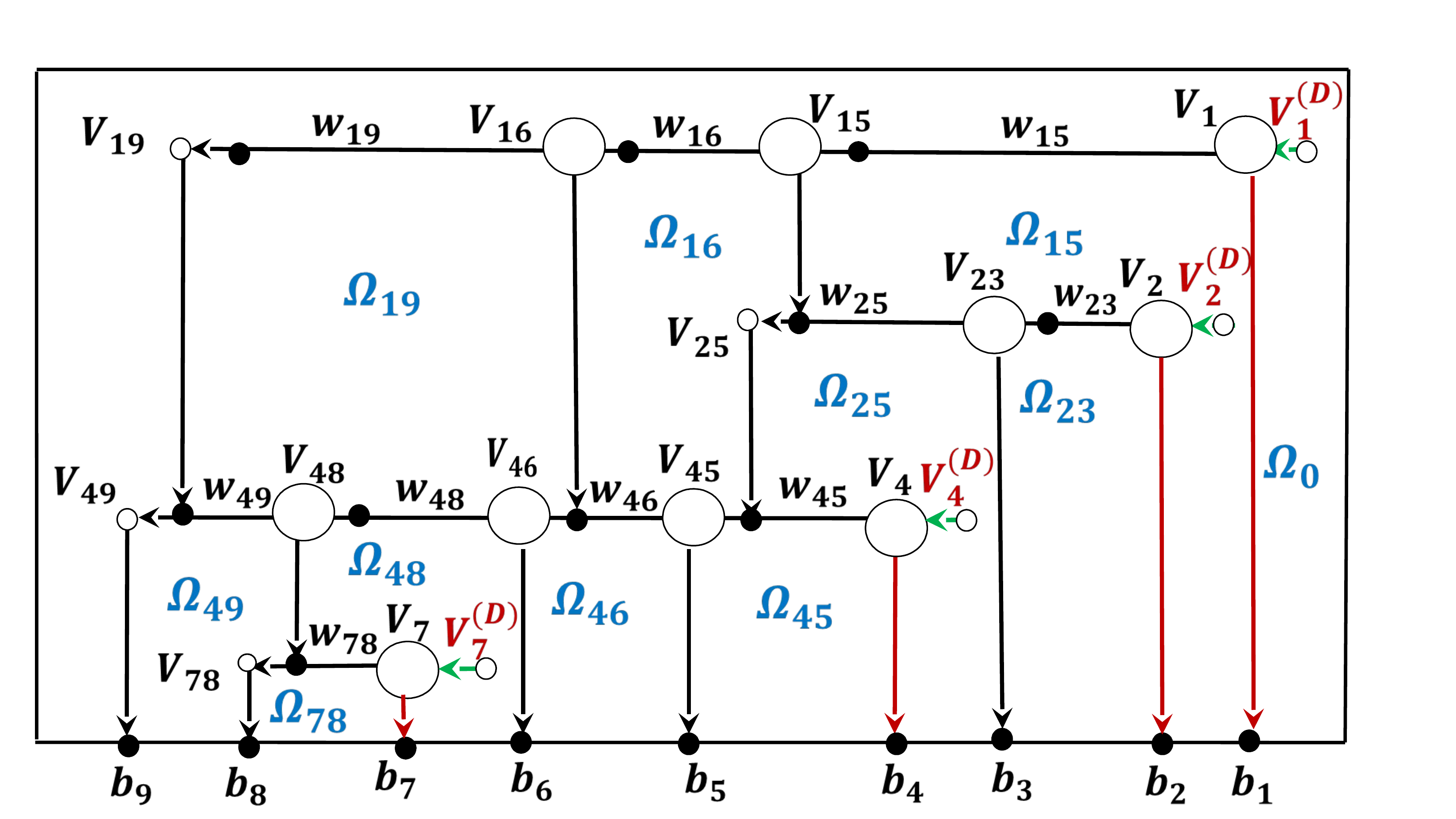}
\caption{\footnotesize{\sl The transformation from the Le-Network $\mathcal N$ (left) to the modified Le-network $\mathcal N'$ (right) for Example \ref{ex:gr492} (see also Figure \ref{fig:lediagex2}).}
\label{fig:NtoNprime1}}
\end{figure}

\begin{remark}\label{rem:datasec5} \textbf{Data and notations.}
From now on we use the modified network ${\mathcal N}^{\prime}$ with the orientation as in Definition \ref{def:NTprime} and, in addition to the assumptions made in the beginning of this Section, we settle the following notations:
\begin{enumerate}
\item $N_r$, $r\in [k]$, is the number of filled boxes in the row $r$ of the Le-tableaux of $A$, 
$\chi^i_j\in\{0,1\}$ is the index of the box $B_{i,j}$ (see (\ref{eq:Nij}) and (\ref{eq:chiindex}));
\item The pivot indexes are denoted $1 \le i_1 < \cdots < i_k\le n$ and, for any $r\in[ k]$,  $1\le j_1 < j_2 < \cdots < j_{N_r} \le n$ are the non--pivot indexes of the boxes $B_{i_r, j_s}$ of index $\chi^{i_r}_{j_s}=1$, $s\in [N_r]$;
\item The index $j_0\equiv 0$ is associated to the vertices, $V_{i_r 0} \equiv V_{i_r}$, and to any quantity referring to them;
\item The edges $e^{(m)}_{i_rj_l}$, $m\in [3]$, are enumerated with the indexes of its incident white vertex $V_{i_r j_l}$, $r\in [k]$ and $l\in [0,N_r]$ (see Figures \ref{fig:markedpoints} and \ref{fig:vertex_added}). Finally, the vertices $V_{i_r j_{N_r}}$ have two edges which we label $e^{(2)}_{i_rj_l}$ and $e^{(3)}_{i_rj_l}$;
\item The families of matrices ${\hat E}^{(r)}$ and vectors ${\hat c}^{(r)}$, $r\in [k]$, are as in Theorem \ref{theo:mainalg};
\item ${\mathfrak E}_\theta (\vec t) = (e^{\theta_1(\vec t)}, \dots, e^{\theta_n(\vec t)})$, where $\theta_j(\vec t) = \sum_{l\ge 1} \kappa_j^{l} t_l$ and $\vec t =(t_1=x,t_2=y,t_3=t, t_4,\dots)$ are the KP times, and 
 $\prec \cdot, \cdot\succ $ denotes the usual scalar product;
\item On each component of $\Gamma$, $\zeta$ is the coordinate of Definition \ref{def:loccoor}. 
\end{enumerate}
In the statements of the Theorems and in the Definitions below we shall not explicitly mention the above data. 
\end{remark}

\subsection{Vacuum and dressed edge wave functions on the modified Le--network}\label{sec:vvw}

We now define a system of vectors  ${\mathfrak E}^{(m)}_{i_r j_{l} }$ on the edges of ${\mathcal N}^{\prime}$ using the vectors introduced in Section~\ref{sec:le}.
First of all, in ${\mathcal N}^{\prime}$, for each fixed $r\in[k]$ and $l\in[0,N_r]$ we assign the row vector $E^{(r)}[j_l]$ to the vertical edge $e^{(2)}_{i_r j_l}$ at the white vertex $V_{i_r j_l}\in {\mathcal N}^{\prime}$: ${\mathfrak E}^{(2)}_{i_r j_{l} }\equiv E^{(r)}[j_l]$. We then assign a vector also at each horizontal edge using linear relations at inner black and white vertices (see  Definition~\ref{def:vertexwf}). The vacuum edge wave function $\Phi^{(m)}_{ij}(\vec t)$ at the edge $e^{(m)}_{ij}\in {\mathcal N}^{\prime}$ is just the product of the edge vector ${\mathfrak E}^{(m)}_{i j}$ with the vector ${\mathfrak E}_\theta (\vec t)$: $\Phi^{(m)}_{ij}(\vec t) = \prec {\mathfrak E}^{(m)}_{i j}, {\mathfrak E}_\theta (\vec t) \succ$. 

\begin{figure}
  \centering
  {\includegraphics[width=0.44\textwidth]{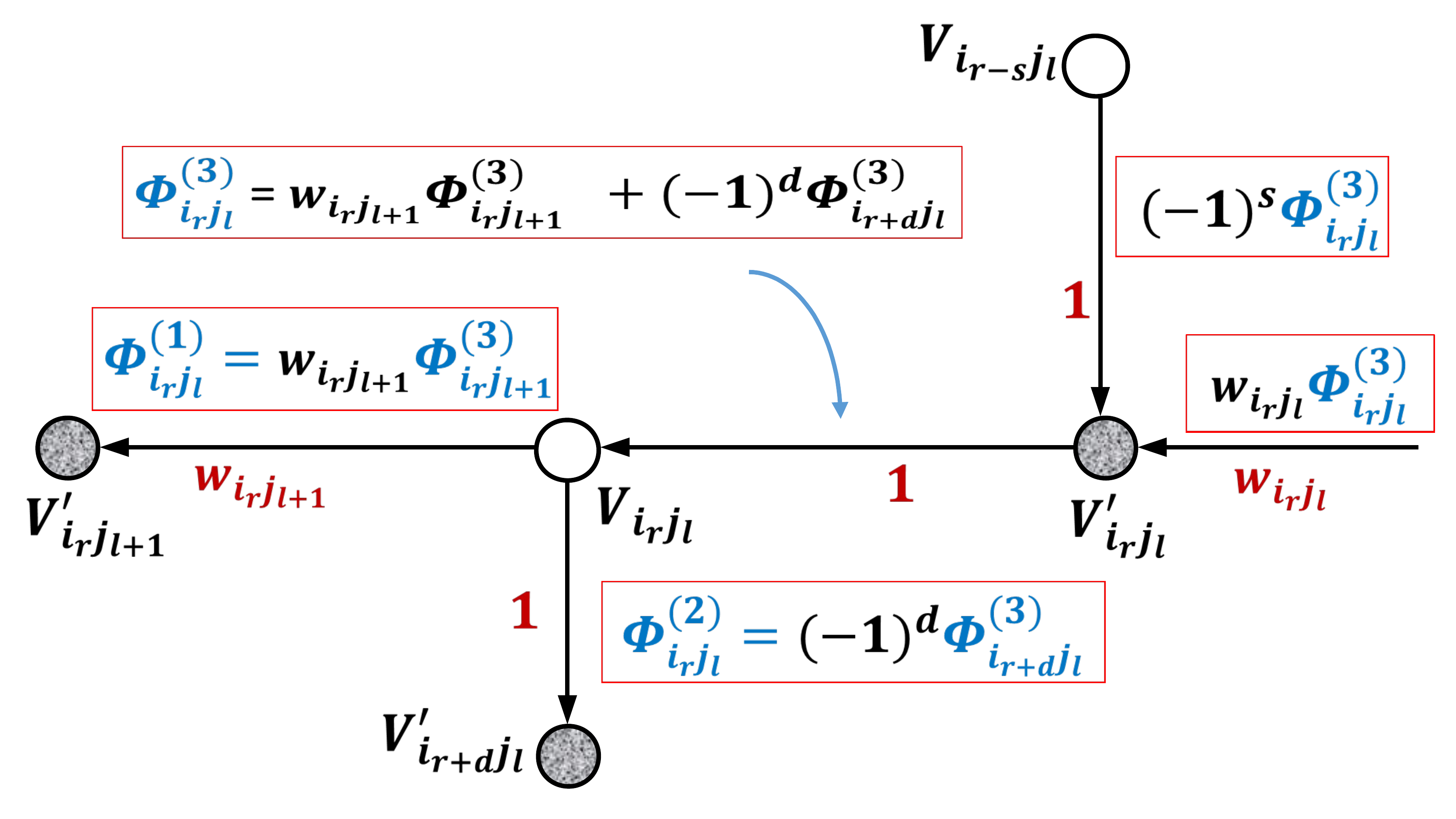}}
  \hfill
	{\includegraphics[width=0.44\textwidth]{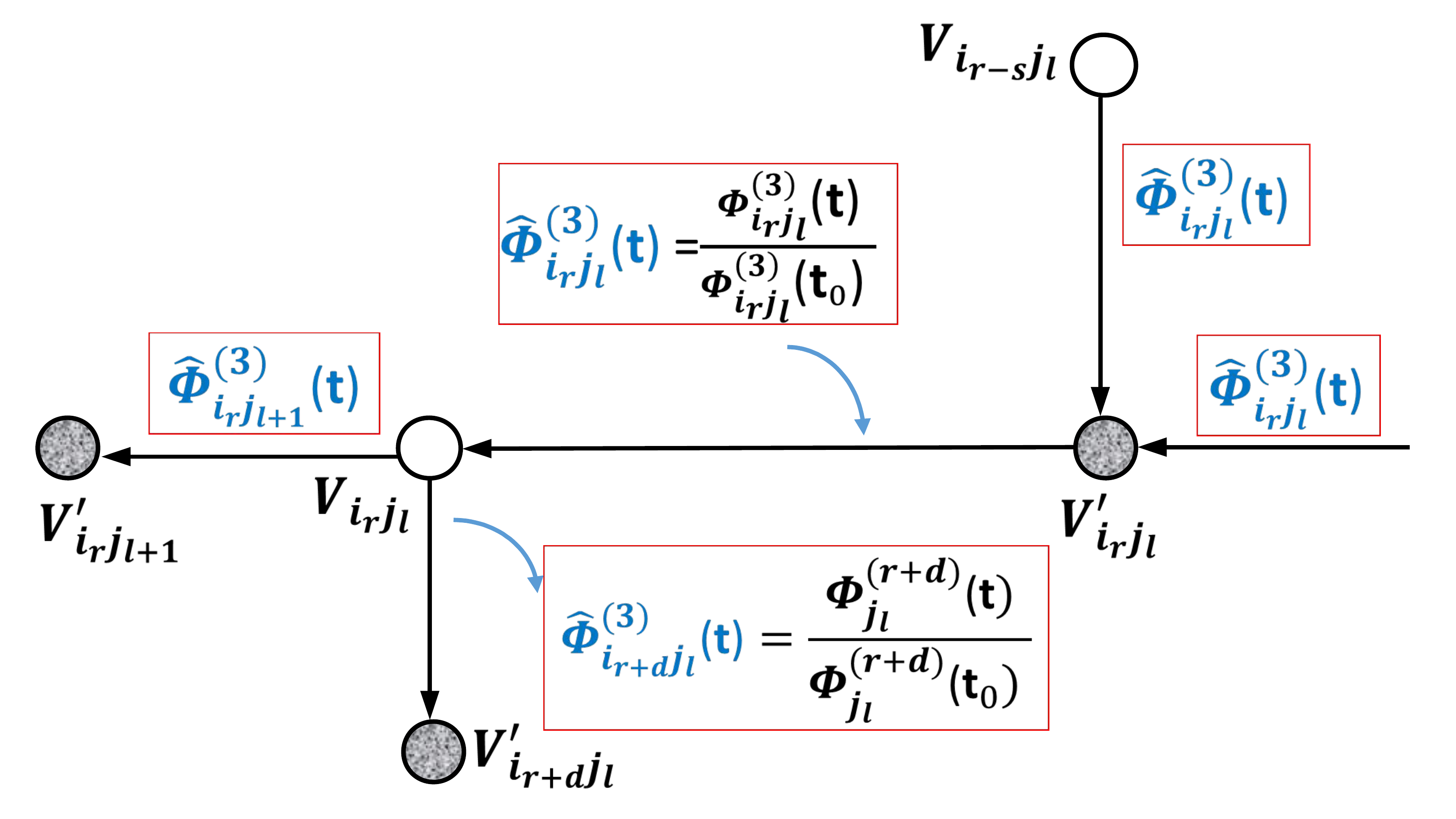}}
	
	\vspace{.8 truecm}
	
	{\includegraphics[width=0.44\textwidth]{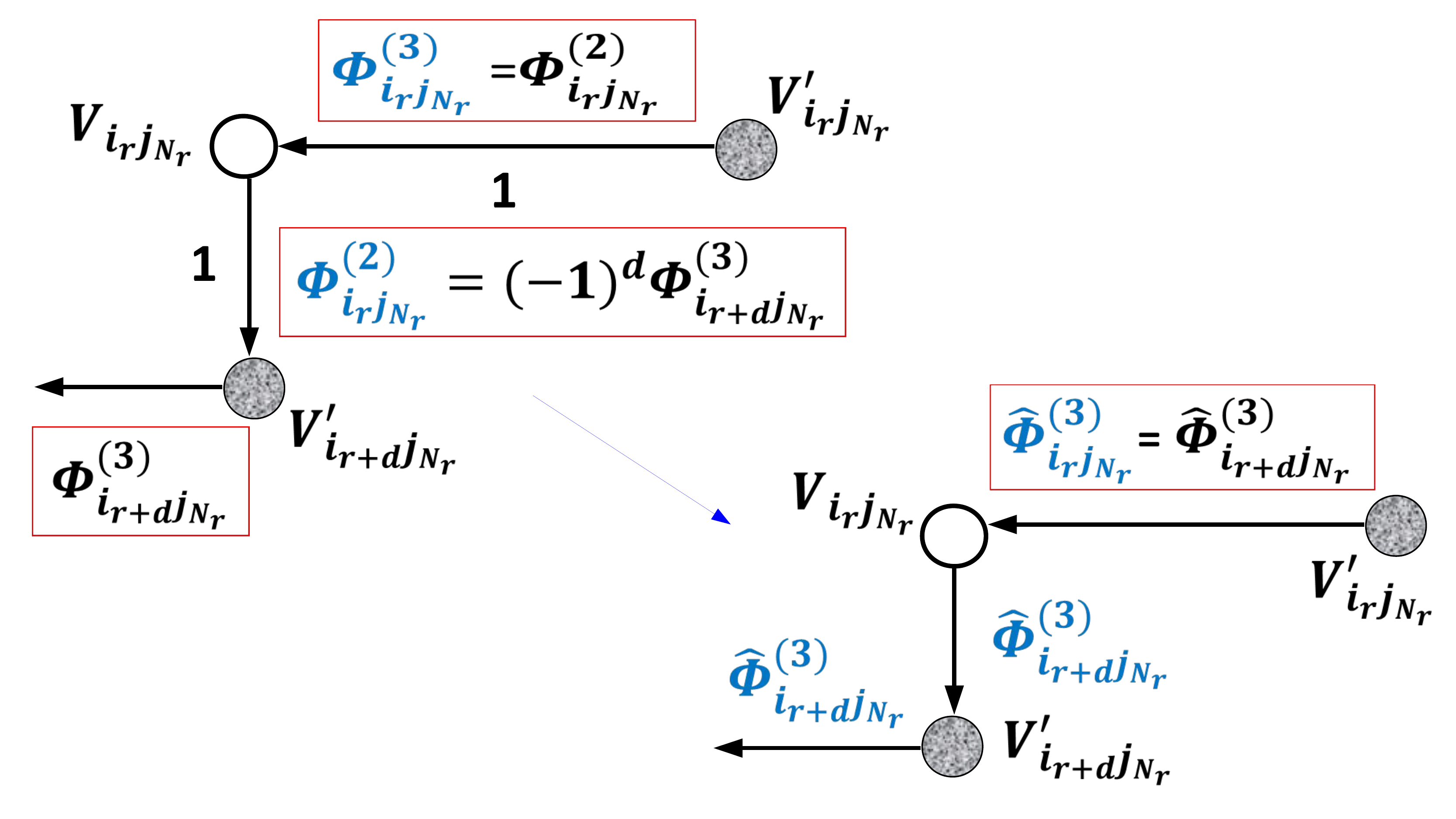}}
  \hfill
	{\includegraphics[width=0.44\textwidth]{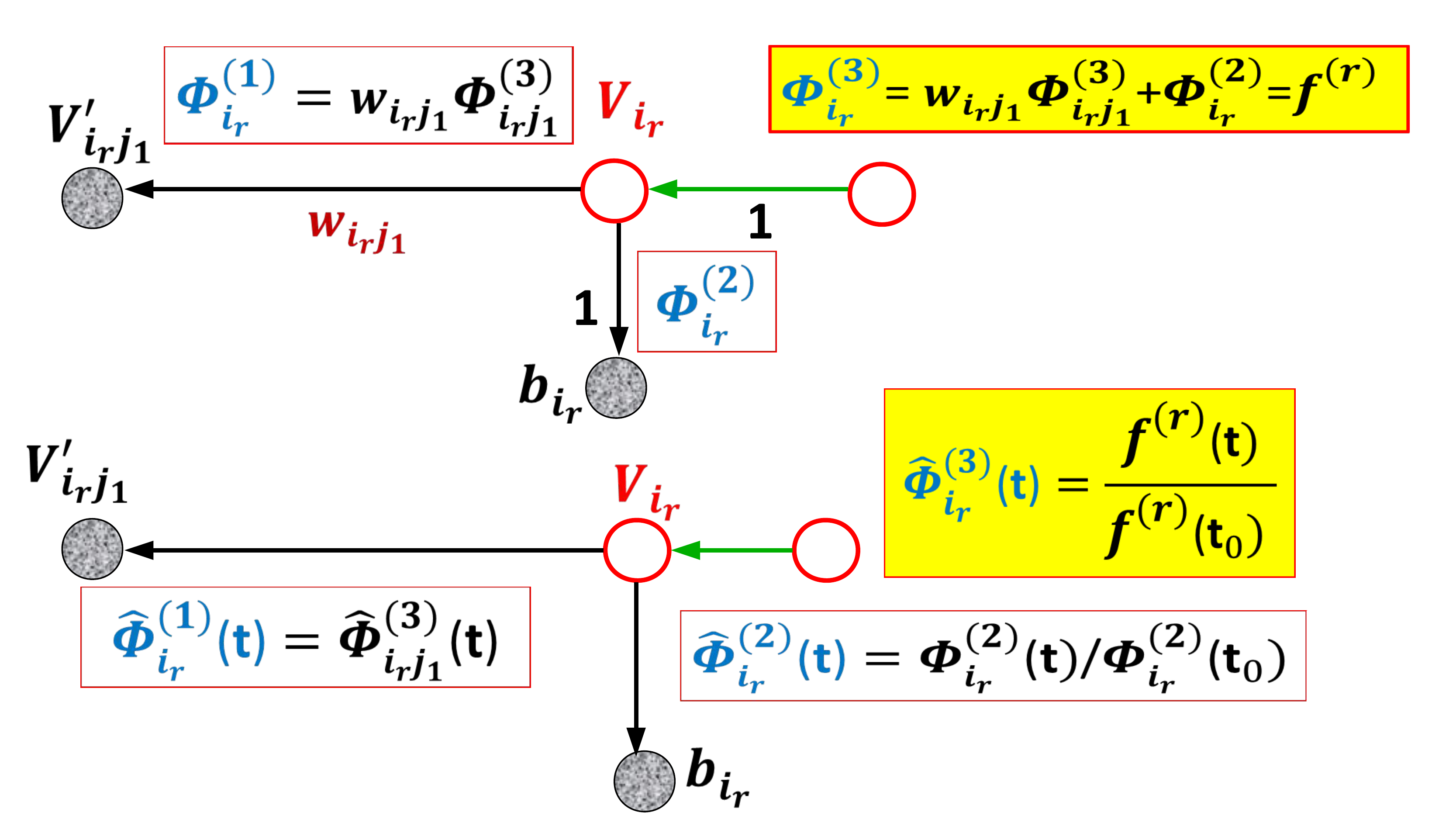}}
	\vspace{-.2 truecm}
  \caption{\footnotesize{\sl The linear relations for the vacuum edge wave function $\Phi^{(m)}_{i_rj_l} (\vec t)$ at white and black vertices [top,left]. The normalized v.e.w. $\hat \Phi^{(m)}_{i_rj_l} (\vec t)$ takes the same value at all the edges at a given black vertex $V^{\prime}_{i_rj_l}$ and different values at the edges at a given trivalent white vertex $V_{i_rj_l}$ [top right]. The same at the bivalent white vertex $V_{i_rj_{N_r}}$ [bottom,left] and at the vertex $V_{i_r}$, $i_r\in I$ [bottom, right]. Also the dressed edge wave function satisfies the same linear relations.}}
	\label{fig:vertexrule}
\end{figure}
Finally, by construction, the edge vector assigned to $e^{(3)}_{i_r}$ is the $r$-th row of the RREF matrix, for any $r\in [k]$, and, therefore, the vacuum edge wave function at $e^{(3)}_{i_r}$ coincides with one of the heat hierarchy solutions generating the Darboux transformation (see Lemma \ref{lemma:Phi3}).

\begin{definition}\label{def:vertexwf}{\bf Edge vectors (e.v.) and vacuum edge wave function (v.e.w.) on ${\mathcal N}^{\prime}$.}
Let the soliton data and the notations be fixed as in Remark \ref{rem:datasec5}. To each edge of ${\mathcal N}^{\prime}$, we associate an edge vector (e.v.) and a vacuum edge wave function (v.e.w.) depending on $\vec t$ as follows:
\begin{enumerate}
\item For any $r\in [k]$, to the vertical edge $e^{(2)}_{i_r}$ joining the white vertex $V_{i_r}$ to the boundary vertex $b_{i_r}$, we assign e.v. and v.e.w.
\begin{equation}\label{eq:Phi_i}
{\mathfrak E}^{(2)}_{i_r} = {\hat E}^{(k)} [i_r],\quad\quad \Phi^{(2)}_{i_r} (\vec t) \equiv \prec {\hat E}^{(k)} [i_r], {\mathfrak E}_\theta (\vec t)\succ = e^{\theta_{i_r}(\vec t)};
\end{equation} 
\item For any fixed $r\in [k]$ and $l\in [N_r]$, to the vertical edge $e^{(2)}_{i_rj_l}$ at the white vertex $V_{i_r j_l}$ we assign
\begin{equation}\label{eq:Phi_2}
{\mathfrak E}^{(2)}_{i_r j_l} = {\hat E}^{(r)} [j_l],
\quad\quad \Phi^{(2)}_{i_r j_l} (\vec t) \equiv \prec {\mathfrak E}^{(2)}_{i_r j_l}, {\mathfrak E}_\theta (\vec t)\succ;
\end{equation} 
\item For any $r\in [k]$, to the horizontal edge $e^{(3)}_{i_rj_{N_r}}$ joining the black vertex $V^{\prime}_{i_r j_{N_r}}$ to the white vertex $V_{i_r j_{N_r}}$ of the $r$-th row we assign:  
\begin{equation}\label{eq:Phi_3N}
{\mathfrak E}^{(3)}_{i_r j_{N_r}} = {\mathfrak E}^{(2)}_{i_r j_{N_r}},
\quad\quad \Phi^{(3)}_{i_r j_{N_r}} (\vec t) \equiv \prec {\mathfrak E}^{(3)}_{i_r j_{N_r}}, {\mathfrak E}_\theta (\vec t)\succ 
= \Phi^{(2)}_{i_r j_{N_r}}(\vec t)\end{equation}
\item For any $r\in [k]$, $l\in [0,N_r-1]$, to the horizontal edge $e^{(1)}_{i_rj_{l}}$ joining the white vertex $V_{i_r j_{l}}$ to the black vertex $V^{\prime}_{i_r j_{l+1}}$ we assign
\begin{equation}\label{eq:Phi_1}
{\mathfrak E}^{(1)}_{i_r j_{l}} = w_{i_r j_{l+1}} {\mathfrak E}^{(3)}_{i_r j_{l+1}},
\quad\quad \Phi^{(1)}_{i_r j_l} (\vec t) \equiv \prec {\mathfrak E}^{(1)}_{i_r j_{l}}, {\mathfrak E}_\theta (\vec t)\succ =w_{i_r j_{l+1}}\Phi^{(3)}_{i_r j_{l+1}}(\vec t), 
\end{equation}
where $w_{i_r j_1}$ is the weight of $e^{(1)}_{i_rj_{l}}$. Here $j_0=0$;
\item For any $r\in [k]$, $l\in [0,N_r-1]$, to the edge  $e^{(3)}_{i_rj_l}$ joining the black vertex $V^{\prime}_{i_r j_{l}}$ to the white vertex $V_{i_r j_{l}}$ we assign
\begin{equation}\label{eq:Phi_3}
{\mathfrak E}^{(3)}_{i_r j_{l}} = {\mathfrak E}^{(1)}_{i_r j_{l}} + {\mathfrak E}^{(2)}_{i_r j_{l}},
\quad\quad \Phi^{(3)}_{i_r j_l} (\vec t) \equiv \prec {\mathfrak E}^{(3)}_{i_r j_{l}}, {\mathfrak E}_\theta (\vec t)\succ = \Phi^{(1)}_{i_r j_l} (\vec t)+\Phi^{(2)}_{i_r j_l} (\vec t), 
\end{equation}
{\sl i.e.} the sum of the contributions from the edges to the left and below $V_{i_r j_{l}}$.
\end{enumerate}
\end{definition}

We illustrate Definition \ref{def:vertexwf} in Figure \ref{fig:vertexrule}.
\begin{remark}\textbf{Edge vectors and oriented paths in ${\mathcal N}^{\prime}$.} It is easy to check that, for any given edge $e$, the absolute value of the $j$--th component of the vector at $e$ is simply the sum of the weights of all the paths starting at $e$ and having destination $b_j$ and the sign of such component depends only on the number of boundary sources in $]i_r,j[$, where $r\in [k]$ is the biggest index such that there exists a walk from the Darboux source vertex $V^{(D)}_{i_r}$ to $e$. Recall that the acyclic orientation of $\mathcal N^{\prime}$ is associated to the lexicographically minimal base $I$; therefore there is a canonical way to count the number of Darboux sources for any path from an internal edge $e$ to a destination $b_j$, and it is also straightforward to check that all paths from $e$ to $b_j$ are assigned the same sign.
\end{remark}

\begin{definition}\label{def:dress_edgewf}{\bf The dressed edge wave function (d.e.w.) on ${\mathcal N}^{\prime}$.}
In the setting of Definition \ref{def:vertexwf}, for any fixed $r\in [k]$ and $l\in [0,N_r]$, we define the dressed edge wave function $\Psi^{(m)}_{i_r j_l} (\vec t)$  on the edge $e^{(m)}_{i_r j_l}\in {\mathcal N}^{\prime}$ as the dressing of the v.e.w.  $\Phi^{(m)}_{i_r j_l} (\vec t)$
\begin{equation}
\Psi^{(m)}_{i_r j_l} (\vec t) = {\mathfrak D} \Phi^{(m)}_{i_r j_l} (\vec t),
\end{equation}
where ${\mathfrak D}$ is the Darboux transformation associated to the given soliton data $({\mathcal K}, [A])$.
\end{definition}

\begin{remark}
In \cite{AG2}, we generalize the construction of edge vectors (and edge wave functions) to any network representing $[A]$ in Postnikov class and adapt the construction in \cite{Pos,Tal2} to express the components of the edge vectors in case of non acyclic orientations. In the latter case one has to deal with sums over infinite number if paths.
\end{remark}

By construction, the system of edge vectors ${\mathfrak E}^{(m)}_{i_r,j_l}$, the v.e.w. $\Phi^{(m)}_{i_r j_l} (\vec t)$ and the d.e.w. $\Psi^{(m)}_{i_r j_l} (\vec t)$ on ${\mathcal N}^{\prime}$ solve the following linear system at the inner vertices of ${\mathcal N}^{\prime}$ under suitable boundary conditions.

\begin{lemma}\label{lemma:lin_sys}\textbf{The linear system in ${\mathcal N}^{\prime}$}
Let $G_j (\vec t)$, $j\in [n]$, be smooth functions in $\vec t$. Then there exists a unique function $G_{e} (\vec t)$ defined on the edges $e$ of ${\mathcal N}^{\prime}$ satisfying for all $\vec t$:
\begin{enumerate}
\item For any $r\in [k]$ and $l\in [0,N_r]$, on the unit edge $e^{(3)}_{i_r j_l}$ pointing in at the white vertex $V_{i_rj_l}$,  define $G_{e^{(3)}_{i_r,j_l}}$ as the sum of the values of $G_{e}$ on the edges $e=e^{(1)}_{i_rj_l},e^{(2)}_{i_rj_l}$ pointing out at $V_{i_rj_l}$:
\[
G_{e^{(3)}_{i_r,j_l}} (\vec t) =G_{e^{(1)}_{i_rj_l}} (\vec t) +G_{e^{(2)}_{i_rj_l}} (\vec t);
\]
\item For any $r\in [k]$ and $l\in [0,N_r-1]$, on the horizontal edge $e^{(1)}_{i_r j_{l}}$ of weight $w_{i_r j_l}$ pointing in at the black vertex $V^{\prime}_{i_r j_{l+1}}$, define $G_{e^{(1)}_{i_r,j_l}}$ as
\[
G_{e^{(1)}_{i_r,j_l}} (\vec t)= w_{i_r j_{l+1}} G_{e^{(3)}_{i_r,j_{l+1}}} (\vec t),
\]
where $e^{(3)}_{i_r,j_{l+1}}$ is the unique edge pointing out at $V^{\prime}_{i_r j_{l+1}}$;
\item For any $r\in [k]$, the unit vertical edge $e^{(2)}_{i_r}$ joins the white vertex $V_{i_r}$ to the boundary vertex $b_{i_r}$. Define
\[
G_{e^{(2)}_{i_r j_l}} (\vec t) = G_{i_r} (\vec t);
\]
\item For any $r\in [N_r]$, if the unit vertical edge $e^{(2)}_{i_r, j_l}$ joins the white vertex $V_{i_r,j_l}$ to the boundary vertex $b_{j_l}$, define
\[
G_{e^{(2)}_{i_r j_l}} (\vec t) = G_{j_l} (\vec t).
\]
Otherwise, the black vertex $V^{\prime}_{i_{{\bar s}} j_l}$, with ${\bar s} = \mbox{ min } S$, with $S= \{ s\in [r+1,k] \, : \, \chi^{i_s}_{j_l}=1\, \}\not = \emptyset$ as in the proof of Theorem \ref{theo:mainalg}, is the ending vertex of $e^{(2)}_{i_r j_l}$, $d=d(i_r j_l) = \bar s - r $ is the number of pivot indexes in the interval $]i_r,j_l[$ and define
\[
G_{e^{(2)}_{i_r j_l}} (\vec t) = (-1)^d G_{e^{(3)}_{i_{r+d} j_l}} (\vec t).
\]
\end{enumerate}
In particular, if we assign the boundary conditions $G_j (\vec t) \equiv \Phi^{(2)}_{j} (\vec t)$, (respectively $G_j (\vec t) \equiv \Psi^{(2)}_{j} (\vec t)$), for all $j\in [n]$, then the edge function coincides with the v.e.w. of Definition \ref{def:vertexwf} (respectively with the d.e.w. of Definition \ref{def:dress_edgewf}).
\end{lemma}

\begin{proof}
First let ${\mathcal N}^{\prime}$ correspond to an irreducible positroid cell $\S \subset \Grkn$ of dimension $g$.
Then, the number of variables in the linear system defined in the above Lemma is equal to the number of edges, $3g+2k$ ($g+k$ vertical edges and $2g+k$ horizontal edges). 
Any trivalent black vertex carries two equations, while each bivalent black vertex carries one condition. There are $g$ internal black vertices, and $g-n+k$ of them are trivalent. Therefore the total number of equations at black vertices is $2g-n+k$.  The total number of equations at white vertices is $g+k$, since the total number of internal white vertices is $g+2k$, but the univalent Darboux vertices on ${\mathcal N}^{\prime}$ do not carry any extra condition. So, we may freely impose a value to  $n$ variables (edges).

The presence of an isolated boundary sink $b_j$ implies the addition of an internal univalent vertex joined to $b_j$: we have a new variable (the univalent edge) and no extra condition. The presence of an isolated boundary source $b_i$ implies the addition of a bivalent vertex $V_i$ and of an univalent Darboux vertex $V^{(D)}_i$: we have two variables and one equation.

The system is well-defined and the solution is unique because the network is acyclic. Finally the system can be solved recurrently.
\end{proof}

\begin{remark}
We remark that the linear relations satisfied by the edge vectors, the vacuum edge wave function and its dressing at the black and white vertices (see Definition \ref{def:vertexwf} and Figure \ref{fig:vertexrule}) are of the same type as those imposed by momentum conservation at trivalent vertices of on--shell diagrams in \cite{AGP1,AGP2} (formulas (4.54) and (4.55) in \cite{AGP1}). 
\end{remark}

Comparing Theorem \ref{theo:mainalg} and Definition \ref{def:vertexwf}, it is not difficult to prove that the e.v. at the Darboux edge $e^{(3)}_{i_r}$ coincides with the $r$--th row of the RREF matrix $A$ (see equation \ref{eq:Er0} below); therefore the v.e.w. at the same edge is $f^{(r)}(\vec t)$ as required in Theorem~\ref{theo:exist}. 

\begin{lemma}\label{lemma:Phi3}
Let the soliton data and the notations be fixed as in Remark \ref{rem:datasec5}. Let the e.v. and v.e.w. be as in Definition \ref{def:vertexwf}. Then,
\begin{enumerate}
\item for any $r\in [k]$ and $l\in [N_r]$, the edge vector at the edge joining $V_{i_r j_l}$, $V^{\prime}_{i_r j_l}$ is 
\begin{equation}\label{eq:Erj}
{\mathfrak E}^{(3)}_{i_r j_l} = \sum_{s=l}^{N_r} {\hat c}^r_{j_s} {\hat E}^{(r)} [j_s];
\end{equation} 
\item for any $r\in [k]$ the edge vector at the edge $e^{(3)}_{i_r }$ is the $r$--th row of the RREF matrix representing $[A]\in \Grkn$ and the v.e.w. is the $r$--th generator of the dressing transformation associated to the soliton data $({\mathcal K}, [A])$
\begin{equation}\label{eq:Er0}
{\mathfrak E}^{(3)}_{i_r} = A[r],
\quad\quad
\Phi^{(3)}_{i_r} (\vec t) = f^{(r)} (\vec t).
\end{equation}
Therefore the d.e.w. satisfies $\Psi^{(3)}_{i_r} (\vec t)\equiv 0$, for all $r\in [k]$ and for all $\vec t$.
\end{enumerate}
\end{lemma}

The proof of the Lemma easily follows comparing (\ref{eq:Erj}) and (\ref{eq:Er0}) with (\ref{eq:coeff}), (\ref{eq:ECE}), (\ref{eq:coelastrow}) and (\ref{eq:Asum}) in Theorem \ref{theo:mainalg} and with (\ref{eq:Phi_i})--(\ref{eq:Phi_3}) in Definition \ref{def:vertexwf}.

As remarked in the proof of Lemma~\ref{lemma:lin_sys}, the edge vectors and vacuum and dressed edge wave functions on ${\mathcal N}^{\prime}$ may be recursively computed using Theorem \ref{theo:mainalg} starting from the last row of the corresponding Le--diagram ($r=k$) and moving up decreasing the index $r$ till $r=1$. For simplicity, we illustrate the algorithm for the edge vectors only (see also Figure \ref{fig:vertexrule} for the edge wave functions).

\textbf{Algorithm for the edge vectors:}\label{alg:intver}
For any $r\in [k]$:
\begin{enumerate}
\item If $N_r=0$, {\sl i.e.} all the boxes of the $r$--th row of the Le--diagram contain zeros, assign to the white vertex $V_r$ the vector ${\hat E}^{(r)}[i_r]\equiv {\hat E}^{(k+1)}[i_r]$ and proceed to (\ref{item:pippo});
\item Otherwise:
\begin{enumerate}
\item\label{item:pluto} Start from the leftmost white vertex of the line, $V_{i_r j_{N_r}}$ and 
assign to the edge joining $V_{i_r j_{N_r}}$ and $V^{\prime}_{i_r j_{N_r}}$ the edge vector
\[
{\mathfrak E}^{(3)}_{i_r j_{N_r}} = {\hat E}^{(r)} [j_{N_r}]  
\]
\item For any $l$ from $N_r-1$ to 1, assign to the edge joining $V'_{i_r j_{l}}$ and $V_{i_r j_{l}}$, the vector 
\[
{\mathfrak E}^{(3)}_{i_r j_l} = w_{i_r j_{l+1}} {\mathfrak E}^{(3)}_{i_r j_{l+1}} + {\hat E}^{(r)}[ j_l].
\]
\item At the white vertex $V_r$ assign the vector
\[
{\mathfrak E}^{(3)}_{i_r 0} = w_{i_r j_{1}} {\mathfrak E}^{(3)}_{i_r j_1} + {\hat E}^{(r)} [i_r] =A[r],
\]
and go to (\ref{item:pippo}).
\end{enumerate}
\item\label{item:pippo} If $r=1$, just end. Otherwise set the counter to $r-1$ and repeat the whole procedure.
\end{enumerate}

\begin{example}\label{ex:vertexvector}
We illustrate such procedure for Example \ref{ex:gr492} (see Figure \ref{fig:NtoNprime1}[right]). At the horizontal edge joining $V_{45}$ and $V^{\prime}_{45}$ we assign the edge vector 
\[
{\mathfrak E}^{(3)}_{45} =E^{(3)}[5] + w_{46} {\mathfrak E}^{(3)}_{46} =(0,0,0,0,1,w_{46},0,-w_{46}w_{48},-w_{46}w_{48}w_{49}),
\]
since ${\mathfrak E}^{(3)}_{46} =E^{(3)}[6]+w_{48} {\mathfrak E}^{(3)}_{48} =E^{(3)}[6]+w_{48} ( {\hat E}^{(3)}[8]+ w_{49} {\hat E}^{(3)}[9]),$ and the edge wave function
\[
\Phi^{(3)}_{4\, 5} (\vec t) = e^{\theta_5 (\vec t)} +w_{46}e^{\theta_6 (\vec t)}-w_{46}w_{48}e^{\theta_8 (\vec t)}-w_{46}w_{48}w_{49}e^{\theta_9 (\vec t)}.
\]
At the horizontal edge joining $V_{25}$ and $V^{\prime}_{25}$ we associate the edge vector ${\mathfrak E}^{(3)}_{25} = E^{(2)}[5]=-{\mathfrak E}^{(3)}_{45} $ (compare with Example \ref{ex:rec}) and the edge wave function $\Phi^{(3)}_{25} (\vec t) =-\Phi^{(3)}_{45} (\vec t)$ for all times $\vec t$.
\end{example}

\subsection{The vacuum and dressed network divisors}\label{sec:vacuumdiv}

We now assign two real numbers, a vacuum network divisor number $\gvac_{ij}$ and a dressed network divisor number $\gdr_{ij}$, to each trivalent white vertex of ${\mathcal N}^{\prime}$ using the linear system considered in the previous section, after choosing a convenient initial time $\vec t_0$ with respect to which we normalize the wave functions. The fact that, for any soliton data $(\mathcal K, [A])$, there exists a proper choice of time $\vec t_0$ such that both the vacuum and the dressed wavefunctions are non zero at all double points using a is an essential condition to construct non--special (vacuum and dressed) divisors on $\Gamma$. In particular, we use the fact that the all horizontal edge vectors on the $r$-th row have the highest non-zero component in position $j_{N_r}$, and these components share the same sign.

\begin{lemma}
\label{rem:sign}\textbf{Choice of the initial time.}
Let the Le--network ${\mathcal N}^{\prime}$, the v.e.w. $\Phi^{(m)}_{i j} (\vec t)$ and the d.e.w. $\Psi^{(m)}_{i j} (\vec t)$ 
be as in the previous section for the soliton data $(\mathcal K,[A])$. Then there exists $x_0$ such that for all real $x>x_0$:
\begin{enumerate}
\item The signs of $\Phi^{(1)}_{i_r j_l} (x,0,0,\ldots)$ and $\Phi^{(3)}_{i_r j_l} (x,0,0,\ldots)$, for any given $l\in [0,N_r]$, $r\in[k]$,  are equal to the sign of their highest non zero coefficient;
\item The d.e.w. $\Psi^{(m)}_{i_r j_{l}} (x,0,0,\ldots)\not =0$ on any edge $e^{(m)}_{i_r j_l}$ distinct from the Darboux edges, $e^{(m)}_{i_r j_l} \not \in \{ e^{(3)}_{i_r}\, : \, r\in [k] \}$.  
\end{enumerate}
\end{lemma}
\begin{proof}
The first statement easily follows from the definition of the vacuum wave function, since no edge vector is null. For the second statement we recall that we use the reduced Darboux transformation if the soliton data belongs to a reducible cell. Therefore, we assume without loss of generality that $[A]$ belongs to an irreducible cell.  The Darboux operator $\mathfrak D$ is a regular $k$-th order ordinary linear differential operator, and the d.e.w. is identically zero by definition at the Darboux edges $e^{(3)}_{i_r}$, $r\in [k]$ and nowhere else. At all other edges the d.e.w. is a linear combination of real exponentials divided by the $\tau$-function. For this reason, the d.e.w. has only finite number of real zeroes, and, if $x_0$ is sufficiently big, $\Psi^{(m)}_{i_r j_{l}} (x,0,0,\ldots)\not =0$.
\end{proof}

\begin{definition}\label{def:vac_div_gen}\textbf{The vacuum network divisor $\DVN$ and the dressed network divisor $\DDN$.}
Let $\Phi^{(m)}_{i_rj_l} (\vec t)$, $\Psi^{(m)}_{i_rj_l} (\vec t)$ be the vacuum and the dressed edge wave functions on the edges $e^{(m)}_{i_r j_l}$, $m\in [3]$, at $V_{i_rj_l}$ of the network ${\mathcal N}^{\prime}$. Let $\vec t_0=(x_0,0,0,\ldots) $ be fixed as in Lemma~\ref{rem:sign}.

We assign a vacuum network divisor number $\gvac_{i_rj_l}$ to \textbf{each} trivalent white vertex $V_{i_rj_l}$ ($r\in [k]$, $l	\in [0,N_r-1]$) : 
\begin{equation}\label{eq:vac_pole_def2}
\gvac_{i_rj_l} = \frac{ \Phi^{(1)}_{i_rj_l} (\vec t_0)}{ \Phi^{(1)}_{i_rj_l} (\vec t_0)+\Phi^{(2)}_{i_rj_l} (\vec t_0)}=
\frac{ \Phi^{(1)}_{i_rj_l} (\vec t_0)}{ \Phi^{(3)}_{i_rj_l} (\vec t_0)}  .
\end{equation}
We call the collection of the $g$ pairs: $\DVN= \{( \gvac_{i_rj_l},V_{i_rj_l}),  \;  l	\in [0,N_r-1],  r\in [k] \}$ the vacuum network divisor on ${\mathcal N}^{\prime}$.

Analogously, we assign a dressed network divisor number $\gdr_{i_rj_l}$ to each trivalent white vertex $V_{i_rj_l}$ \textbf{not containing a Darboux edge} ($r\in [k]$, $l	\in [N_r-1]$):
\begin{equation}\label{eq:dress_pole_def}
\gdr_{i_rj_l} = \frac{ \Psi^{(1)}_{i_rj_l} (\vec t_0) }{\Psi^{(1)}_{i_rj_l} (\vec t_0)+\Psi^{(2)}_{i_rj_l} (\vec t_0)}=
 \frac{ \Psi^{(1)}_{i_rj_l} (\vec t_0) }{\Psi^{(3)}_{i_rj_l} (\vec t_0)}.
\end{equation}
We call the collection of the pairs $\DDN = \{ (\gdr_{i_rj_l}, V_{i_rj_l}), \; l	\in [N_r-1],  r\in [k]  \}$ the dressed network divisor on ${\mathcal N}^{\prime}$.
\end{definition}

\begin{remark}
We do not assign neither vacuum  nor dressed network divisor numbers to vertices connected with isolated boundary vertices, since these vertices are not trivalent in our construction.
\end{remark}

\begin{remark}
If the Le-network corresponds to a point in an irreducible positroid cell in $\Grkn$, then $\DDN$ is a collection of $g-k$ pairs.  If the Le-network corresponds to a point in a reducible positroid cell in $\Grkn$, and the reduced cell belongs to $Gr^{\mbox{\tiny TNN}}(k',n')$, then $\DDN$ is a collection of $g-k'$ pairs. 
\end{remark}
\begin{remark}
The vacuum and the dressed network divisors in Definition~\ref{def:vac_div_gen} are constructed using a certain orientation of the Le-network, but they behave differently with respect to changes of orientation and of the Darboux points positions \cite{AG2}. Indeed, any such change induces an untrivial transformation of the vacuum network divisor numbers. On the contrary, in case of changes of orientation, the dressed divisor numbers transform in agreement with the change of coordinates induced on the components of $\Gamma$. Moreover, the dressed divisor numbers are invariant under changes of Darboux points positions. Therefore the position of the KP divisor in the ovals is independent on both the orientation of the network and the position chosen for the Darboux source points. 
\end{remark}

In the next Lemma, we normalize the v.e.w. previously defined and characterize the vacuum network divisor.

\begin{definition}\label{def:norm_vac_gen}\textbf{The normalized vacuum edge wavefunction $\hat \Phi$, and the KP edge wave function $\hat \Psi$.}  Let $\Phi^{(m)}_{i_rj_l} (\vec t)$ and $\Psi^{(m)}_{i_rj_l} (\vec t)$ be the vacuum and the dressed edge wave functions on ${\mathcal N}^{\prime}$ with $\vec t_0$ as above.

Then the normalized vacuum edge wave function on ${\mathcal N}^{\prime}$ is
\begin{equation}\label{eq:vnvwf}
{\hat \Phi}^{(m)}_{i_r j_l} (\vec t) = \frac{\Phi^{(m)}_{i_r j_l} (\vec t)}{\Phi^{(m)}_{i_r j_l} (\vec t_0)},\quad \quad l\in [0,N_r], \,\, r\in[k],\,\, m\in [3],
\end{equation}
and the KP edge wave function on ${\mathcal N}^{\prime}$ is 
\begin{equation}\label{eq:KPvfnorN}
\hat \Psi^{(m)}_{i_r j_l} (\vec t) = \left\{ \begin{array}{ll} \displaystyle\frac{\Psi^{(m)}_{i_r j_l} (\vec t) }{\Psi^{(m)}_{i_r j_l} (\vec t_0) }, &\quad\quad \mbox{for all } m\in [3], \, l\in[N_r], \, r \in [k],\\
\displaystyle\frac{{\mathfrak D} e^{i\theta_{i_r} (\vec t)}}{{\mathfrak D} e^{i\theta_{i_r} (\vec t_0)}},&\quad\quad \mbox{if } m=3, l=0 \mbox{ i.e. on the Darboux edge } e^{(3)}_{i_r}, \;\; r\in [k].
\end{array}\right.
\end{equation}
\end{definition}

\begin{lemma}\label{lemma:vacvertexwf}
Let $({\mathcal K}, [A])$ and $\hat \Phi$, $\hat \Psi$ be as in the Definition \ref{def:norm_vac_gen} with  
$\vec t_0 =(x_0,0\dots)$ as in Lemma~\ref{rem:sign}. 
Then they have the following properties:
\begin{enumerate}
\item ${\hat \Phi}^{(m)}_{i_r j_l} (\vec t)$ and  ${\hat \Psi}^{(m)}_{i_r j_l} (\vec t)$ are regular functions of $\vec t$ for any $r\in[k]$, $l\in [0,N_r]$, $m\in [3]$;
\item  ${\hat \Phi}(\vec t)$  ( ${\hat \Psi}(\vec t)$ respectively) takes the same value on all edges at any given black vertex and any given bivalent white vertex;
\item For any fixed $r\in [k]$, $j\in [N_r-1]$, ${\hat \Phi}(\vec t)$  ( ${\hat \Psi}(\vec t)$ respectively) takes different values on the edges at the white trivalent vertex $V_{i_r j_l}$ for almost all $\vec t$, and the vacuum and dressed network divisor numbers $\gvac_{i_r j_l}$, $\gdr_{i_r j_l}$, defined in (\ref{eq:vac_pole_def2}) are the unique numbers such that
\[
{\hat \Phi}^{(3)}_{i_r j_l} (\vec t)\equiv\gvac_{i_r j_l} {\hat \Phi}^{(1)}_{i_r j_l} (\vec t) +(1-\gvac_{i_r j_l}) {\hat \Phi}^{(2)}_{i_r j_l} (\vec t), \quad\quad \forall \vec t;
\]
\[
{\hat \Psi}^{(3)}_{i_r j_l} (\vec t)\equiv\gdr_{i_r j_l} {\hat \Psi}^{(1)}_{i_r j_l} (\vec t) +(1-\gdr_{i_r j_l}) {\hat \Psi}^{(2)}_{i_r j_l} (\vec t), \quad\quad \forall \vec t.
\]
Moreover, all $\gvac_{i_r j_l}$ are \textbf{positive};
\item For any $r\in [k]$, such that $b_{i_r}$ is not an isolated boundary vertex, the normalized v.e.w. takes different values on the edges at the white trivalent vertex $V_{i_r }$ for almost all $\vec t$, and $\gvac_{i_r}$ defined in (\ref{eq:vac_pole_def2}) is the unique \textbf{positive} number such that
\[
{\hat \Phi}^{(3)}_{i_r} (\vec t)\equiv\gvac_{i_r} {\hat \Phi}^{(1)}_{i_r} (\vec t) +(1-\gvac_{i_r}) {\hat \Phi}^{(2)}_{i_r} (\vec t) , \quad\quad \forall \vec t.
\]
If, for some $r\in [k]$, $b_{i_r}$ is an isolated boundary vertex, the white vertex $V_{i_r }$ is bivalent and ${\hat \Phi}^{(3)}_{i_r} (\vec t)\equiv {\hat \Phi}^{(2)}_{i_r} (\vec t) = \exp(\theta_{i_r} (\vec t- \vec t_0))$, $\forall \vec t$;
\item For any $r\in [k]$, the normalized KP edge wave function takes the same value on the edges at the white trivalent vertex $V_{i_r }$ for all $\vec t$, 
\[
{\hat \Psi}^{(1)}_{i_r} (\vec t) = {\hat \Psi}^{(2)}_{i_r} (\vec t) = {\hat \Psi}^{(3)}_{i_r} (\vec t);
\]
\item The set of normalized v.e.w. ${\hat \Phi}^{(3)}_{i_r} (\vec t) $ defined at the vertices $V_{i_r}$, $r\in [k]$,
\begin{enumerate}
\item Form a basis of heat hierarchy solutions which satisfy ${\mathfrak D} \Phi^{(3)}_{i_r } (\vec t)\equiv 0 $, for all $r\in [k]$, $\vec t$, with 
${\mathfrak D}= W\partial_x^k \equiv \partial_x^k - {\mathfrak w}_1 (\vec t) \partial_x^{k-1} - {\mathfrak w}_2 (\vec t) \partial_x^{k-2} -\cdots -{\mathfrak w}_k (\vec t)$,
the Darboux transformation and $W$ the Sato dressing operator for $({\mathcal K}, [A])$;
\item Generate the KP soliton solution $u(\vec t)$ associated to $({\mathcal K}, [A])$ 
\[
u (\vec t) = 2\partial_{xx} \log (\tau (\vec t)), \quad\quad \tau(\vec t) = \mbox{ Wr}_x ({\hat \Phi}^{(3)}_{i_1 } (\vec t), \dots,{\hat \Phi}^{(3)}_{i_k } (\vec t))
\] 
where $\mbox{ Wr}_x$ denotes the Wronskian of ${\hat \Phi}^{(3)}_{i_r } (\vec t)$, $r\in [k]$, with respect to $x=t_1$. 
\end{enumerate}
\end{enumerate}
\end{lemma}

The proof is trivial and is omitted. We illustrate the above statement in Figure \ref{fig:vertexrule}.

\begin{example}
Let's compute the normalized vacuum edge wave function for Example \ref{ex:vertexvector}. At the vertex $V_{45}$ the normalized vacuum edge wave function is
\[
{\hat \Phi}^{(3)}_{45} (\vec t) = \frac{e^{\theta_5 (\vec t)} +w_{46}e^{\theta_6 (\vec t)}-w_{46}w_{48}e^{\theta_8 (\vec t)}-w_{46}w_{48}w_{49}e^{\theta_9 (\vec t)}}{e^{\theta_5 (\vec t_0)} +w_{46}e^{\theta_6 (\vec t_0)}-w_{46}w_{48}e^{\theta_8 (\vec t_0)}-w_{46}w_{48}w_{49}e^{\theta_9 (\vec t_0)}}.
\]
At the vertex $V_{25}$ we associate the {\bf opposite} edge wave function $\Phi^{(3)}_{25} (\vec t) =-\Phi^{(3)}_{45} (\vec t)$ and the {\bf same normalized} edge wave function ${\hat \Phi}^{(3)}_{25} (\vec t) ={\hat \Phi}^{(3)}_{45} (\vec t)$ for all times $\vec t$. We leave as an exercise to the reader the computation of the vacuum divisor numbers for this example (see also Figure \ref{fig:Mcurve_ex2}).
\end{example}

\begin{figure}
  \centering
	  {\includegraphics[width=0.45\textwidth]{Example2_NG_bis.pdf}}
\hfill
  {\includegraphics[width=0.45\textwidth]{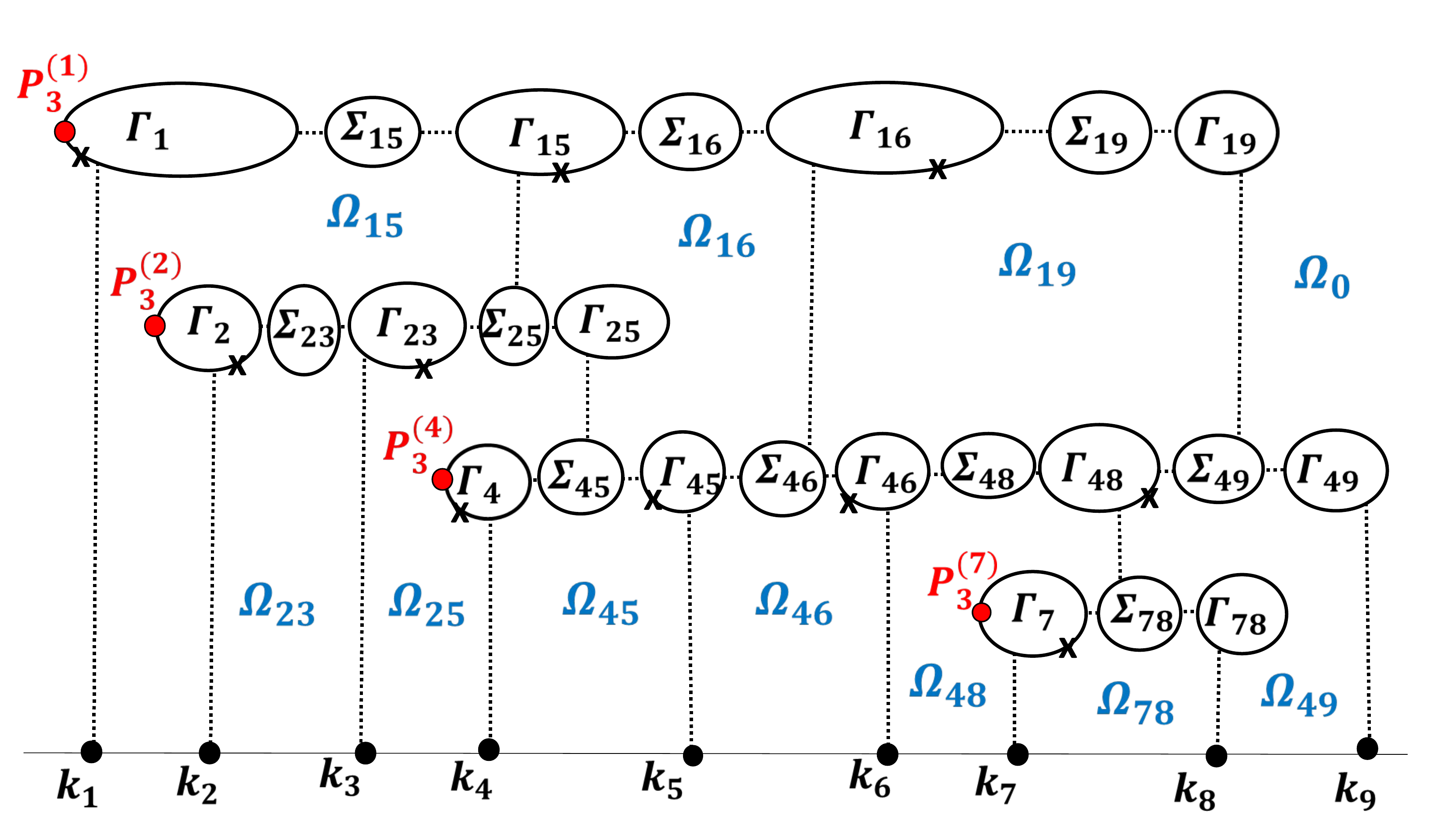}}
	\vspace{-.3 truecm}
  \caption{\footnotesize{\sl The vacuum divisor (crosses) on the $\mathtt M$--curve $\Gamma$ [right] for Example \ref{ex:gr492}: the local coordinate of the vacuum divisor point on the component $\Gamma_{ij}\subset \Gamma$ [right] is the divisor number of the vertex $V_{ij}$ on ${\mathcal N}^{\prime}$ [left].}}
	\label{fig:Mcurve_ex2}
\end{figure}

\subsection{From the edge wave functions on the network to the wave functions on the curve $\Gamma$}\label{sec:vacdiv}

In this section we define the normalized vacuum wave function ${\hat \phi} (P, \vec t)$ and the normalized KP wave function ${\hat \psi} (P, \vec t)$  on the curve $\Gamma$ using the following construction: 
\begin{construction}\textbf{The normalized vacuum and KP wave functions on $\Gamma$.}\label{def:vac0} 
\begin{enumerate}
\item On each component $\Gamma_{i_r j_s}$ of $\Gamma$ corresponding to a white vertex $V_{i_r j_s}$ carrying a vacuum (dressed) divisor number we assign a vacuum (dressed) divisor point $\Pvac_{i_r j_l}$ ($\Pdr_{i_r j_l}$)  with coordinate equal to this number;
\item On $\Gamma_0$ the normalized vacuum (dressed) wave functions coincide with the normalized vacuum (dressed Sato) wave functions respectively;
\item On each marked point of $\Gamma$ the value of the normalized vacuum (dressed) wave function is assigned equal to the value of the vacuum (dressed) edge wave function on the corresponding edge of $\mathcal N'$; 
\item On each component of $\Gamma\backslash\Gamma_0$ carrying no divisor points the corresponding function is extended as constant function in the spectral parameter. We use here that on each such component the values of the normalized (vacuum or KP) wave function are the same at all marked points;
\item On each component of $\Gamma\backslash\Gamma_0$ carrying a divisor point the corresponding function is extended to an order 1 meromorphic function in the spectral parameter with a simple pole at the divisor point. More precisely, for $r\in[k]$ such that $N_r>0$, $\gvac_{i_r j_s} = \frac{\Phi^{(1)}_{i_r j_s} (\vec t_0)}{\Phi^{(3)}_{i_r j_s} (\vec t_0)}$,  $s\in [0,N_r-1]$, and 
\[
 {\hat \phi}_{i_r j_s} (\zeta, \vec t) = \frac{\Phi^{(1)}_{i_r j_s} (\vec t) (\zeta-1) + \Phi^{(2)}_{i_r j_s} (\vec t) \zeta }{\Phi^{(3)}_{i_r j_s} (\vec t_0)  \left( \zeta- \gvac_{i_r j_s} \right)},
\]
where ${\hat \phi}_{i_r j_s} (\zeta, \vec t)$ denotes the restriction of the normalized vacuum wave function to the component $\Gamma_{i_r,j_s}$. Similarly, for $r\in[k]$ such that $N_r>0$, $\gdr_{i_r j_s} = \frac{\Psi^{(1)}_{i_r j_s} (\vec t_0)}{\Psi^{(3)}_{i_r j_s} (\vec t_0)}$,  $s\in [N_r-1]$, and 
\[
{\hat \psi}_{i_r j_s} (\zeta, \vec t) = \frac{\Psi^{(1)}_{i_r j_s} (\vec t) (\zeta-1) + \Psi^{(2)}_{i_r j_s} (\vec t) \zeta }{\Psi^{(3)}_{i_r j_s} (\vec t_0)  \left( \zeta- \gdr_{i_r j_s} \right)},
\]
where ${\hat \psi}_{i_r j_s} (\zeta, \vec t)$ denotes the restriction of the normalized KP wave function to the component $\Gamma_{i_r,j_s}$,  $\zeta$ is the local coordinate on this component.

\end{enumerate}
\end{construction}
The vacuum divisor $\DVG$ is the sum of all points $\Pvac_{i_r j_l}$, $r\in [k], \; j\in [0,N_r-1]$, assigned at $\Gamma\backslash\Gamma_0$, and it contains no points at $\Gamma_0$. The KP divisor $\DKP$ on $\Gamma$ is the sum of all points $\Pdr_{i_r j_l}$,  $r\in [k], \; j\in [N_r-1]$, assigned at $\Gamma\backslash\Gamma_0$ and of the Sato divisor  $\DS(\vec t_0)$ on $\Gamma_0$.

Finally, we check that we have the correct number of divisor points at each oval.

\begin{theorem}\label{theo:existvac}{\bf The vacuum divisor (Theorem~\ref{theo:exist}).} ${\hat \phi}$ as in Construction \ref{def:vac0}  is the real and regular vacuum wave function on $\Gamma$ for the soliton data $(\mathcal K, [A])$, that is it satisfies all the properties of Definitions \ref{def:vacuumKP} and \ref{def:vacuumKP2} on $\Gamma\backslash\{P_0\}$. In particular, it satisfies (\ref{eq:darboux_cond}), that is
\[
{\hat \phi} (P^{(3)}_{i_r}, \vec t) = \frac{\sum_{l=1}^n A^{r}_l \exp (\theta_l (\vec t))}{\sum_{l=1}^n A^{r}_l \exp (\theta_l (\vec t_0))}, \quad\quad \mbox{ for } r\in [k], \;\; \forall \vec t,
\]
where $A$ is the RREF representative matrix in $[A]$. Finally the vacuum divisor $\DVG = \{ \Pvac_{i_r j_l} \}$  of $\hat \phi$ satisfies:
\begin{enumerate} 
\item It consists of $g$ simple poles and no pole belongs to $\Gamma_0$;
\item There is exactly one pole on each component $\Gamma_{i_rj_l}$ corresponding to a trivalent white vertex in ${\mathcal N}^{\prime}$; 
\item\label{it:4} For any finite oval  $\Omega_{s}$, $s\in [g]$, let $\nu_{s} = \# \{ \DVG \cap \Omega_{s} \}$ and $\mu_{s} =  \# \{ P^{(3)}_{i_r} \in \Omega_{s}, \; r\in [k] \}$ respectively be the number of poles and the number of Darboux points in $\Omega_{s}$. Then 
\begin{equation}\label{eq:odd}
\nu_{s}+ \mu_{s} =\mbox{ odd number}, \quad\quad \mbox{ for any } s\in [g];
\end{equation}
\item\label{it:5} Let $\nu_{0} = \# \{ \DVG \cap \Omega_{0} \}$ and $\mu_{0} =  \# \{ P^{(3)}_{i_r} \in \Omega_{0}, \; r\in [k] \}$, respectively be the number of poles and the number of Darboux points in the infinite oval $\Omega_{0}$. Then
\begin{equation}\label{eq:even}
\nu_{0}+ \mu_{0} - k =\mbox{ even number}.
\end{equation}
Let us remark that $k-\mu_0$ is the number of Darboux source points not belonging to the infinite oval $\Omega_0$, therefore, equivalently, 
\begin{equation}\label{eq:even2}
\nu_{0} =  \#\mbox{Darboux points not belongig to $\Omega_0$} \mod 2.
\end{equation}

\end{enumerate}
\end{theorem}

\begin{figure}
  \centering
  {\includegraphics[width=0.45\textwidth]{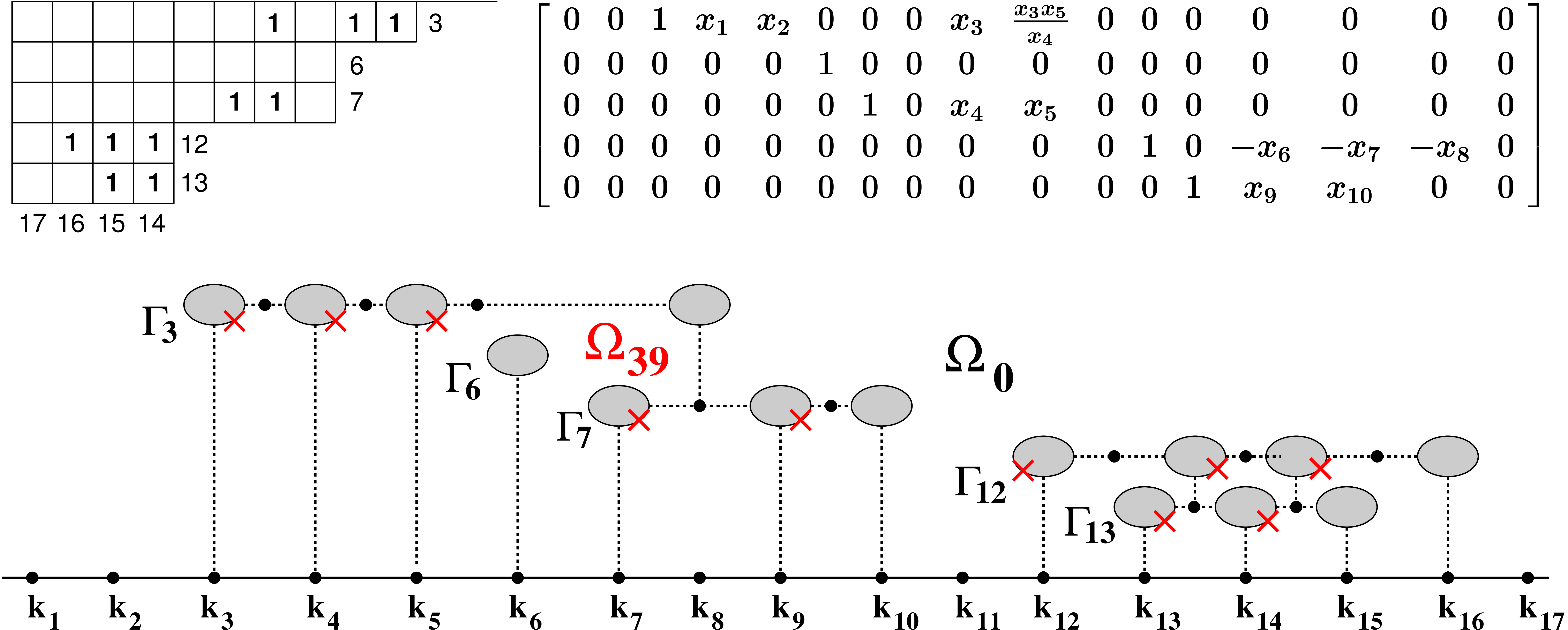}}
	\hspace{.5 truecm}
	{\includegraphics[width=0.45\textwidth]{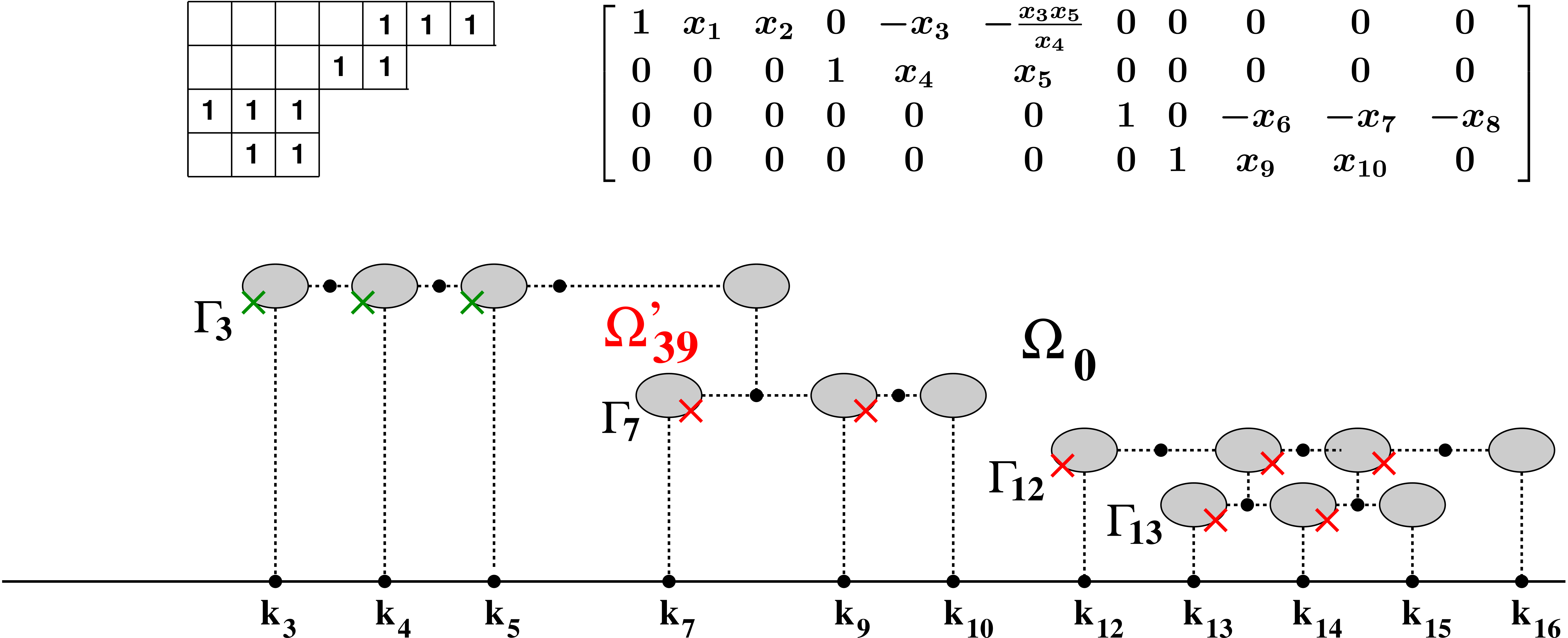}}
  \caption{\footnotesize{\sl  We illustrate Case b) in the proof of Theorem \ref{theo:existvac} for the example discussed in Remark \ref{rem:caseb}.}}
\label{fig:CaseB}
\end{figure}

\begin{proof}

The only untrivial statements are properties (\ref{it:4}) and (\ref{it:5}). Copies of $\mathbb{CP}^1$ corresponding to isolated boundary sources contribute to the counting only with a Darboux point, while those corresponding to not isolated boundary sources contribute with a pole and a Darboux point. Moreover, by definition, all poles have positive $\zeta$--coordinate so that:
\begin{enumerate}
\item for any fixed $r\in [k]$ and $l\in [N_r]$, the finite oval $\Omega_{i_rj_l}$ may contain only poles belonging either to $\Gamma_{i_rj_l}$ or to $\Gamma_{i_rj_{l-1}}$ or to $\Gamma_{i_{r+d}}$, with $d\in [k-r]$;
\item $\Omega_0$  may contain only poles belonging to $\Gamma_{i_{r}}$, with $r\in [k]$.
\end{enumerate}

Let's prove (\ref{eq:even}) first.
 
{\bf Case a)} First of all, consider the simplest situation in which the only Darboux point belonging to $\Omega_0$ is 
$P^{(D)}_{i_1}$. In this case, $\mu_0=1$ in ${\mathcal N}'$, since $\Gamma_{i_s} \cap\Omega_0\not =\emptyset$ if and only if $s=1$. The pole on $\Gamma_{i_1}$ has $\zeta$--coordinate
\[
\gvac_{i_1} = \frac{\Phi^{(1)}_{i_1}(\vec t_0)}{\Phi^{(1)}_{i_1}(\vec t_0)+ e^{\theta_{i_1} (t_0)}}>0
\]
and ${\rm sign} (\Phi^{(1)}_{i_1}(\vec t_0))= (-1)^{k-1}$ by construction (compare Lemma \ref{lemma:vacvertexwf} with Theorem \ref{theo:mainalg}).
Then, if $k$ is even, $\gvac_{i_1} >1$, the vacuum divisor point $\Pvac_{i_1} \in \Omega_0$, $\nu_{0}=1$ and (\ref{eq:even}) holds. If $k$ is odd, $0<\gvac_{i_1} <1$,  $\Pvac_{i_1} \not \in \Omega_0$,  $\nu_0=0$ and (\ref{eq:even}) again holds true.

{\bf Case b)} Otherwise, the oval  $\Omega_0$ contains $d>1$ Darboux points $P^{(D)}_{i_{s_1}}$,\ldots,  $P^{(D)}_{i_{s_d}}$, $1\le s_1< \cdots <s_d\le k$. Then, for any $l\in [d]$ , there exists an index ${\bar \jmath_l }$ satisfying $1\le i_{s_1} \le {\bar \jmath_1} <i_{s_2}\le {\bar \jmath_2 } <\cdots < i_{s_d}\le {\bar \jmath_d }\le n$ such that, for any fixed $l\in [d]$, if $i_{s_l} = {\bar \jmath_l }$ then $b_{i_{s_l}}$ is an isolated boundary source in ${\mathcal N}$(an isolated boundary sink in ${\mathcal N}^{\prime}$); otherwise ${\bar \jmath_l }$ is the maximum non pivot index such that there exists a path from the boundary source $b_{i_{s_l}}$ to the boundary sink $b_{{\bar \jmath_l}}$.

In this case, in ${\mathcal N}'$, $\Omega_0 \cap \Gamma_{i_r}\not =\emptyset$ if and only if $r\in \{ s_1,\dots, s_d\}$. For any fixed $l\in [d]$, we apply the argument of Case a) to the subgraph of ${\mathcal N}'$ containing only boundary vertices between $b_{i_{s_l}}$ and $b_{\bar \jmath_l}$, and again (\ref{eq:even}) holds true.  We refer to Figure \ref{fig:CaseB} for an illustrating example.

\smallskip

To prove (\ref{eq:odd}), we treat separately the cases:
\begin{enumerate}
\item The finite oval $\Omega_{i_r,j_l}$ does not intersect $\Gamma_0$;
\item The finite oval $\Omega_{i_r,j_l}$ intersects $\Gamma_0$;
\end{enumerate}
In the first case there are no Darboux points belonging to $\Omega_{i_r,j_l}$, whereas the only poles which may belong to $\Omega_{i_r,j_l}$ are those belonging to the upper components $\Gamma_{i_r,j_l}$ and $\Gamma_{i_r,j_{l-1}}$. Let $i_s$ be the biggest pivot index such that there exists a component with the first index $i_s$ intersecting  $\Omega_{i_r,j_l}$. By construction, the vacuum wave function has the same sign at all horizontal edges with the same first index at time $\vec t_0$, and $\sign \Phi^{(2)}_{{i_r}j_{l}} (\vec t_0) = (-1)^{s-r} \sign \Phi^{(3)}_{i_{s}j_{N_s}} (\vec t_0)$, 
$\sign \Phi^{(2)}_{{i_r}j_{l-1}} (\vec t_0) = (-1)^{s-r} \sign \Phi^{(3)}_{i_{s}j_{N_s}} (\vec t_0)$. Therefore, only one pole belongs to  $\Omega_{i_r,j_l}$ since $\Phi^{(2)}_{{i_r}j_{l}} (\vec t_0)\, \Phi^{(2)}_{i_{r}j_{l-1}} (\vec t_0)>0$. 

In the second case let us denote $s_L$ and $s_R$ respectively:
$$
s_L=\min\{s\in[n],\kappa_s\in\Omega_{i_r,j_l} \}, \ \ s_R=\max\{s\in[n],\kappa_s\in\Omega_{i_r,j_l} \}.
$$
Let us remark that $s_L$ and $s_R$ may be both pivot or non-pivot indexes. Let us also introduce an index $\nu_{s_R}$ which tells whether the divisor point $\gvac_{s_R}$ belong to $\Omega_{i_r,j_l}$ if $s_R\in I$.  
$$
\nu_{s_R}=\left\{\begin{array}{ll} 0 & \mbox{if} \ s_R \not\in I\\
0 & \mbox{if} \ s_R \in I \ \mbox{and} \ \gvac_{s_R}\not\in\Omega_{i_r,j_l}   \\
1 & \mbox{if} \ s_R \in I \ \mbox{and} \ \gvac_{s_R} \in\Omega_{i_r,j_l}.
  \end{array} \right.
$$
Then
$$
\sign  \Phi^{(2)}_{{i_r}j_{l-1}} (\vec t_0) = (-1)^{\#(]i_r,s_L]\cap I)}, \ \ 
\sign  \Phi^{(2)}_{{i_r}j_{l}} (\vec t_0) = (-1)^{\#(]i_r,s_R]\cap I) + \nu_{s_R}}.
$$
Therefore
\begin{equation}
\label{eq:even3}
\sign \left( \Phi^{(2)}_{{i_r}j_{l-1}} (\vec t_0)\,  \Phi^{(2)}_{{i_r}j_{l}} (\vec t_0) \right)=  (-1)^{\#(]s_L,s_R]\cap I) + \nu_{s_R}}. 
\end{equation}
If $s_R-s_L>1$, then the boundary of $\Omega_{i_r,j_l}$ includes components associated to the boundary sources in $]s_L,s_R[$, behaving like components of the infinite oval in the intersection with  $\Omega_{i_r,j_l}$. Using (\ref{eq:even2}) we immediately conclude that $\tilde\nu_0(i_r,j_l)$, the number of divisor points on the intersection of $\Omega_{i_r,j_l}$ with such components, is equal to the number of Darboux points in  $]s_L,s_R[$ not belonging to  $\Omega_{i_r,j_l}$ $\mod 2$. 

Therefore, substituting $\tilde\nu_0(i_r,j_l)$ into (\ref{eq:even3}) we get
\begin{equation}
\label{eq:even4}
\sign \left( \Phi^{(2)}_{{i_r}j_{l-1}} (\vec t_0)\,  \Phi^{(2)}_{{i_r}j_{l}} (\vec t_0) \right) =  (-1)^{\tilde\nu_0(i_r,j_l) + \mu_{i_r,j_l} + \nu_{s_R}},
\end{equation}
where $ \mu_{i_r,j_l}$ is the number of Darboux points in  $\Omega_{i_r,j_l}$.

If $\tilde\nu_0(i_r,j_l) + \mu_{i_r,j_l} + \nu_{s_R}=0\mod 2$, then $\# \big(\{\gvac_{i_r,j_{l-1}},\gvac_{i_r,j_{l}}\}\cap\Omega_{i_r,j_l}\big)=1$, otherwise  $\#\big(\{\gvac_{i_r,j_{l-1}},\gvac_{i_r,j_{l}}\} \cap\Omega_{i_r,j_l}\big)=0\mod 2$, and the proof of  (\ref{eq:odd}) is complete.  
\end{proof}

\begin{remark}
One may adapt this proof to check directly that the KP divisor has exactly one divisor point at each finite oval and no divisor points at the infinite oval, see \cite{AG3}.
\end{remark}

\begin{remark}\label{rem:caseb}{\bf The effect of the isolated boundary vertices on the effective vacuum divisor.}
Any KP soliton solution is associated to an unique irreducible cell in $Gr^{\mbox{\tiny TNN}}(k^{\prime},n^{\prime})$ which is obtained eliminating all isolated boundary vertices from the Le--diagram. Let $I$ be the pivot set. The elimination of an isolated boundary vertex $b_j$, $j\in \bar I$, does not affect the vacuum divisor. The elimination of an isolated boundary vertex $b_{i_s}$, $i_s\in I$, produces a change of sign in the matrix elements in all the rows $r\in [1,s-1]$ which lie to the right of the $i_s$--th column. This change of sign is automatically encoded in the basis of vectors by the recursion associated to the reduced matrix, and effects both the sign of the vacuum edge wave function and the position of the vacuum divisor not just in the oval where we remove $\Gamma_{i_s}$, but also in all the ovals to the left and above it. We show an example in Figure \ref{fig:CaseB} of both the reducible and reduced vacuum divisors. The original Le--diagram in Fig.\ref{fig:CaseB}[above left] has RREF matrix $A$ depending on 10 positive parameters $x_l$, $l\in [10]$ and the corresponding reducible rational $\mathtt M$--curve is shown in Fig.\ref{fig:CaseB}[bottom left]. The crosses are the divisor points of the KP vacuum wave function: here $k=5$, $d=2$, $s_1=1$, $s_2=4$, $i_{1}=3$, $i_4=12$, $j_1=10$ and $j_2=16$. The infinite oval $\Omega_0$ intersects $\Gamma_3$ and $\Gamma_{12}$ and $\mu_0 =2$. $\Phi^{(1)}_{12} (\vec t_0)$ has the sign of the entry $A^4_{12}=-x_8<0$, $\Phi^{(2)} (\vec t_0) ) =\exp(\theta_{12} (\vec t_0))>0$ so that $\gvac_{12}>1$ and $\Pvac_{12}\in \Omega_0$. $\Phi^{(1)}_{3} (\vec t_0)$ has the sign of the entry $A^1_{10}=(x_3x_5)/x_4>0$, $\Phi^{(2)} (\vec t_0) ) =\exp(\theta_{3} (\vec t_0))>0$ so that $\gvac_{3}<1$ and $\Pvac_{3}\not\in \Omega_0$. In conclusion $\nu_0=1$ and $k+\mu_0+\nu_0=8$ is even.

The reduced Le--diagram and the reduced matrix (see Fig.\ref{fig:CaseB}[top, right]) are obtained eliminating, respectively, all isolated boundary vertices. This transformation corresponds to the elimination of the component $\Gamma_6$ in the reducible rational curve (Fig.\ref{fig:CaseB}[bottom right]) and to the change of the vacuum pole divisor points (crosses) in the oval $\Omega_{39}^{\prime}$ which corresponds to $\Omega_{39}$ and in all the ovals to the left and above such oval. All other vacuum divisor points (crosses) in the ovals to the right and below respectively of $\Omega_{39}$ and of $\Omega_{39}^{\prime}$ are the same in Figure \ref{fig:CaseB} [bottom, left] and \ref{fig:CaseB}[bottom, right].
\end{remark}

\section{Construction of the plane curve and the divisor to soliton data in $Gr^{\mbox{\tiny TP}}(2,4)$ and comparison with the construction in \cite{AG1}}\label{sec:example}

\begin{figure}
\includegraphics[width=0.49\textwidth]{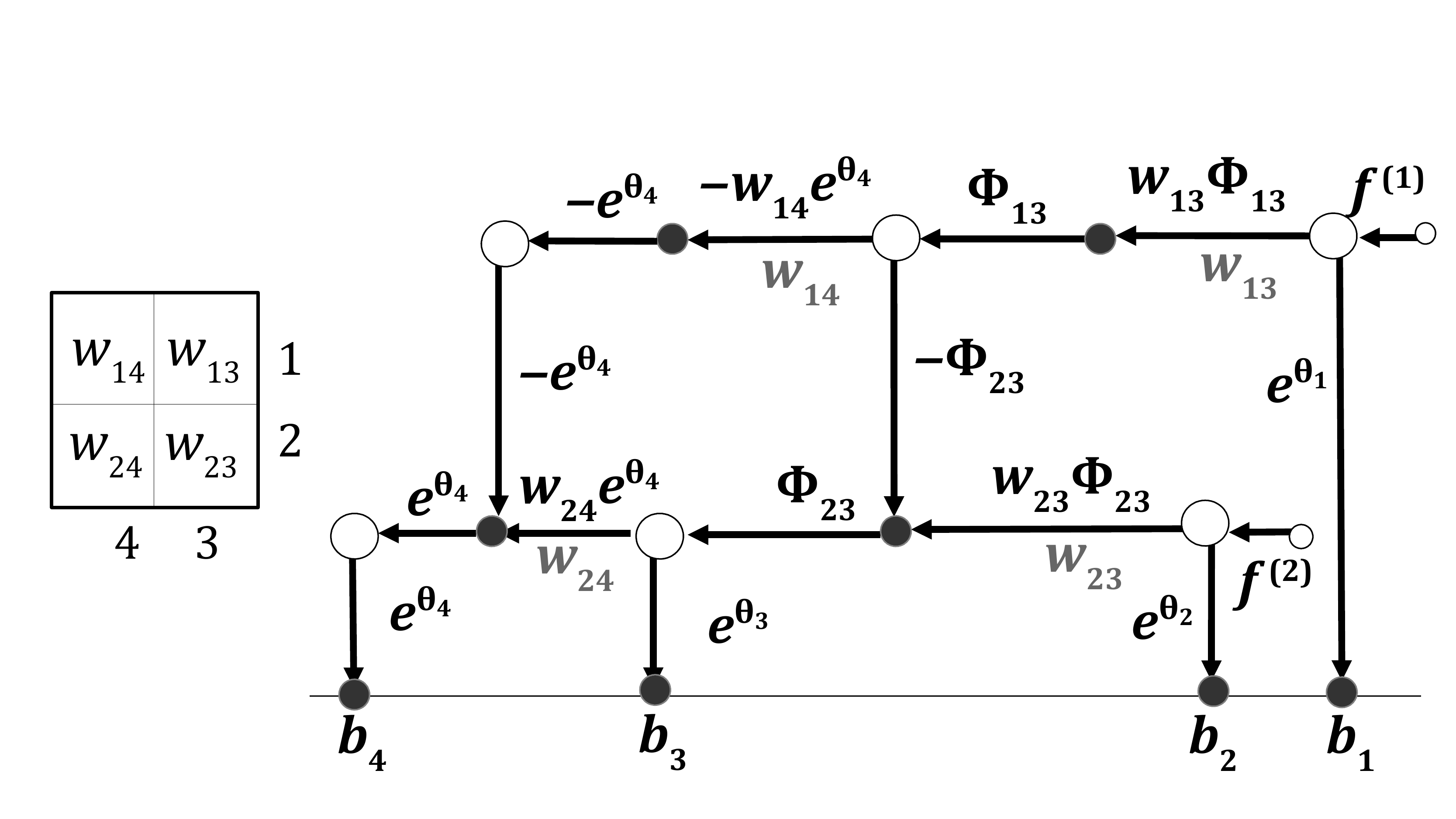}
\vspace{-.3 truecm}
\caption{\footnotesize{\sl The Le-tableau and vacuum edge wave function on the modified Le--network for soliton data in $Gr^{\mbox{\tiny TP}}(2,4)$. $\Phi_{13} (\vec t)=-e^{\theta_3(\vec t)}-(w_{14}+w_{24}) e^{\theta_4(\vec t)}$, while $\Phi_{23} (\vec t)=e^{\theta_3(\vec t)}+w_{24} e^{\theta_4(\vec t)}$.}}
\label{fig:Gr24_cur_1}
\end{figure}

$Gr^{\mbox{\tiny TP}}(2,4)$ is the main cell in $Gr^{\mbox{\tiny TNN}} (2,4)$ and its elements $[A]$ are equivalence classes of real $2\times 4$ matrices $A$ with all maximal minors positive. Such matrices are parametrized by the four weights of the Le-tableau (see Figure \ref{fig:Gr24_cur_1}), $w_{ij}$, $i=1,2$, $j=3,4$, and may be represented in the reduced row echelon form (RREF),
\begin{equation}
\label{eq:RREF}
A = \left( \begin{array}{cccc}
1 & 0 & - w_{13} & -w_{13} (w_{14}+w_{24})\\
0 & 1 & w_{23}   & w_{23}w_{24}
\end{array}
\right).
\end{equation}
The generators of the Darboux transformation ${\mathfrak D}=\partial_x^2 -{\mathfrak w}_1 (\vec t)\partial_x -{\mathfrak w}_2(\vec t)$ are
\begin{equation}\label{eq:gen_f2}
f^{(1)} (\vec t) = e^{\theta_1(\vec t)} -w_{13}e^{\theta_3(\vec t)}-w_{13} (w_{14}+w_{24})e^{\theta_4(\vec t)}, \quad\quad
f^{(2)} (\vec t) = e^{\theta_2(\vec t)} +w_{23}e^{\theta_3(\vec t)}+w_{23}w_{24}e^{\theta_4(\vec t)}.
\end{equation}
In \cite{AG1} we have proposed a plane curve representation of $\Gamma(\xi)$ for soliton data in $Gr^{\mbox{\tiny TP}} (2,4)$ as the product of a line, a quadric and a cubic and we have desingularized $\Gamma(\xi)$ to a genus 4 $\mathtt M$--curve. In \cite{A2} and \cite{AG3}, we use the reduced Le--network ${\mathcal G}_{{\mbox{\scriptsize red}}}$ to implement the construction of the present paper, and represent $\Gamma({\mathcal G}_{{\mbox{\scriptsize red}}})$ by five lines. We recall that the elimination of $\mathbb{CP}^1$ copies associated to bivalent vertices is fully justified since it does not effect neither the properties of the desigularized curve (see Section \ref{sec:gamma}) nor the divisor (see Remark \ref{rem:red_div}).

In \cite{AG3} we also desingularize $\Gamma({\mathcal G}_{{\mbox{\scriptsize red}}})$ to a smooth genus 4 $\mathtt M$--curve on which we numerically construct real--regular KP--II finite gap solutions, while in \cite{A2} evidence is provided that the asymptotic behavior of the KP soliton zero divisor in $\| (x,y) \| \gg 1$ for fixed time $t$ is consistent with the characterization in \cite{CK2}.

In this section, we discuss the same example with a different spirit from our previous publications: we represent 
\[
\Gamma({\mathcal G}_{{\mbox{\scriptsize red}}}) =\Gamma_0 \sqcup\Gamma_{13}^{\prime}\sqcup\Gamma_{23}\sqcup\Sigma_{23}^{\prime}\sqcup\Sigma_{24}^{\prime},
\]
as five lines which are the limit of two lines and a cubic representing
\[
\Gamma(\xi) =\Gamma_0\sqcup\Gamma_1(\xi)\sqcup\Gamma_2(\xi), \ \ \xi\gg1,
\]
so that $\Gamma({\mathcal G}_{{\mbox{\scriptsize red}}})=\Gamma(\infty)$. We present the 
topological model of both $\Gamma({\mathcal G}_{{\mbox{\scriptsize red}}})$ and $\Gamma(\xi)$ when $\xi\gg 1$ in Figure \ref{fig:gr24_top_3}[top]. 

$\Gamma(\xi)$ in Section \ref{sec:gamma_24_xi} is the 
reducible curve obtained by the intersection of two lines representing $\Gamma_0$ and $\Gamma_2(\xi)$ and a cubic representing
$\Gamma_1(\xi)$, such that
\[
\Gamma_1(\infty) = \Gamma_{23}\sqcup\Sigma^{\prime}_{23}\sqcup\Sigma^{\prime}_{24},\quad\quad \Gamma_2(\infty) = \Gamma_{13}^{\prime}.
\]
We remark that the above representation of $\Gamma(\xi)$ is different from that proposed in \cite{AG1}.
We also compute the KP divisors both on $\Gamma({\mathcal G}_{{\mbox{\scriptsize red}}})$ and $\Gamma(\xi)$. For $\xi\gg 1$, the
KP divisor on $\Gamma(\xi)$ coincides at leading order with that of $\Gamma({\mathcal G}_{{\mbox{\scriptsize red}}})$ in appropriate coordinates. We remark that the KP divisor is independent on the plane curve representation.
Finally, we desingularize both $\Gamma({\mathcal G}_{{\mbox{\scriptsize red}}})$ and  $\Gamma(\xi)$ to genus 4 $\mathtt M$--curves.

\subsection{A spectral curve for the reduced Le--network for soliton data in $Gr^{\mbox{\tiny TP}}(2,4)$ and its desingularization}\label{sec:gamma_24}
 
We briefly illustrate the construction of the rational spectral curve $\Gamma({\mathcal G}_{{\mbox{\scriptsize red}}})$ for soliton data in $Gr^{\mbox{\tiny TP}} (2,4)$. We reduce the degree of the curve from 11 to 5 by eliminating the $\mathbb{CP}^1$ components corresponding to bivalent vertices in ${\mathcal G}$. 
We represent the topological model of $\Gamma({\mathcal G})$ and of $\Gamma({\mathcal G}_{{\mbox{\scriptsize red}}})$ in Figure \ref{fig:gr24_top_1}: at all double points marked with the same symbol the value of the normalized dressed wave function is the same for all times. 

\begin{figure}
  \centering
  {\includegraphics[width=0.42\textwidth]{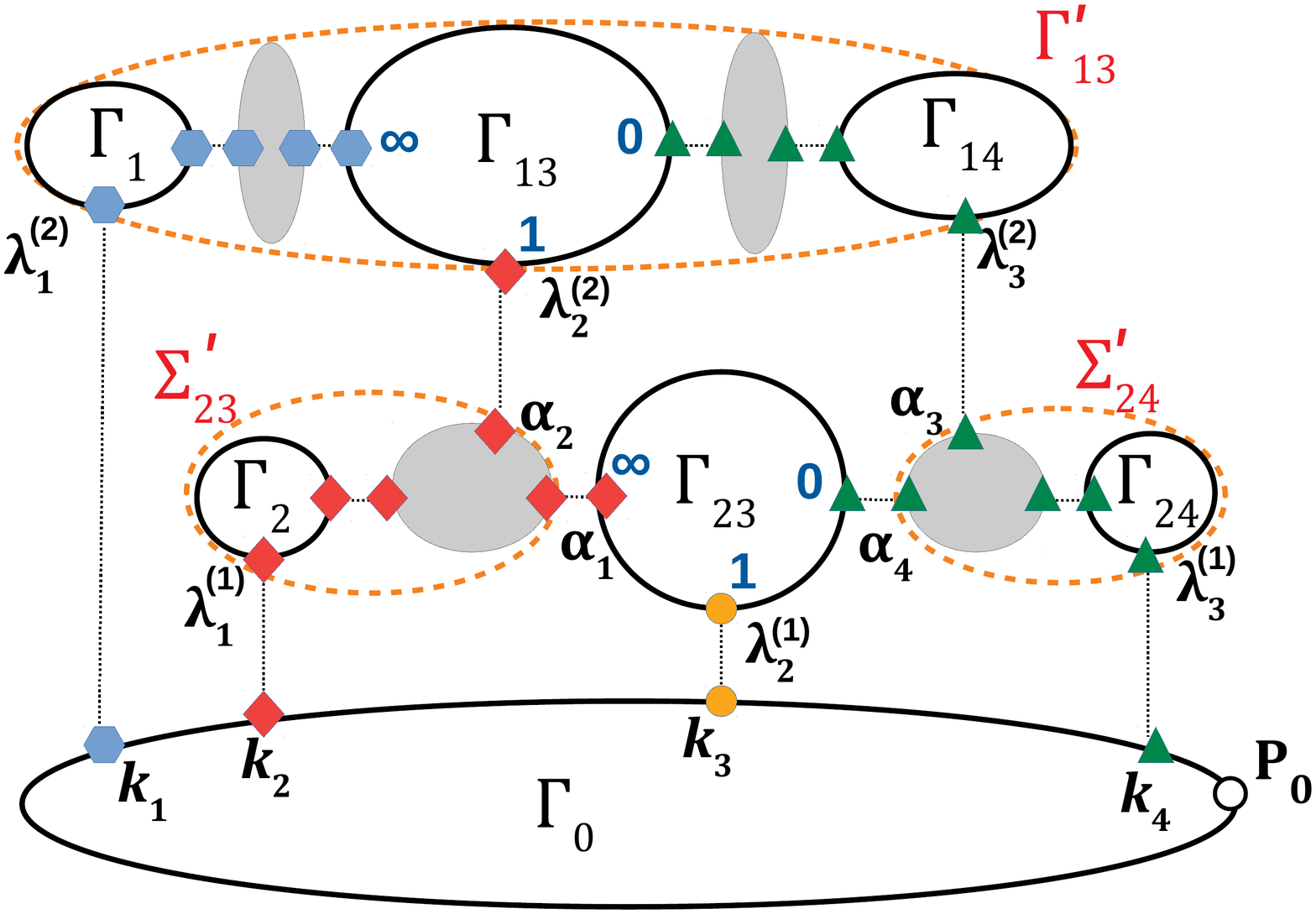}}
  \hspace{.6 truecm}
  {\includegraphics[,width=0.42\textwidth]{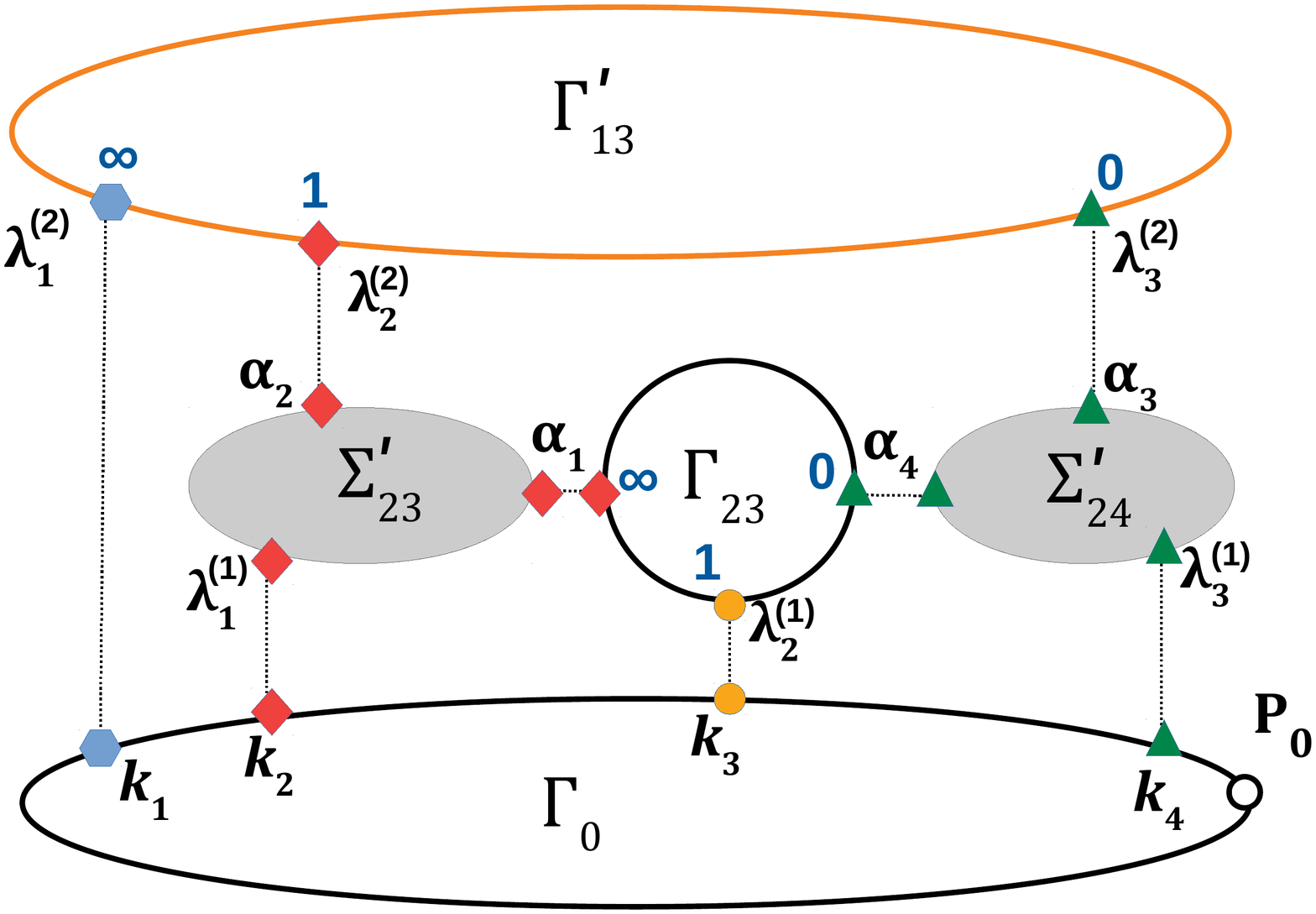}}
\caption{\footnotesize{\sl The topological model of the spectral curves $\Gamma({\mathcal G})$ [left] and $\Gamma({\mathcal G}_{{\mbox{\scriptsize red}}})$ [right] for soliton data $Gr^{\mbox{\tiny TP}}(2,4)$. In both figures, the value of the KP edge wave function is the same at all points marked with the same symbol.}}\label{fig:gr24_top_1}        
\end{figure}

The curve $\Gamma({\mathcal G}_{{\mbox{\scriptsize red}}})$ is the partial normalization of the nodal plane curve $\Pi_0(	\lambda,\mu)=0$, with $\Pi_0$ as in (\ref{eq:curveGr24}), and it is the rational degeneration of the genus 4 $\mathtt M$--curve $\Gamma_{\varepsilon}$ in (\ref{eq:curveGr24_pert}) for $\varepsilon\to 0$. 
Here we modify the plane curve representation in \cite{AG3} so that the plane curve 
representing $\Gamma(\xi)$ is a rational deformation of $\Pi_0(	\lambda,\mu)=0$ for $\xi\gg 1$.
We plot both the topological model and the plane curve for this example in Figure \ref{fig:gr24_top_2}. 

The reducible rational curve $\Gamma({\mathcal G}_{{\mbox{\scriptsize red}}})$ is obtained gluing five copies of $\mathbb{CP}^1$: $\Gamma_0$, $\Gamma_{13}^{\prime}$, $\Gamma_{23}$, $\Sigma_{23}^{\prime}$ and $\Sigma_{24}^{\prime}$ and it may be represented as a plane curve given by the intersection of five lines. To simplify its representation, we impose that $\Gamma_0$ is one of the coordinate axis in the $(\lambda,\mu)$--plane, say $\mu=0$, that $P_0\in \Gamma_0$ is the infinite point, that the lines $\Sigma_{23}^{\prime}$, $\Sigma_{24}^{\prime}$ are orthogonal to $\Gamma_0$, that $\Gamma_{23}$ is parallel to the first bisector and that $\Gamma_{13}^{\prime}$ and $\Gamma_{23}$ intersect at a finite point $\alpha_5$:
\begin{equation}\label{eq:lines}
\Gamma_0:  \mu=0, \ \ \Gamma_{13}^{\prime}: \mu-c_{13} (\lambda - \kappa_1) =0, \ \  \Gamma_{23}: \mu-\lambda +\kappa_3 =0, \ \ 
\Sigma_{23}^{\prime}: \lambda -\kappa_2=0, \ \ \Sigma_{24}^{\prime}: \lambda -\kappa_4=0.
\end{equation}
Notice that here we do not follow the generic construction of Section~\ref{sec:planar} and use parallel lines.

We choose
\begin{equation}\label{eq:c_13}
c_{13} =\frac{(\kappa_2+\kappa_4-2\kappa_3)(\kappa_4-\kappa_3)(\kappa_2-\kappa_1)+(\kappa_3-\kappa_1)(\kappa_3-\kappa_2)^2}{(\kappa_4-\kappa_1)(\kappa_2-\kappa_1)(\kappa_2+
\kappa_4 -2\kappa_3)},
\end{equation}
so that it coincides at leading order in $\xi$ with the coefficient of the line $\Gamma_2 (\xi)$ in (\ref{eq:gamma_2}).
Our representation fits generic values of $\kappa_j$. The figures in this and the following sections all refer to the case $\kappa_2+\kappa_4-2\kappa_3<0$.

As usual we denote $\Omega_0$ the infinite oval, that is $P_0\in \Omega_0$, and $\Omega_j$, $j\in[4]$, the finite ovals (see Figure \ref{fig:gr24_top_2}). 
Since the singularity at infinity is completely resolved, the lines $\Sigma_{23}^{\prime}$ and $\Sigma_{24}^{\prime}$ do not intersect at infinity. Finally the ovals $\Omega_1$ and $\Omega_4$ are both finite since neither of them passes through $P_0$. 

\begin{figure}
  \centering
  {\includegraphics[width=0.44\textwidth]{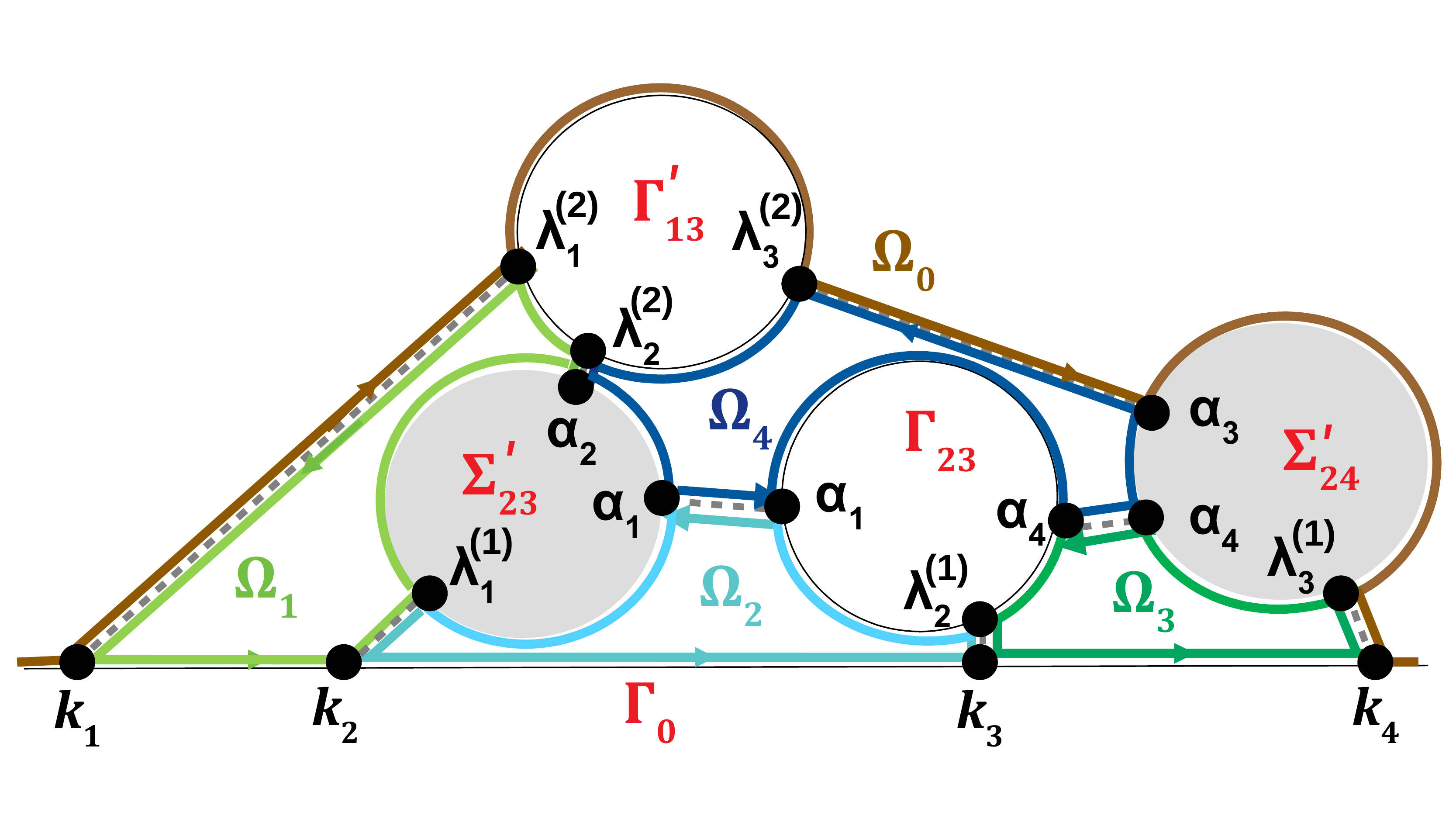}}
  \hspace{.6 truecm}
  {\includegraphics[,width=0.44\textwidth]{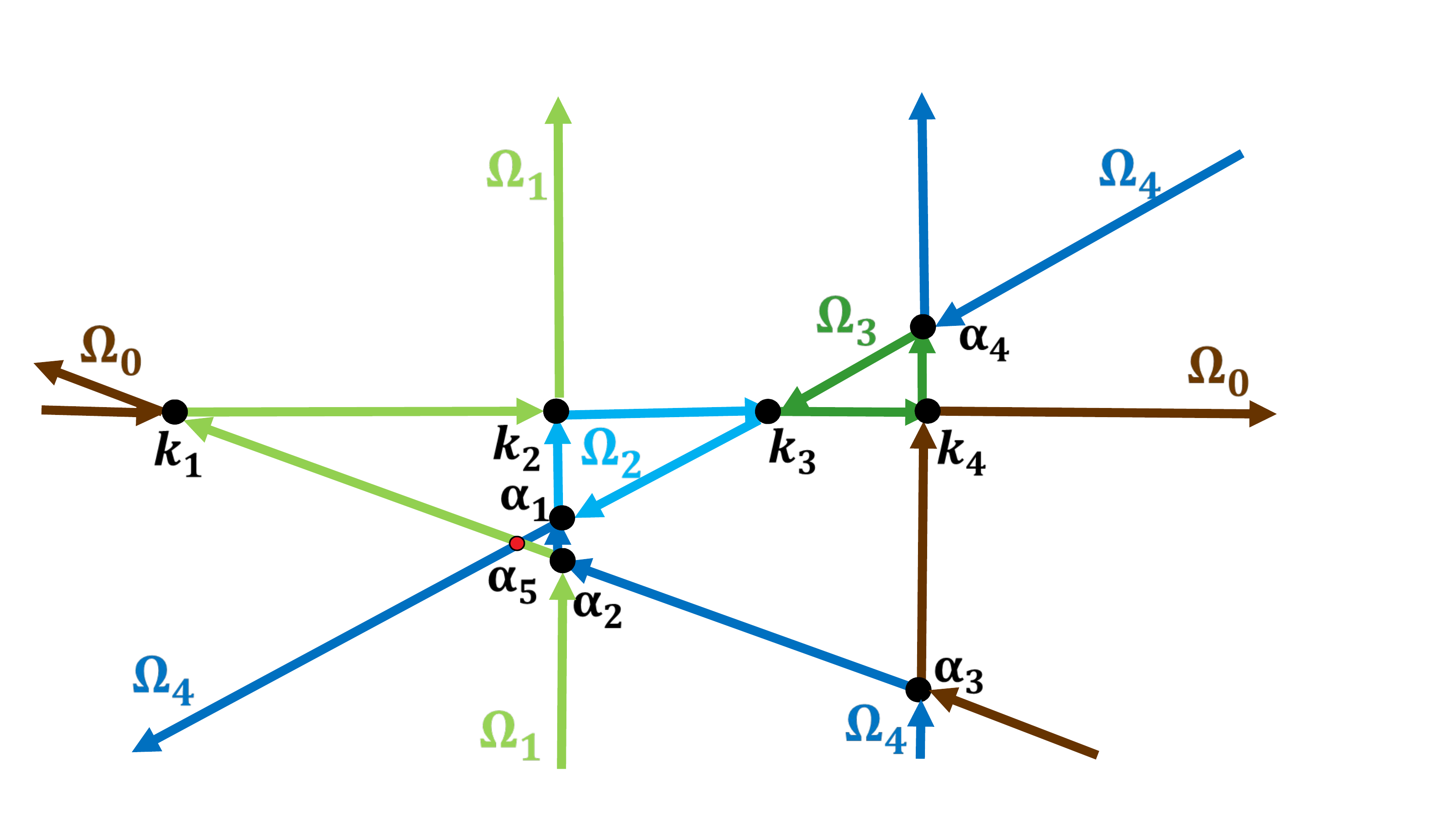}}
\caption{\footnotesize{\sl The topological model of spectral curve for soliton data $Gr^{\mbox{\tiny TP}}(2,4)$, $\Gamma({\mathcal G}_{{\mbox{\scriptsize red}}})$ [left] is the partial normalization of the plane algebraic curve [right]. The ovals in the nodal plane curve are labeled as in the real part of its partial normalization. }}\label{fig:gr24_top_2}        
\end{figure}

The relation between the coordinate $\lambda$ in the plane curve representation and the coordinate $\zeta$ introduced in Definition \ref{def:loccoor} may be easily worked out at each component of
$\Gamma({\mathcal G}_{{\mbox{\scriptsize red}}})$.
On $\Gamma_{13}^{\prime}$, we have 3 real ordered marked points, with $\zeta$--coordinates: $\zeta(\lambda_3^{(2)})=0<\zeta(\lambda_2^{(2)})=1<\zeta(\lambda_1^{(2)})=\infty$.
Comparing with (\ref{eq:lines}) we then easily conclude that
\[
\lambda = \frac{\kappa_1 (\kappa_4-\kappa_2) \zeta +(\kappa_2-\kappa_1)\kappa_4}{(\kappa_4-\kappa_2) \zeta +(\kappa_2-\kappa_1)}.
\]
On $\Sigma_{23}^{\prime}$, we have 3 real ordered marked points, and the following constraints: $\mu(\lambda_1^{(1)})=\mu(\kappa_2)=0$, $\mu(\alpha_1)=\kappa_2-\kappa_3$, $\mu(\alpha_2)=\mu(\lambda_2^{(2)})=c_{13}(\kappa_2-\kappa_1)$.
Similarly on $\Sigma_{24}^{\prime}$, we have 3 real ordered marked points, and the following constraints: $\mu(\lambda_3^{(1)})=\mu(\kappa_4)=0$, $\mu(\alpha_4)=\kappa_4-\kappa_3$, $\mu(\alpha_3)=\mu(\lambda_1^{(2)})=c_{13}(\kappa_4-\kappa_1)$.
Analogously, on $\Gamma_{23}$ in the initial $\zeta$ coordinates we have 3 real ordered marked points and $\zeta(\alpha_4)=0<\zeta(\lambda_2^{(1)})=1<\zeta(\alpha_1)=\infty$, therefore the fractional linear transformation to the $\lambda$ is:
\[
\lambda =\frac{\kappa_2 (\kappa_4-\kappa_3) \zeta +(\kappa_3-\kappa_2)\kappa_4}{(\kappa_4-\kappa_3) \zeta +(\kappa_3-\kappa_2)}.
\]
We remark that $\Gamma^{\prime}_{13}$ and $\Gamma_{23}$ intersect at 
\begin{equation}\label{eq:alpha_5}
\alpha_5 = (\lambda_5,\mu_5) = \left( -\frac{\kappa_3-c_{13}\kappa_1}{c_{13}-1} , -\frac{c_{13}(\kappa_3-\kappa_1)}{c_{13}-1}\right),
\end{equation}
which does not correspond to any of the marked points of the topological model of $\Gamma$ (see Figure \ref{fig:gr24_top_2}). Such singularity is resolved in the partial normalization and therefore there are no extra conditions to be satisfied by the vacuum and the dressed wave functions at $\alpha_5$.
Finally
$\Gamma({\mathcal G}_{{\mbox{\scriptsize red}}})$ is represented by the reducible plane curve $\Pi_0(\lambda,\mu)=0$, with
\begin{equation}
\label{eq:curveGr24}
\Pi_0(\lambda,\mu)=\mu\cdot\big(\mu-c_{13}(\lambda-\kappa_1)\big)\cdot\big(\lambda-\kappa_2\big)\cdot\big(\mu-\lambda+\kappa_3\big)\cdot\big(\lambda-\kappa_4\big).
\end{equation}

The desingularization of $\Gamma({\mathcal G}_{{\mbox{\scriptsize red}}})$ gives a genus 4 $\mathtt M$--curve $\Gamma_{\varepsilon}$ for the following choice ($0<\varepsilon \ll 1$):
\begin{equation}
\label{eq:curveGr24_pert}
\Gamma_{\varepsilon} \; : \quad\quad \Pi(\lambda, \mu;\varepsilon)= \left\{
\begin{array}{ll} \Pi_0(\lambda,\mu)+\varepsilon^2\left(\lambda -\lambda_5\right)^2C_0(\lambda,\mu)=0, & \mbox{ if } c_{13}\not =1,\\
\Pi_0(\lambda,\mu)+\varepsilon^2C_0(\lambda,\mu)=0, & \mbox{ if } c_{13} =1,
\end{array}
\right.
\end{equation}
where $\lambda_5$ is as in (\ref{eq:alpha_5}) and $C_0$ is a cubic polynomial in $\lambda,\mu$,
\[
C_0(\lambda,\mu)= \beta_0 +\beta_{1,0} \lambda + \beta_{0,1} \mu + \beta_{2,0} \lambda^2 +\beta_{2,1} \lambda\mu +\beta_{0,2} \mu^2 +
\beta_{3,0} \lambda^3 +\beta_{2,1} \lambda^2 \mu +\beta_{1,2} \lambda\mu^2+\beta_{0,3}\mu^3,
\]
of which we control the sign at the marked points. If $\kappa_2+\kappa_4-2\kappa_3<0$, the conditions are
\[
\mbox{sign }\left( C_0(\kappa_1) \right)= \mbox{sign }\left( C_0(\alpha_2)\right)\not =\mbox{sign }\left( C_0(\kappa_j)\right),\mbox{sign }\left( C_0(\alpha_l)\right),\quad \mbox{ for } j=2,3,4, \,\, l=1,3,4,
\]
(see also Figure \ref{eq:curveGr24_pert} [bottom, left]). The case $\kappa_2+\kappa_4-2\kappa_3>0$ can be treated similarly.
For instance, possible choices are $C_0 = -(35\lambda^3+\mu^3+70\lambda^2)$ for ${\mathcal K} = \{-3, -1, 2, 3 \}$, and $C_0 = \mu^2-35$ for ${\mathcal K} = \{-3, -1, 2, 6 \}$.

In Figure \ref{fig:gr24_top_3}[left], we represent the topological model of $\Gamma ({\mathcal G}_{{\mbox{\scriptsize red}}})$ [top], which is the partial normalization of the plane curve $\Pi_0=0$ [middle] and the oval structure of the plane curve $\Gamma_{\varepsilon}$ as in (\ref{eq:curveGr24_pert}) when $\kappa_2+\kappa_4-2\kappa_3<0$. We remark that the normalized KP wave function takes the same value at all points marked with the same symbol in Figures \ref{fig:gr24_top_1} and \ref{fig:gr24_top_3}.

\subsection{Construction of $\Gamma(\xi)$ as in \cite{AG1} starting from $\Gamma({\mathcal G}_{{\mbox{\scriptsize red}}})$}\label{sec:gamma_24_xi}

In \cite{AG1} the reducible curve  $\Gamma(\xi)=\Gamma_0\sqcup\Gamma_1(\xi)\sqcup\Gamma_2(\xi)$ is obtained gluing the three rational irreducible components, $\Gamma_0,\Gamma_1(\xi)$ and $\Gamma_2(\xi)$, at prescribed real points whose position is ruled by the parameter $\xi\gg 1$. 
In Figure \ref{fig:gr24_top_3} [top,right], we represent the topological model of $\Gamma(\xi)$ for soliton data in $Gr^{\mbox{\tiny TP}}(2,4)$ and call $\tilde \zeta$ the local coordinate introduced in Section 4 in \cite{AG1}. For instance the point $\alpha_2\in \Gamma_1(\xi)$ has local coordinate $\tilde \zeta (\alpha_2)=\xi^{-1}$ and is glued to the point $\lambda^{(2)}_2\in \Gamma_2(\xi)$ having local coordinate $\tilde \zeta (\lambda_2^{(2)})=-1$.

The plane curve $\Gamma(\xi)$ constructed in this section satisfies $\Gamma(\infty) =\Gamma ({\mathcal G}_{{\mbox{\scriptsize red}}})$, so it is more suitable than the one presented in \cite{AG1} for comparing the two models curves and the KP divisors:

\begin{figure}
  \centering
	  {\includegraphics[width=0.42\textwidth]{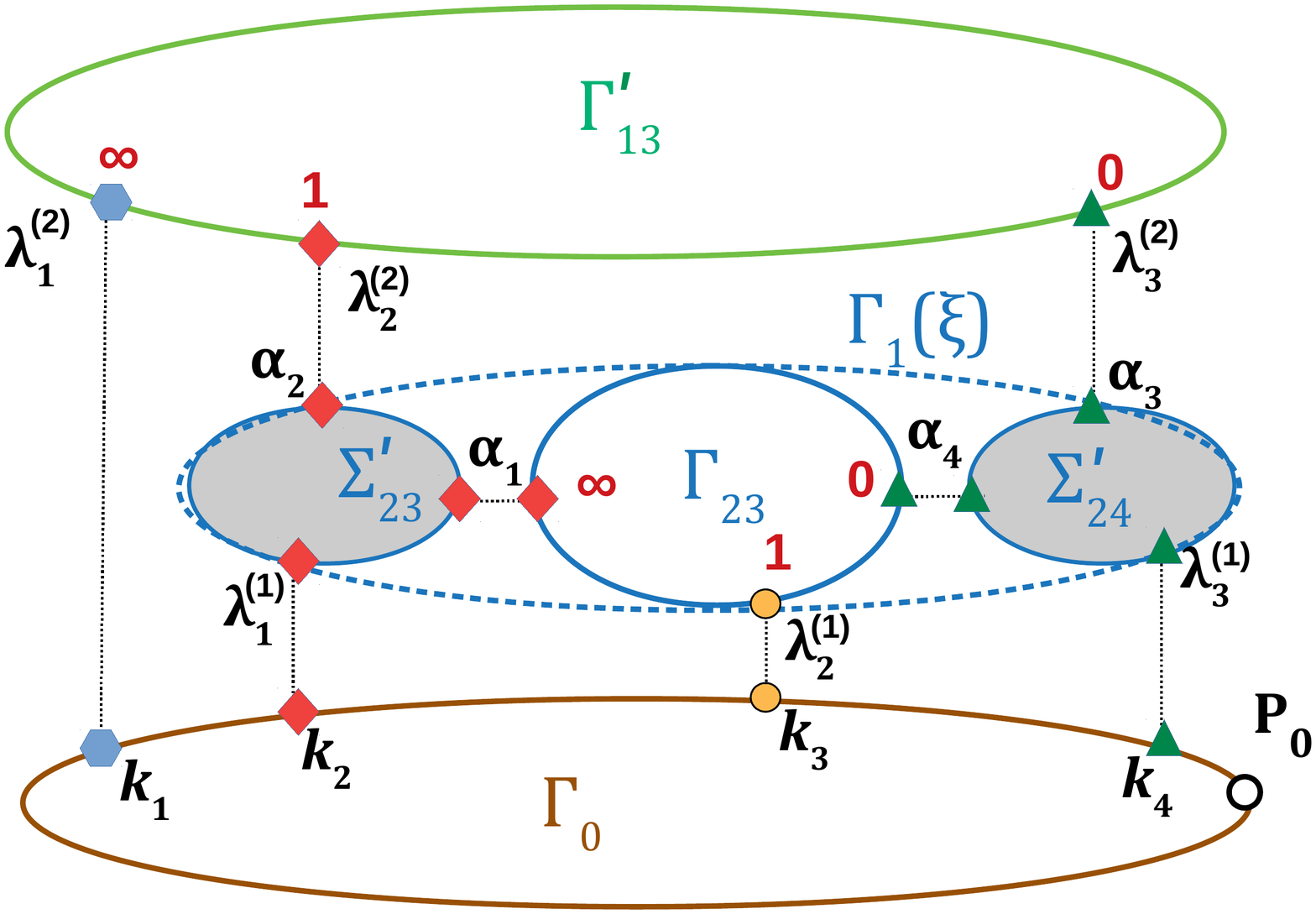}}
  \hspace{.6 truecm}
  {\includegraphics[width=0.42\textwidth]{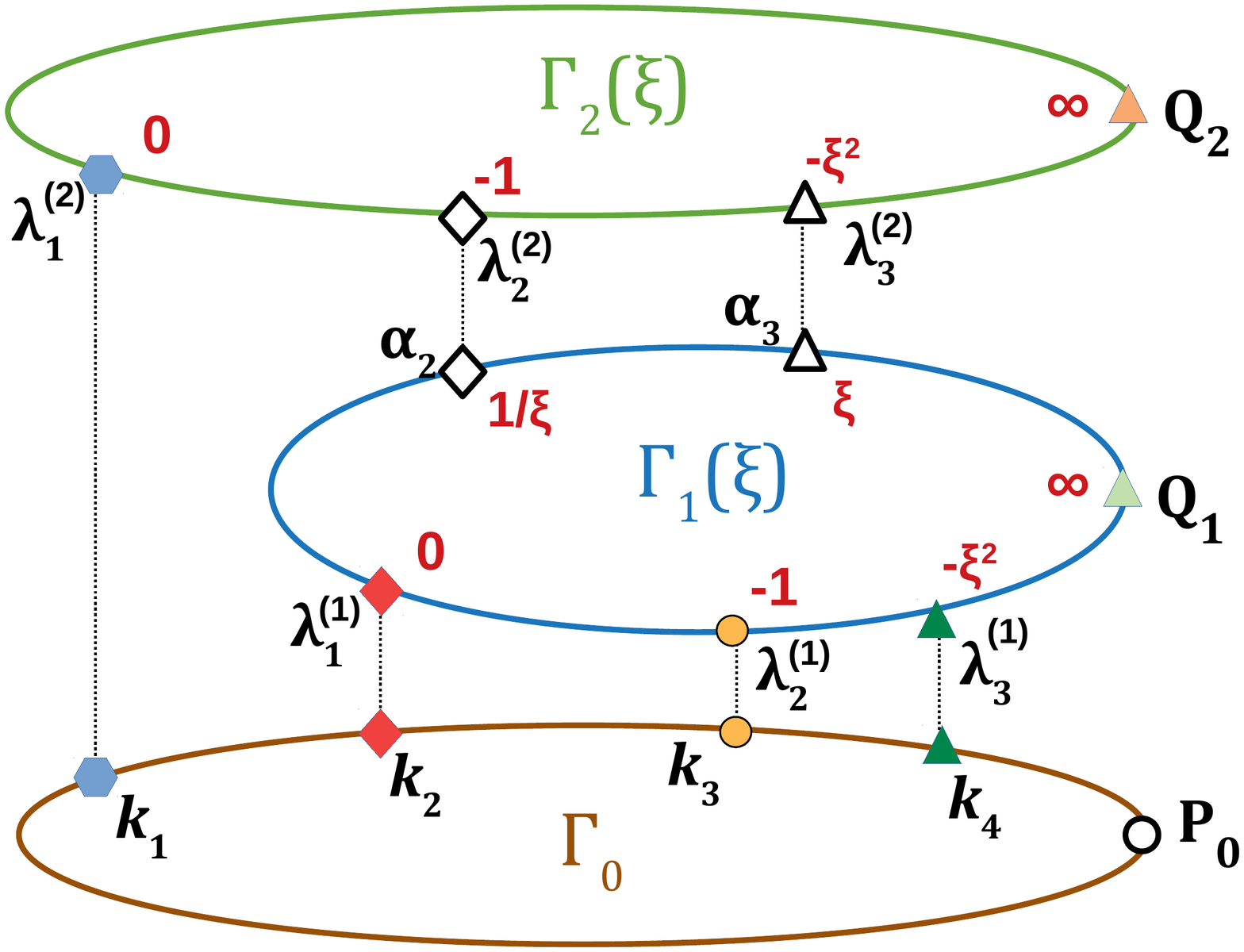}}
	{\includegraphics[width=0.44\textwidth]{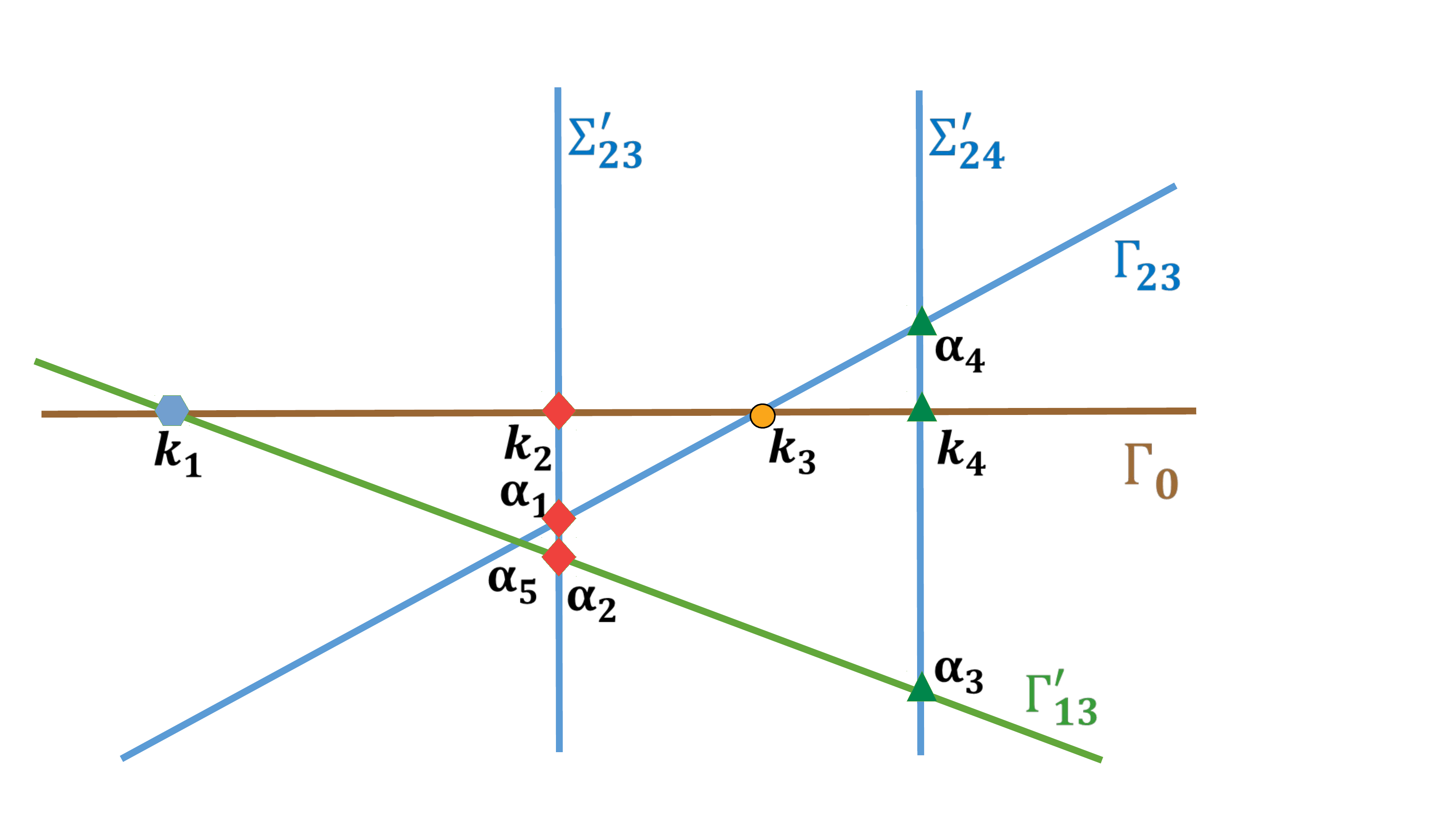}}
	 \hspace{.6 truecm}
	{\includegraphics[width=0.44\textwidth]{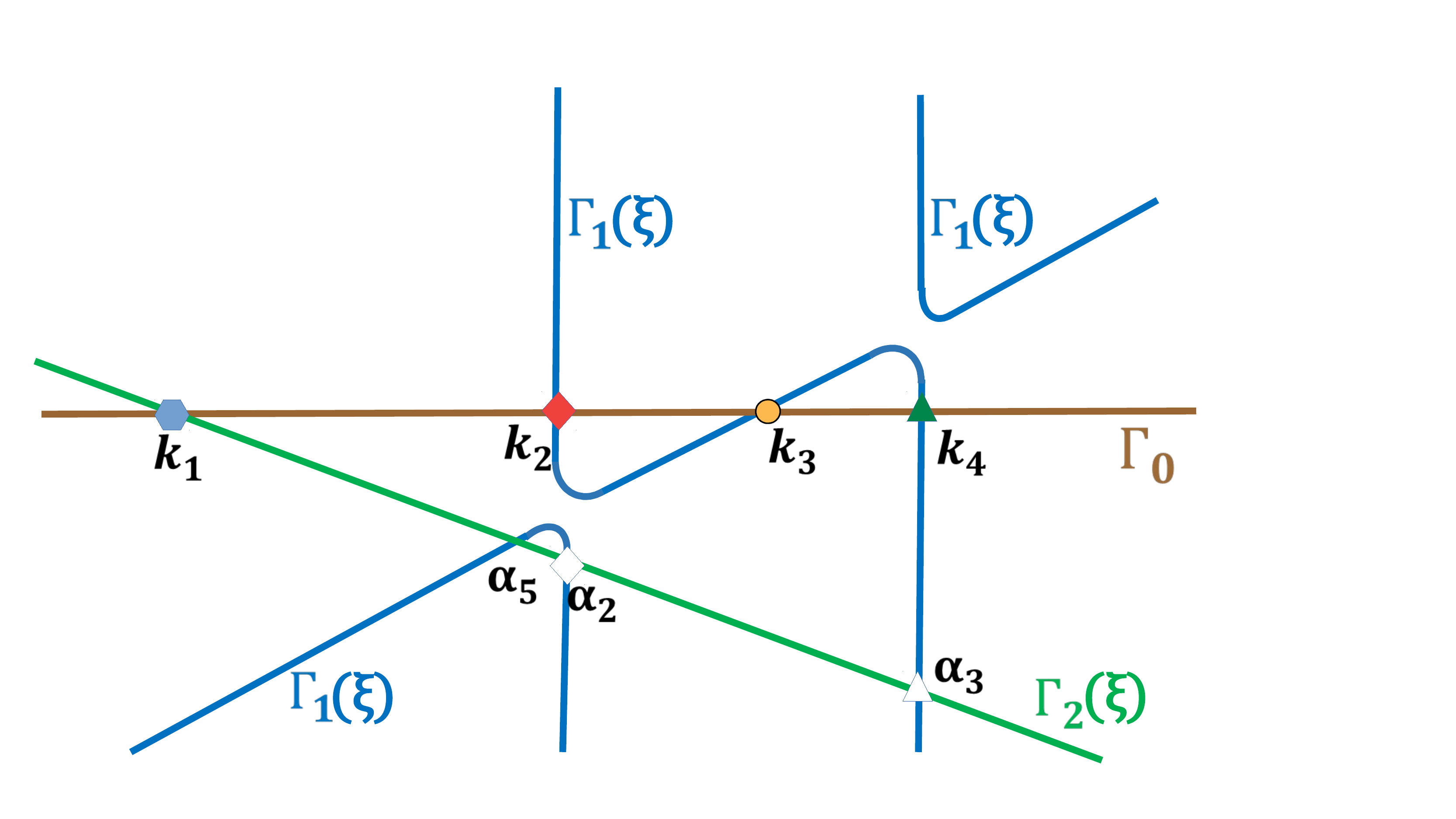}}
		{\includegraphics[width=0.44\textwidth]{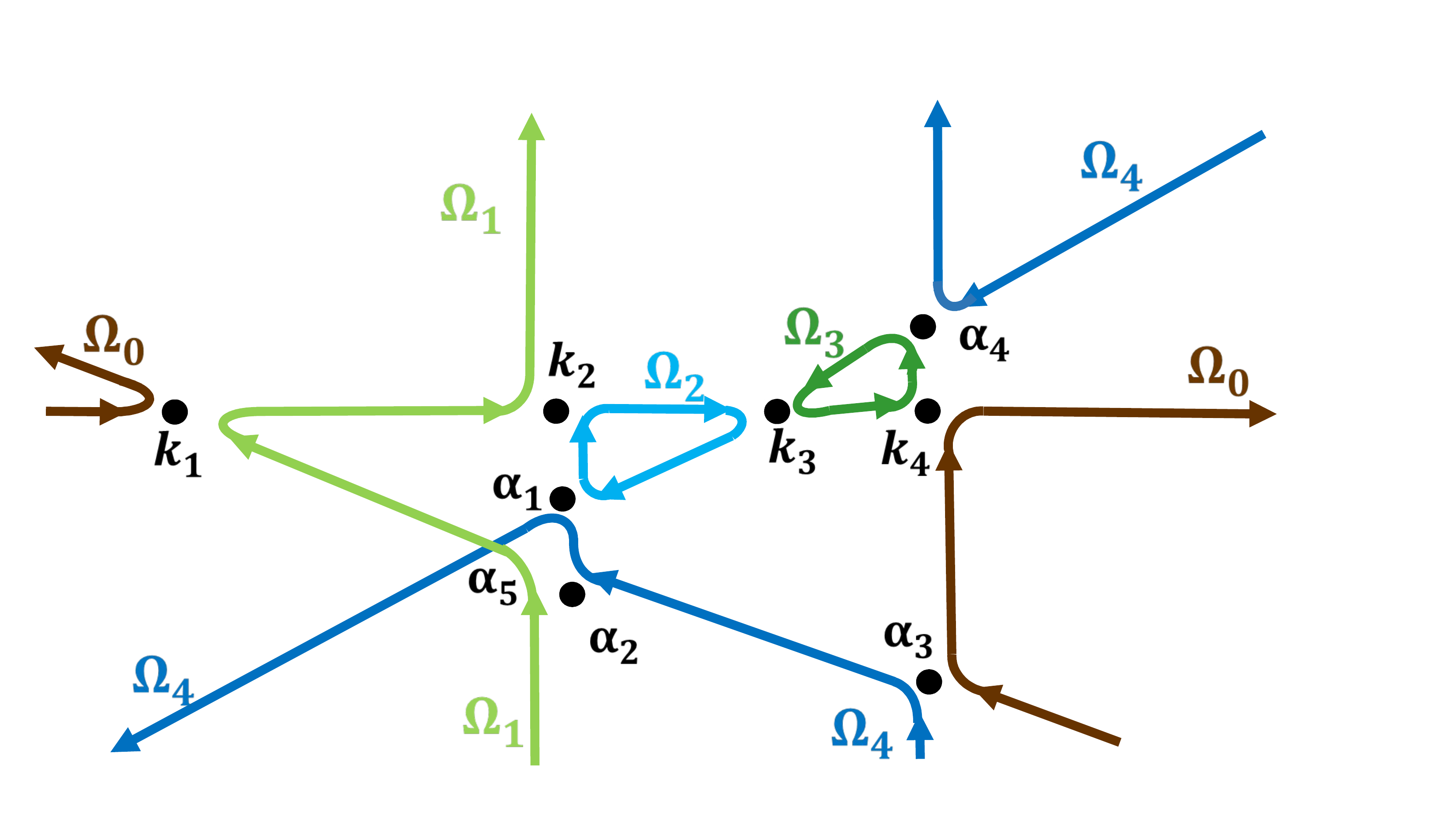}}
	 \hspace{.6 truecm}
	{\includegraphics[width=0.44\textwidth]{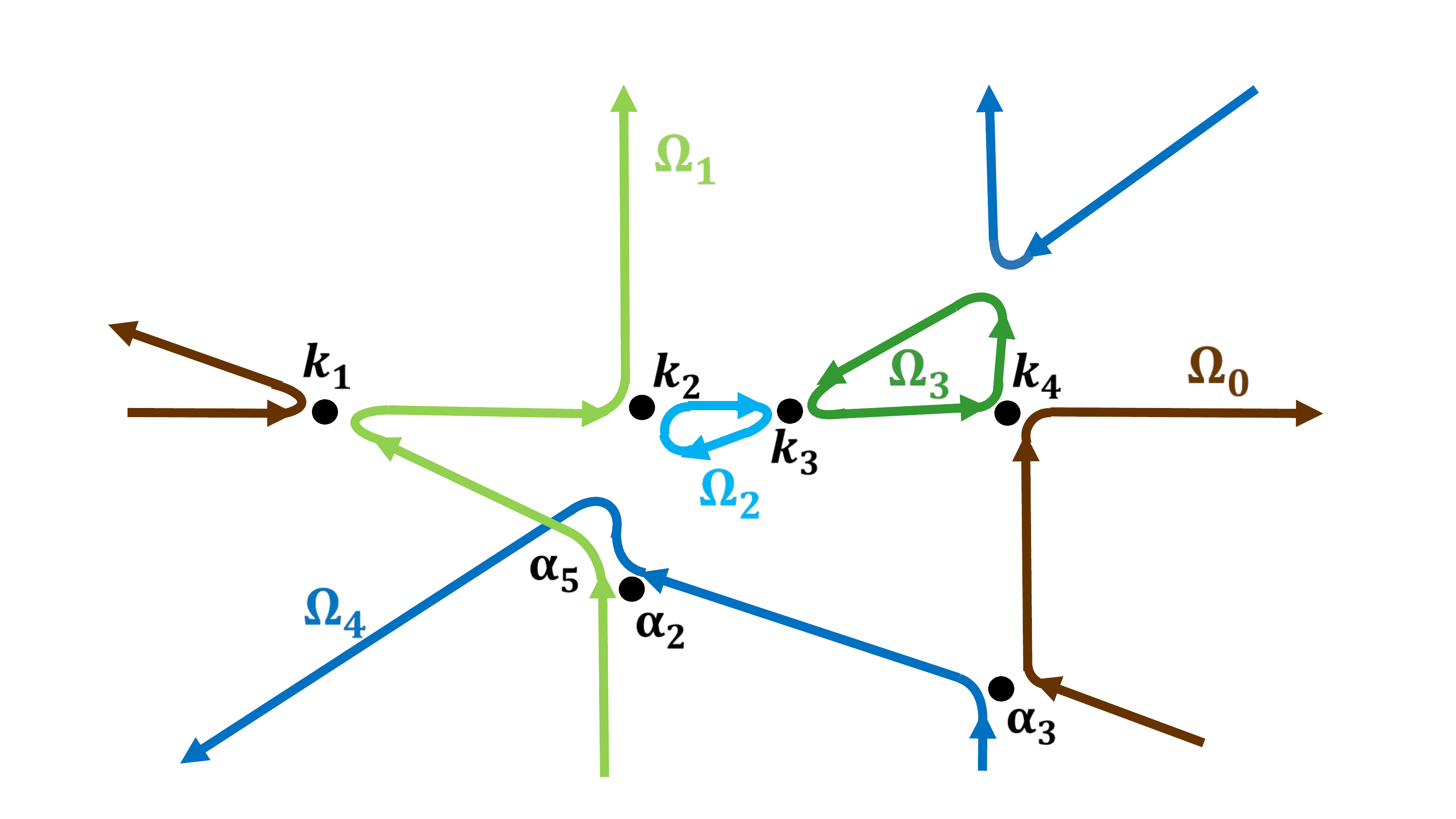}}
\caption{\footnotesize{\sl For the same soliton data in $Gr^{\mbox{\tiny TP}}(2,4)$ we compare the topological model of the spectral curve
[top], its realization as a plane curve [center] and the oval structure of the smooth genus 4 $\mathtt M$--curve [bottom] for
the reducible rational curves $\Gamma({\mathcal G}_{{\mbox{\scriptsize red}}})=\Gamma_0\sqcup\Gamma_{13}^{\prime}\sqcup\Gamma_{23}\sqcup\Sigma^{\prime}_{23}\sqcup\Sigma^{\prime}_{24}$ [left] and $\Gamma(\xi)=\Gamma_0\sqcup\Gamma_{1}(\xi)\sqcup\Gamma_{2}(\xi)$ [right]. In all figures the value of the normalized KP wave function for given soliton data in $Gr^{\mbox{\tiny TP}}(2,4)$ is equal for all times at the points marked with the same symbol.}}\label{fig:gr24_top_3}        
\end{figure}

\begin{enumerate}
\item  $\Gamma_0$ is the same rational component both for $\Gamma ({\mathcal G}_{{\mbox{\scriptsize red}}})$ and $\Gamma(\xi)$, and we represent it by the line $\mu=0$ in both cases;
\item $\Gamma_1(\xi)$ is a rational component which may be obtained from $\Gamma_{23}\sqcup\Sigma_{23}\sqcup\Sigma_{24}$ desingularising it at the double points $\alpha_1$ and $\alpha_4$. This transformation effects the position of the marked points $\alpha_j\equiv\alpha_j(\xi)$ and $j=2,3$, which have $\lambda$--coordinates $\lambda_{2,\xi}\equiv \lambda(\alpha_2)=\kappa_2+O(\xi^{-1})$ and $\lambda_{3,\xi}\equiv\lambda(\alpha_3)=\kappa_4+O(\xi^{-1})$ (compare with Section 7 in \cite{AG1}). Then $\Gamma_1(\xi)$ may be represented by a cubic depending on $\xi$ which tends to the product of the three lines when $\xi\to\infty$.
\item $\Gamma_2(\xi)$ has three marked points: $\lambda^{(2)}_1$ is glued to $\kappa_1$, while $\lambda^{(2)}_j$, $j=2,3$ are glued
 respectively to $\alpha_j$, $j=2,3$. Then $\Gamma_2(\xi)$ may be represented by a line which tends to that representing $\Gamma _{13}^{\prime}$ when $\xi\to\infty$.
\end{enumerate}
The above ansatz is consistent with the fact that, for any fixed $\xi\gg 1$ and for all times, the normalized KP wave function on $\Gamma(\xi)$ coincides exactly with that of $\Gamma ({\mathcal G}_{{\mbox{\scriptsize red}}})$ at the marked points $\kappa_j$, $j\in [4]$, and at leading order in $\xi$ at the marked points $\alpha_2$, $\alpha_3$ (see Figure \ref{fig:gr24_top_3}).

Let us now construct a plane curve $\Gamma(\xi)$ satisfying the above requirements.
$\Gamma_0$ is the line $\mu=0$ as in the previous section. We require that $\Gamma_2(\xi)$ is a line passing through the point $(\kappa_1,0)$ 
\begin{equation}\label{eq:gamma_2}
\Gamma_2(\xi) \; : \; \mu -c_{13}(\xi) (\lambda-\kappa_1)=0,
\end{equation}
where $c_{13}(\xi)$ tends to $c_{13}$ as in (\ref{eq:lines}) in the limit $\xi\to \infty$, so that the line $\Gamma_2(\infty)$ coincides with that representing $\Gamma_{13}^{\prime}$.

$\Gamma_1(\xi)$ is a cubic which degenerates to the product of the three lines representing $\Gamma_{23}$,$\Sigma_{23}^{\prime}$ and $\Sigma_{24}^{\prime}$ in the limit $\xi\to \infty$. We make the following choice
\begin{equation}\label{eq:gamma_1}
\Gamma_1(\xi) \; : \; \frac{\mu}{\xi}-c_{23} (\xi) (\lambda-\kappa_2)(\lambda-\kappa_3-\mu)(\lambda-\kappa_4) =0.
\end{equation}
The coefficients $c_{13}(\xi)$, $c_{23} (\xi)$ are uniquely defined by the conditions that $\Gamma_1(\xi)$ and $\Gamma_2(\xi)$ intersect at the real points $\alpha_2 \equiv \alpha_2 (\xi)$, $\alpha_3\equiv \alpha_3 (\xi)$ such that in the local coordinate $\tilde \zeta$ used in \cite{AG1} 
$\tilde \zeta(\alpha_2) = \xi^{-1}$, $\tilde \zeta(\alpha_3)=\xi$ (see also Figure \ref{fig:gr24_top_3} [top,right]). In the coordinate $\lambda$ used throughout this section, we have
\[
\resizebox{\textwidth}{!}{$
\lambda_{2,\xi}\equiv \lambda(\alpha_2) = \frac{\xi^2\kappa_2(\kappa_4-\kappa_3)+(1-\xi)\kappa_3(\kappa_4-\kappa_2)}{\xi^2(\kappa_4-\kappa_3)+(1-\xi)(\kappa_4-\kappa_2)}, \quad \lambda_{3,\xi}\equiv \lambda(\alpha_3) =\frac{\xi\kappa_4(\kappa_3-\kappa_2)-\kappa_3(\kappa_4-\kappa_2) }{\xi(\kappa_3-\kappa_2)-\kappa_4+\kappa_2}.
$}
\]
The third real intersection point between $\Gamma_1(\xi)$ and $\Gamma_2(\xi)$, $\alpha_5\equiv\alpha_5(\xi)$, has then $\lambda$ 
coordinate
\begin{equation}\label{eq:alpha_5_xi}
\resizebox{\textwidth}{!}{$
\lambda_{5,\xi}\equiv\lambda(\alpha_5) = \frac{( \lambda_{2,\xi}\lambda_{3,\xi}\kappa_1- (\lambda_{2,\xi}+\lambda_{3,\xi}) (\kappa_1\kappa_2+\kappa_1\kappa_4-\kappa_2\kappa_4)-\kappa_2^2(\kappa_4-\kappa_1)-\kappa_4^2(\kappa_2-\kappa_1)+\kappa_1\kappa_2\kappa_4}{ \lambda_{2,\xi}\lambda_{3,\xi}-\kappa_1(\lambda_{2,\xi}+\lambda_{3,\xi})+\kappa_1\kappa_2+\kappa_1\kappa_4-\kappa_2\kappa_4}
$}
\end{equation}
while the coefficients satisfy
\[
\resizebox{\textwidth}{!}{$
\begin{array}{ll}
c_{13} (\xi) &=\displaystyle \frac{\displaystyle\lambda_{2,\xi}\lambda_{3,\xi}(\lambda_{2,\xi}+\lambda_{3,\xi}-\sum_{j=1}^4\kappa_j)+\kappa_1(\lambda_{2,\xi}+\lambda_{3,\xi})\sum_{j\not=1}\kappa_j-\kappa_1(\lambda_{2,\xi}^2+\lambda_{3,\xi}^2+\sum_{1<i<j\le 4}\kappa_i\kappa_j)+\kappa_2\kappa_3\kappa_4}{(\lambda_{2,\xi}-\kappa_1)(\lambda_{3,\xi}-\kappa_1)(\lambda_{2,\xi}+\lambda_{3,\xi}-\kappa_2-\kappa_4)}= c_{13} + O(\xi^{-1}),\\
&\\
c_{23} (\xi) &=-\displaystyle\frac{c_{13}(\xi)( \lambda_{2,\xi}-\kappa_1)}{\xi( \lambda_{2,\xi}-\kappa_4)( \lambda_{2,\xi}-\kappa_2)(c_{13}(\xi) ( \lambda_{2,\xi}-\kappa_1)- \lambda_{2,\xi}+\kappa_3)}=\displaystyle -\frac{c_{13}(\kappa_4-\kappa_1)(\kappa_2-\kappa_1)(\kappa_2+\kappa_4-2\kappa_3)}{(\kappa_3-\kappa_1)(\kappa_3-\kappa_2)(\kappa_4-\kappa_3)(\kappa_4-\kappa_2)^2} + O(\xi^{-1}).
\end{array}
$}
\]
Therefore $\Gamma(\xi)$ is represented by the plane curve
\begin{equation}\label{eq:gr24_xi_unp}
\Gamma(\xi) \, : \quad\quad \Pi_{\xi} (\lambda, \mu) = \mu \Big(\mu-c_{13}(\xi) (\lambda-\kappa_1) \Big)\left(\frac{\mu}{\xi}- c_{23} (\xi)(\lambda-\kappa_2)(\lambda-\kappa_3 -\mu)(\lambda -\kappa_4) \right) =0.
\end{equation}
We represent both the plane curve (Figure \ref{fig:gr24_top_3}[right,middle]) and its partial normalization (Figure \ref{fig:gr24_top_3}[right,top]) in the case $\kappa_2+\kappa_4-2\kappa_3<0$.

The genus 4 $\mathtt M$--curve, $\Gamma_{\varepsilon}(\xi)$, obtained from $\Gamma(\xi)$ is then a perturbation of $\Gamma_{\varepsilon}$,
under the genericity assumption $c_{13}\not =1$:
\begin{equation}
\label{eq:curveGr24_pert_xi}
\Gamma_{\varepsilon}(\xi) \; : \quad\quad \Pi_{\xi}(\lambda, \mu;\varepsilon)= \Pi_{\xi}(\lambda,\mu) +\varepsilon^2\left(\lambda -\lambda_{5,\xi}^2\right)^2C_0(\lambda,\mu)=0, \quad\quad
0<\varepsilon \ll 1,
\end{equation}
where $\lambda_{5,\xi}$ is as in (\ref{eq:alpha_5_xi}), and the coefficients of the cubic polynomial $C_0$ coincide with those in (\ref{eq:curveGr24_pert_xi}), for $\xi$ fixed and sufficiently big.
In Figure \ref{fig:gr24_top_3}[right,bottom], we present the oval structure of $\Gamma_{\varepsilon}(\xi)$ as in (\ref{eq:curveGr24_pert_xi}) in the case $\kappa_2+\kappa_4-2\kappa_3<0$. 

\subsection{The KP divisors on $\Gamma({\mathcal G}_{{\mbox{\scriptsize red}}})$ and on $\Gamma(\xi)$}\label{sec:div_24}

In \cite{AG1}, \cite{AG3} and \cite{A2} we have computed the vacuum and the KP divisor for soliton data in $Gr^{\mbox{\tiny TP}}(2,4)$ 
respectively on $\Gamma(\xi)$, on $\Gamma({\mathcal G}_{{\mbox{\scriptsize red}}})$ and on $\Gamma({\mathcal G})$; therefore we do not repeat this computation here. We just verify that the KP divisors $\DKP= (\Ps_1,\Ps_2,\Pdr_{1,\xi},\Pdr_{2,\xi})$ on $\Gamma(\xi)$ and $\DKP= (\Ps_1,\Ps_2,\Pdr_{13},\Pdr_{23})$ on $\Gamma({\mathcal G}_{{\mbox{\scriptsize red}}})$, in the coordinate $\zeta$ introduced in Definition \ref{def:loccoor}, satisfy
\begin{equation}\label{eq:rel_div}
\zeta(\Pdr_{1,\xi}) = \zeta(\Pdr_{23}),\quad\quad \zeta(\Pdr_{2,\xi}) = \zeta(\Pdr_{13}) + O(\xi^{-1}).
\end{equation}
In the following we use the abridged notation $e^{\theta_{j,0}} = e^{\theta_j(\vec t_0)} =e^{\kappa_j x_0}$.
By construction $\DKP$ consists of the degree $2$ Sato divisor $(\Ps_1,\Ps_2)$ defined in (\ref{eq:Satodiv}),
$\zeta(\Ps_l)= \gs_l$, $l=1,2$, and of $2$ simple poles $(\Pdr_{13},\Pdr_{23})$ respectively belonging to $\Gamma_{13}$ and  $\Gamma_{23}$. In the local coordinates induced by the orientation of the Le--network, we have \cite{AG3,A2}
\begin{equation}\label{eq:div_GN}
\begin{array}{ll}
\zeta(\Ps_1) +\zeta(\Ps_2) = {\mathfrak w}_1 (\vec t_0), &\quad  \zeta(\Ps_1)\zeta(\Ps_2) = -{\mathfrak w}_2 (\vec t_0),\\
\displaystyle\zeta(\Pdr_{13}) = \frac{w_{14} ( {\mathfrak D} e^{\theta_{2,0}}+ w_{23} {\mathfrak D} e^{\theta_{3,0}})}{
(w_{14}+w_{24}){\mathfrak D} e^{\theta_{2,0}}+ w_{23}w_{14} {\mathfrak D} e^{\theta_{3,0}}}, 
&\quad\displaystyle
\zeta(\Pdr_{23}) =1 + w_{23}\frac{ {\mathfrak D} e^{\theta_{3,0}}}{{\mathfrak D} e^{\theta_{2,0}}}.
\end{array}
\end{equation}
In \cite{AG3,A2}, we have also discussed the position of the divisor points in dependence on the signs of ${\mathfrak D} e^{\theta_{2,0}}$ and ${\mathfrak D} e^{\theta_{3,0}}$.

The relation between the local coordinates $x_{i,j}$ used in \cite{AG1} and the weights $w_{ij}$ is
\[
x_{1,1} = w_{23}, \quad x_{1,2} =w_{23}w_{24},\quad x_{2,1} = w_{13}, \quad x_{2,2} = w_{13}w_{14}w_{23}.
\]
Applying the Darboux transformation ${\mathfrak D}$ to the vacuum wave function on $\Gamma(\xi)$ computed in \cite{AG1}, it is not difficult to prove that on $\Gamma_1(\xi)$ the divisor point $\Pdr_{1,\xi}$ has ${\tilde \zeta}$--coordinate
\begin{equation}\label{eq:div_1}
{\tilde \zeta}(\Pdr_{1,\xi}) = -\frac{\xi^2 {\mathfrak D} e^{\theta_{2,0}}}{\xi^2{\mathfrak D} e^{\theta_{2,0}}+ w_{23}(\xi^2-1){\mathfrak D}e^{\theta_{3,0}}}.
\end{equation}
and that on $\Gamma_2(\xi)$ the divisor point $\Pdr_{2,\xi}$ has ${\tilde \zeta}$--coordinate
\begin{equation}\label{eq:div_2}
{\tilde \zeta}(\Pdr_{2,\xi}) =-\frac{\xi^3(\xi-1)\left[ (w_{14}+w_{24}) {\mathfrak D} e^{\theta_{2,0}}+w_{14}w_{23} {\mathfrak D} e^{\theta_{3,0}}\right]}{w_{23}(\xi-1)[w_{14}(\xi^3 +\xi+1) + w_{24}(1-\xi^2)]{\mathfrak D} e^{\theta_{3,0}}+\xi^2[ (\xi^2-1)w_{14}+(1-\xi)w_{24}]e^{\theta_{2,0}}}.
\end{equation}
The value of the KP wave function is the same at all points marked with the same symbol in Figure \ref{fig:gr24_top_3} and coincides at leading order in $\xi$ at the points $\alpha_2$ and $\alpha_3$ for $\xi\gg 1$. Therefore, for any $\xi\gg 1$, the KP divisor point $\Pdr_{1,\xi}\in \Gamma(\xi)$ coincides with the divisor point $\Pdr_{23}\in \Gamma({\mathcal G}_{{\mbox{\scriptsize red}}})\equiv\Gamma(\infty)$, while $\Pdr_{2,\xi}\in \Gamma(\xi)$ coincides \textbf{at leading order in $\xi$} with the divisor point $\Pdr_{13}\in \Gamma({\mathcal G}_{{\mbox{\scriptsize red}}})\equiv\Gamma(\infty)$. The relation between the coordinate $\tilde \zeta$ used in \cite{AG1} and the coordinate $\zeta$ at the marked points $\lambda^{(1)}_{s}$ in $\Gamma_{23}$ (resp. $\lambda^{(2)}_{s}$ in $\Gamma_{23}$), $s\in [3]$, is the following: $\zeta =\infty,1,0$, respectively correspond to ${\tilde z}=0,-1,-\xi^2$. Finally, inserting the fractional linear transformation
\[
{\tilde \zeta} = \frac{\xi^2}{(1-\xi^2)\zeta -1},
\] 
in (\ref{eq:div_1}) and (\ref{eq:div_2}) and using (\ref{eq:div_GN}), it is straightforward to verify (\ref{eq:rel_div}).

\appendix

\section{The totally nonnegative Grassmannian}\label{app:TNN}

In this appendix we recall some useful definitions and theorems from \cite{Pos} to make the paper self--contained. For more details on the topological properties of $\Grkn$ and on generalizations of total positivity to reductive groups we refer to \cite{Lus1,Lus2,MR,PSW,Rie}. In particular we use Postnikov rules to represent each Le--tableau $D$ by a unique  bipartite trivalent oriented network ${\mathcal N}$ in the disk. In Section~\ref{sec:gamma} we use the Le--graph ${\mathcal G}$ of ${\mathcal N}$
to construct a curve $\Gamma({\mathcal G})$, which is a rational degeneration of a 
smooth ${\mathtt M}$-curve of genus equal to the dimension $d$ of the corresponding positroid cell. 

\begin{definition}{\bf The totally non--negative part of $Gr(k,n)$ \cite{Pos}.}
The totally non--negative Grassmannian $\Grkn$ is the subset of the Grassmannian $Gr(k,n)$  with all Pl\"ucker coordinates non-negative, 
i.e. it may be defined as the following quotient: $\Grkn = GL^+_k \backslash \mbox{ Mat}^{\mbox{\tiny TNN}}_{kn}$. Here 
$GL^+_k$ is the group of $k\times k$ matrices with positive determinant, and $\mbox{ Mat}^{\mbox{\tiny TNN}}_{kn}$ is the set of real $k\times n$ matrices $A$ of rank $k$
such that all maximal minors are non-negative, i.e. $\Delta_I (A) \ge 0$, for all $k$--element subsets $I\subset [n]$. 

The totally positive Grassmannian $Gr^{\mbox{\tiny TP}} (k,n)\subset \Grkn$ is the subset of $Gr(k,n)$ whose elements may be represented by $k\times n$ matrices with all strictly positive maximal minors $\Delta_I(A)$.
\end{definition}

It is well-known that $Gr(k,n)$ is decomposed into a disjoint union of Schubert cells $\Omega_{\lambda}$ indexed by partitions $\lambda\subset (n-k)^k$ whose Young diagrams fit inside the $k\times (n-k)$ rectangle (we use the so-called English notation in our text). A refinement of this decomposition was proposed in \cite{GGMS,GS}. Intersecting the matroid strata with 
$\Grkn$ one obtains the totally non-negative Grassmann cells \cite{Pos} with the following property: each cell is birationally equivalent to an 
open octant of appropriate dimension, and this birational map is also a topological homeomorphism.

Let us recall these constructions. The Schubert cells are indexed by partitions, or, equivalently by pivot sets. To each 
partition  $\lambda=(\lambda_1,\ldots,\lambda_k)$, $n-k\ge\lambda_1\ge\lambda_2\ldots\ge\lambda_k\ge0$, $\lambda_j\in\Z$ 
there is associated a pivot set $I(\lambda) =\{ 1\le i_1 < \cdots < i_k \le n\}$ defined by the following relations: 
\begin{equation}
\label{eq:young}
i_j = n-k+j-\lambda_j, \ \ j\in[k].
\end{equation}
Each Schubert cell is the union of all Grassmannian points sharing the same set of pivot columns. Therefore,
any point in $\Omega_{\lambda}$ with pivot set $I$ can be represented by a matrix in canonical reduced row echelon form, {\sl i.e.} a matrix $A$ such that $A^i_{i_l}=1$ for $l\in [k]$ and all the 
entries to the left of these 1's are zero.
The Young diagram representing the Schubert cell $\Omega_{\lambda}$ is a collection of boxes arranged in $k$ rows, 
aligned on the left such that the $j$-th row contains $\lambda_j$ boxes, $j\in[k]$. 

\begin{figure}
  \centering
  {\includegraphics[width=0.46\textwidth]{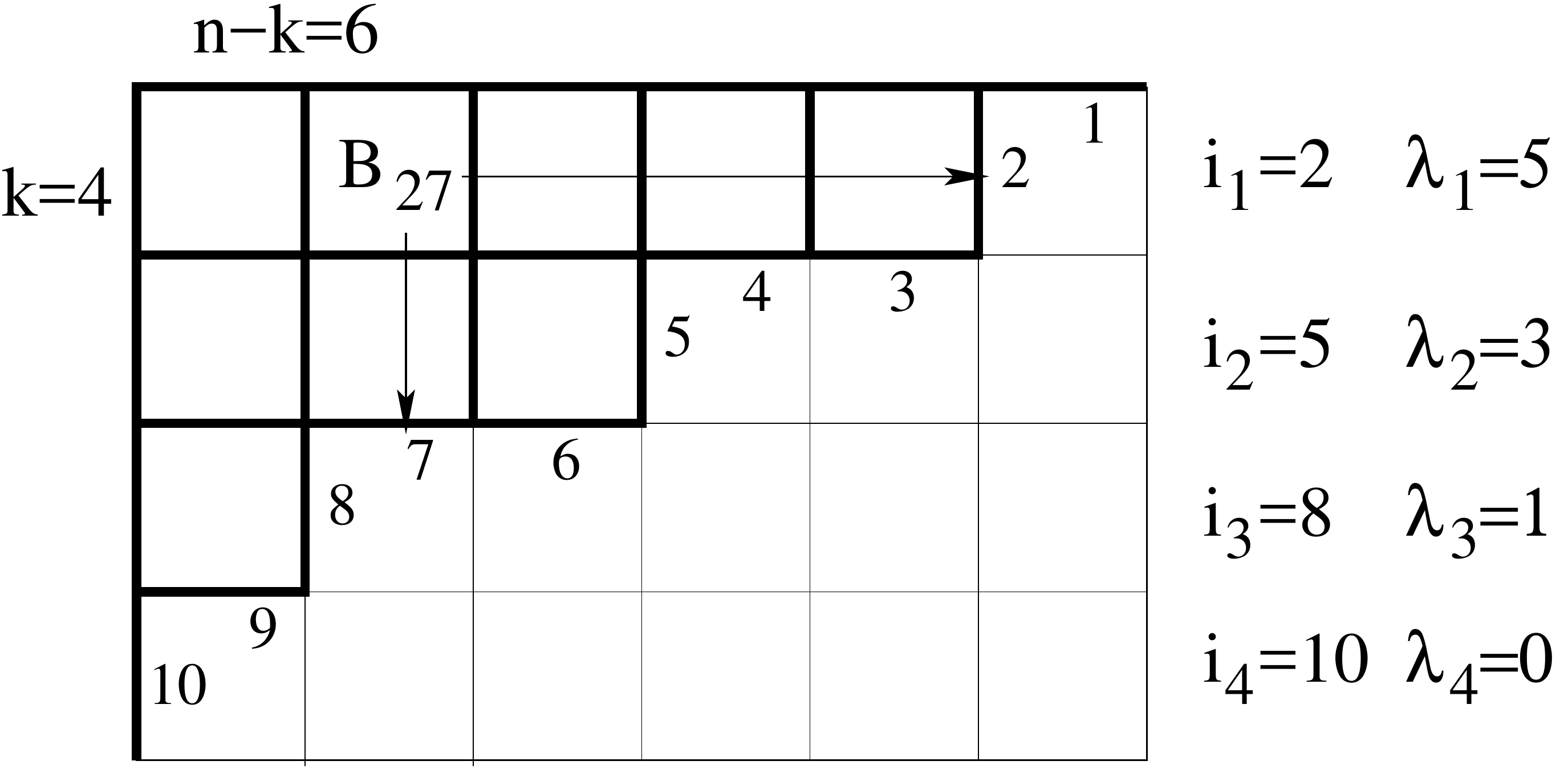}}
  \caption{\label{fig:young}\footnotesize{\sl The Young diagram associated to the partition $(5,3,1,0)$, $k=4$, $n=10$.}}
\end{figure}

\begin{remark}
In this paper we use the ship battle rule to enumerate the boxes of the  Young diagram of a given partition $\lambda$. Let $I=I(\lambda)$ be the pivot set of the $k$ vertical steps in the path along the SE boundary of the Young diagram proceeding from the NE vertex to the SW vertex of the $k\times(n-k)$ bound box, and let ${\bar I} = [n]\backslash I$ be the non--pivot set. Then the box $B_{ij}$ corresponds to the pivot element $i\in I$ and the non--pivot element $j\in {\bar I}$ (see Figure~\ref{fig:young} for an example).
\end{remark}

Each stratum in the refined decomposition of $Gr(k,n)$ into matroid strata \cite{GGMS} is composed by the points of the Grassmannian
which share the same set of non-zero Pl\"ucker coordinates. Each Pl\"ucker coordinate is indexed by a base, i.e. a $k$-element subset in $[n]$, and, for a given stratum, the set of these bases forms a matroid ${\mathcal M}$, i.e. for all $I, J\in {\mathcal M}$ for each $i\in I$ there exists $j\in J$ such that $I\backslash\{i\}\cup\{j\}\in 
{\mathcal M} $. Then the stratum ${\mathcal S}_{\mathcal M}\subset Gr(k,n)$   is defined as
\[
{\mathcal S}_{\mathcal M} = \{ [A] \in Gr(k,n) \, : \, \Delta_I (A) \not = 0 \; \iff I\in {\mathcal M}\}.
\]
A matroid ${\mathcal M}$ is called realizable if ${\mathcal S}_{\mathcal M}\ne\emptyset$. The pivot set $I$ 
is the lexicographically minimal base of the matroid ${\mathcal M}$.
In \cite{Pos}, Postnikov studies the analogous stratification for $\Grkn$. 

\begin{definition}\textbf{Positroid cell. \cite{Pos}}
The totally nonnegative Grassmann (positroid) cell $\S$ is the intersection of the matroid stratum ${\mathcal S}_{\mathcal M}$ with the totally nonnegative Grassmannian $\Grkn$:
$$
\S = \{\, GL^+_k \cdot A \in \Grkn \, : \, \Delta_I(A) >0 \mbox{ if } I\in {\mathcal M}, \mbox{ and } \Delta_I(A) =0 \mbox{ if } I\not \in {\mathcal M}\, \}.
$$ 
The matroid $\mathcal M$ is totally nonnegative if the matroid stratum $\S \not = \emptyset$. 
\end{definition}

\begin{example}
$Gr^{\mbox{\tiny TP}} (k,n)$ is the top dimensional cell and corresponds to the complete matroid ${\mathcal M} = \binom{[n]}{k}$. 
\end{example}

\begin{figure}
\includegraphics[scale=0.25]{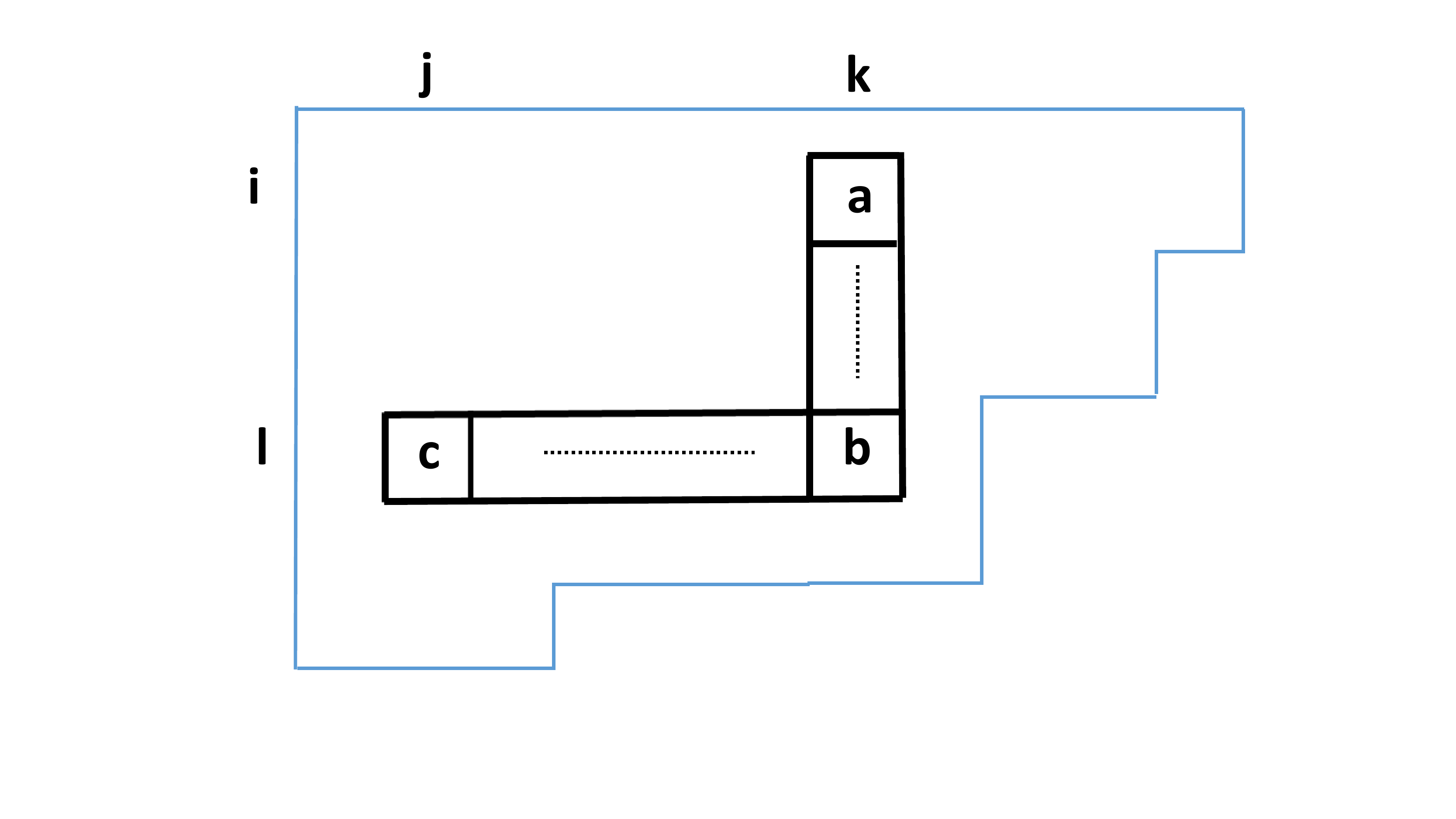}
\vspace{-1.2 truecm}
\caption{\footnotesize{\sl The Le--rule $a,c\not =0$ imply $b\not =0$. }}
\label{fig:lerule}
\end{figure}

A useful tool to classify positroid cells are Le--diagrams and Le--graphs \cite{Pos}. 

\begin{definition}\textbf{Le--diagram and Le--tableau.\cite{Pos}}
For a partition $\lambda$, a Le--diagram $D$ of shape $\lambda$ is a filling of the boxes of its Young diagram with $0$'s and $1$'s such that, for any three boxes indexed $(i,k)$, $(l, k)$, $(l, j)$, where $i<l$ and $k<j$, filled correspondingly with $a,b,c$, if $a,c\not =0$, then $b\not =0$ (see Figure~\ref{fig:lerule}). For such a diagram 
denote by $d$ the number of boxes of $D$ filled with $1$s. 

The Le--tableau $T$  is obtained from a Le--diagram $D$ of shape $\lambda$, by replacing all 1s in $D$ by positive 
numbers $w_{ij}$ (weights).
\end{definition}

\begin{figure}
  \centering
  {\includegraphics[width=0.8\textwidth]{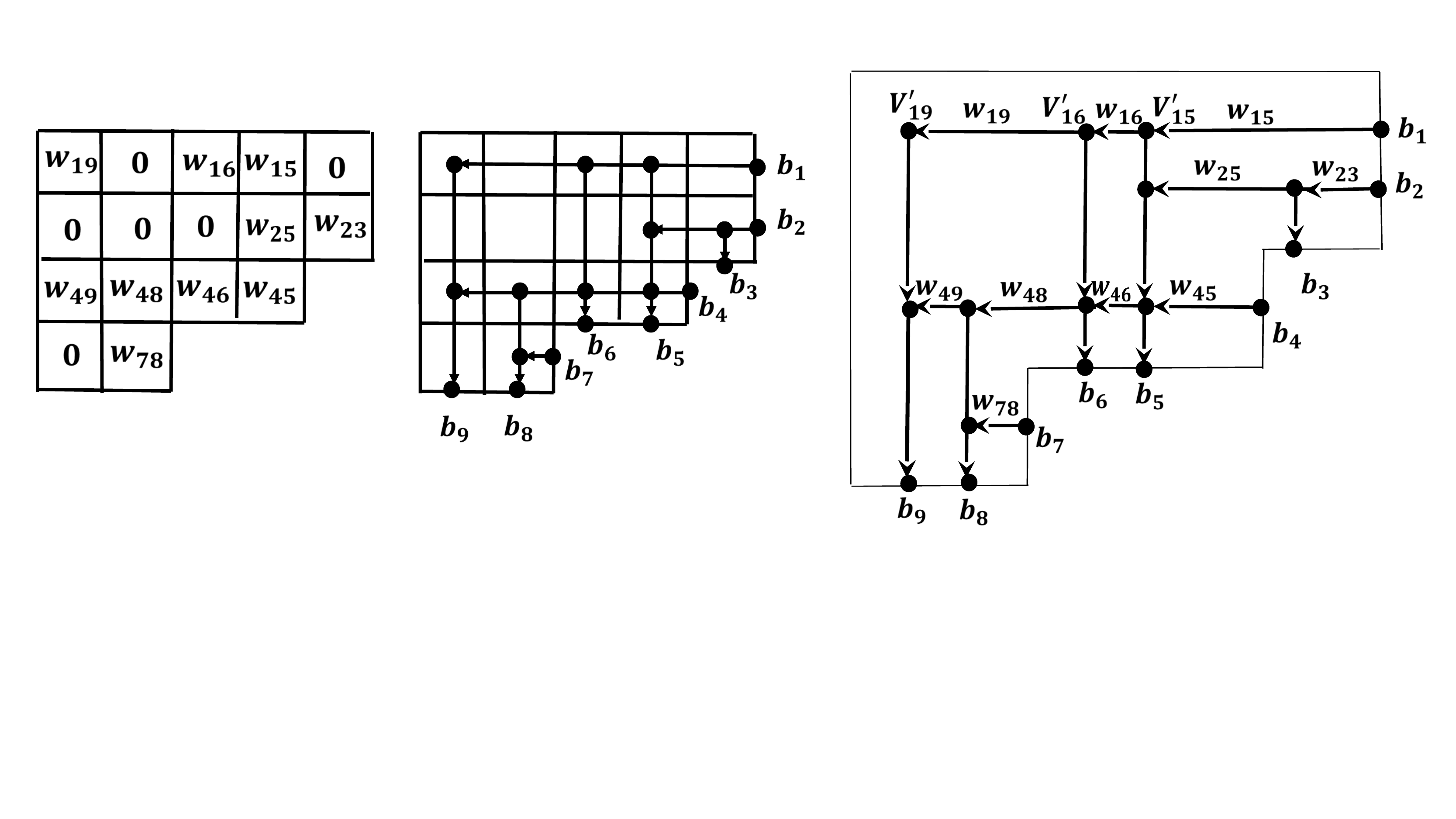}}
		\vspace{-2.5 truecm}
	\caption{\footnotesize{\sl The Le--diagram and the Le--network of Example \ref{ex:gr492}. The horizontal edges are oriented from right to left while the vertical ones from top to bottom.}}
  \label{fig:lediagex2}
\end{figure}

The construction of the Le--graph ${\mathcal G}$ associated to a given Le--diagram is as follows \cite{Pos}.
The boundary of the Young diagram of $\lambda$ gives the lattice path of length $n$ from the upper right corner to the 
lower left corner of the rectangle $k \times (n-k)$.
A vertex is placed in the middle of each step in the lattice path and is marked
$b_1,\dots,b_n$ proceeding NE to SW. The vertices $b_i$, $i\in I\equiv
I(\lambda)$ corresponding to vertical steps are the sources of the network and the
remaining vertices $b_j$, $j\in {\bar I}$, corresponding to horizontal steps are the
sinks. Then the upper right corner is connected to the lower left corner by another path to
obtain a simple close curve containing the Young diagram. For each box of the
Le--diagram $(i,j)$ filled by 1 an internal vertex $V_{ij}$ is placed in the middle of
the box; from such vertex one draws a vertical  line downwards to the
boundary sink $b_j$ and a horizontal line to the right till it reaches the
boundary source $b_i$. By the Le--property any intersection of such lines is also a
vertex. All edges are oriented either to the left or downwards.

To obtain a Le-network ${\mathcal N}$ from a Le--tableaux of shape $\lambda$ one constructs the Le-graph from the corresponding 
Le-diagram, assigns the weight $w_{ij}>0$ from the box $B_{ij}$ to the horizontal edge $e$ 
which enters $V_{ij}$ and assigns unit weights $w_e=1$ to all vertical edges \cite{Pos}. The 
correspondence between the Le--tableau and the Le-network is illustrated in Figure~\ref{fig:lediagex2}.

The map $T\mapsto {\mathcal N}$ gives the isomorphism ${\mathbb R}^{d}_{>0} \simeq {\mathbb R}^{E(G)}_{>0}$ between the set of Le--tableaux $T$ with fixed Le-diagram $D$ and the set of Le-networks (modulo gauge transformations) with fixed graph $\mathcal G$ corresponding to the diagram $D$ as above. 

Given a Le--tableau $T$ with pivot set $I$ it is possible to reconstruct both the matroid and the representing matrix in reduced row echelon form using the Lindstr\"om lemma.   

\begin{proposition}\cite{Pos}
Let ${\mathcal N}$ be the Le-network associated to the Le--diagram $D$ and let $I$ be the pivot set. For any $k$--elements subset $J\subset [n]$, let $K=I\backslash J$ and $L=J\backslash I$. Then the maximal minor $\Delta_J(A)$ of the matrix $A=A({\mathcal N})$ is given by the following subtraction--free polynomial expression in the edge weights $w_e$:
\[
\Delta_J (A) = \sum\limits_P \prod\limits_{i=1}^r w(P_i),
\]
where the sum is over all non--crossing collections $P=(P_1,\dots,P_r)$ of paths joining the boundary vertices $b_i$, $i\in K$ with boundary vertices $b_j$, $j\in L$.

Let $i_r\in I$, where $r\in [k]$ and $j\in [n]$. Then the element $A^r_j$ of the matrix $A$ in reduced row echelon form (RREF)  associated to the Le--network ${\mathcal N}$ is
\begin{equation}\label{eq:ARREF}
A^{r}_j = \left\{ \begin{array}{ll} 0         &\quad\quad j<i_r,\\
1         &\quad\quad j=i_r,\\
  (-1)^{\sigma_{i_r j}} \sum\limits_{P : i_r \mapsto j} \left( \prod\limits_{e\in P} w_e \right) &\quad\quad j>i_r,
\end{array}
\right.
\end{equation}
where the sum is over all paths $P$ from the boundary source $b_{i_r}$ to the boundary sink $b_j$, $j\in {\bar I}$, and $\sigma_{i_r j}$ is the number of pivot elements $i_s\in I$ such that $i_r<i_s <j$.
\end{proposition}

\begin{example}\label{ex:gr492}
Let us consider the Le--diagram $D$ and Le--network represented in Figure \ref{fig:lediagex2}. 
Then $I= (1,2,4,7)$ and the matrix in reduced row echelon form is
\[\resizebox{\textwidth}{!}{$
A = \left( \begin{array}{ccccccccc}
1 & 0 & 0 & 0 & w_{15} & w_{15}(w_{16}+w_{46}) & 0 & -w_{15}w_{48}(w_{16}+w_{46}) & -w_{15}w_{48}w_{49} (w_{16}+w_{46})-w_{15}w_{16}w_{19} \\
0 & 1 & w_{23} & 0 & -w_{23}w_{25} &-w_{23}w_{25}w_{46} & 0 & w_{23}w_{25}w_{46}w_{48} & w_{23}w_{25}w_{46}w_{48}w_{49}\\
0 & 0 & 0 & 1 & w_{45} & w_{45}w_{46} & 0 & -w_{45}w_{46}w_{48} & -w_{45}w_{46}w_{48}w_{49}\\
0 & 0 & 0 & 0 & 0 & 0 & 1 & w_{78} & 0 
\end{array}
\right).$}
\]
The same example is used in \cite{KW2} to illustrate the combinatorial properties of the KP--soliton tropical asymptotics in the limit $t\to -\infty$.
\end{example}

\begin{remark}\label{rem:redLe}{\bf Reducible positroid cells}
A totally non--negative cell $\S\subset \Grkn$ is reducible if its Le--diagram contains either columns or rows filled by 0s \cite{Pos}. The Le--diagram has the $j$--th column filled by zeros if and only if no base in $\mathcal M$ contains the element $j$. Similarly, the Le--diagram has the $r$-th row filled by zeros if and only if  all bases in $\mathcal M$ contain $i_r$, the $r$--th element in the lexicographically minimal base $I\in \mathcal M$.

In the first case, there is no path in the Le--network with destination $j$, and
the RREF matrix $A$ has the $j$--th column filled by zeroes. One can then shift by one all indexes bigger than $j$ in $\mathcal M$, call $\mathcal M^{\prime}$ the resulting matroid of $k$ element subsets in $[n-1]$, correspondingly eliminate the $j$--th column from the Le--diagram and the 
$j$--th column from the matrix $A$, and
represent the same point in the totally non--negative cell 
${\mathcal S}^{\mbox{\tiny TNN}}_{\mathcal M^{\prime}}\subset Gr^{\mbox{\tiny TNN}} (k,n-1)$.

In the second case, there is no path in the Le--network starting from the boundary source $b_{i_r}$, and
the $r$--th row of $A$ contains just the pivot element. One can then eliminate $i_r$ from all the bases in $\mathcal M$ and shift by one all the indexes greater than $i_r$, call $\mathcal M^{\prime}$ the resulting matroid of $k-1$ elements in $[n-1]$, 
correspondingly eliminate the $r$--th row of the Le--diagram, eliminate the $r$--th row and the $i_r$--th columns from the matrix $A$ and change the sign of all elements of
$A^i_j$ with $i<r$ and $j>i_r$, and represent the same point in the totally non--negative cell ${\mathcal S}^{\mbox{\tiny TNN}}_{\mathcal M^{\prime}}\subset Gr^{\mbox{\tiny TNN}} (k-1,n-1)$.
\end{remark}

\begin{figure}
\includegraphics[scale=0.2]{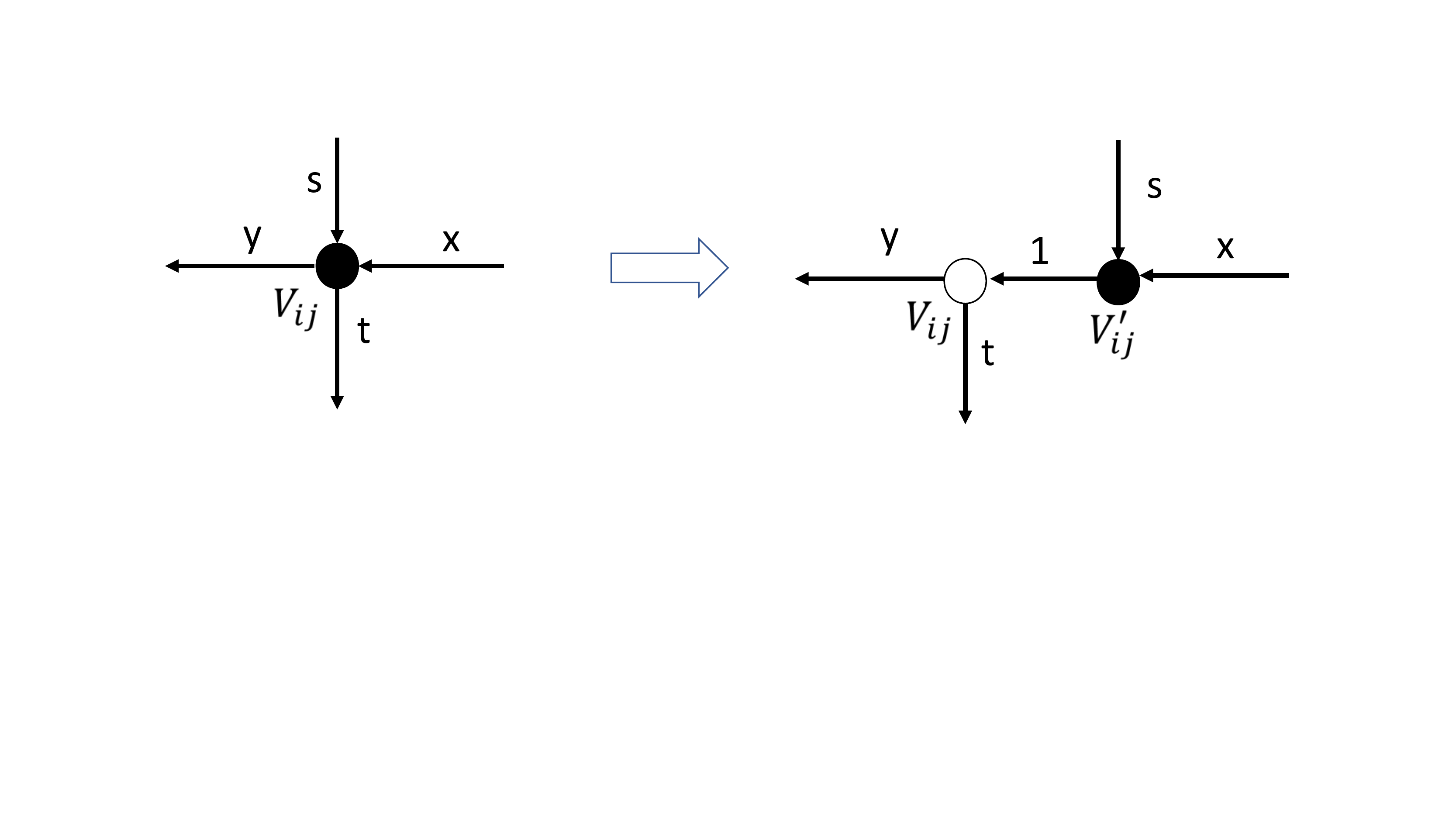}
\vspace{-1.5 truecm}
\caption{\footnotesize{\sl All internal vertices of the Le--network are transformed to trivalent vertices, preserving the boundary measurement map.}}
\label{fig:trivalent}
\end{figure}

In section \ref{sec:gamma} we associate an unique \text{universal} curve to each positroid cell, by modeling the construction of a 
rational degeneration of an $\mathtt M$--curve on the Le--graph: vertices of the graph correspond to copies of $\mathbb{CP}^1$, whereas the edges govern the positions of the double points. To provide a construction of the curve without parameters, we require that each copy of 
$\mathbb{CP}^1$ associated to an internal vertex has three marked points. Moreover, the recursive construction of the wave function 
and the characterization of its divisor is technically simpler if modeled on a bipartite graph where black and white vertices alternate. 
For the above reasons we follow Postnikov's rules to transform the Le--network ${\mathcal N}$ into a planar bipartite perfect network with internal vertices 
of degree at most three, and we continue to denote it with ${\mathcal N}$, since this transformation is well-defined.
We remark that, after such transformation, ${\mathcal N}$ is perfect since each boundary vertex has degree one and each internal vertex in $\mathcal G$ is either the initial vertex of exactly one edge or the final vertex of exactly one edge.
For the Le--graph the only relevant transformation concerns the degree four internal vertices which become couples of trivalent vertices of opposite colour \cite{Pos} (see Figure \ref{fig:trivalent}). Moreover, following Postnikov \cite{Pos}, we assign black color to each internal vertex with exactly one outgoing edge and white color to each trivalent internal vertex with exactly one incoming edge. We also assign black color to all boundary vertices.
Finally we move all boundary vertices to the same line (see Figure \ref{fig:boundary_vertex}). Therefore we get to the following definition.

\begin{figure}
\includegraphics[width=0.65\textwidth]{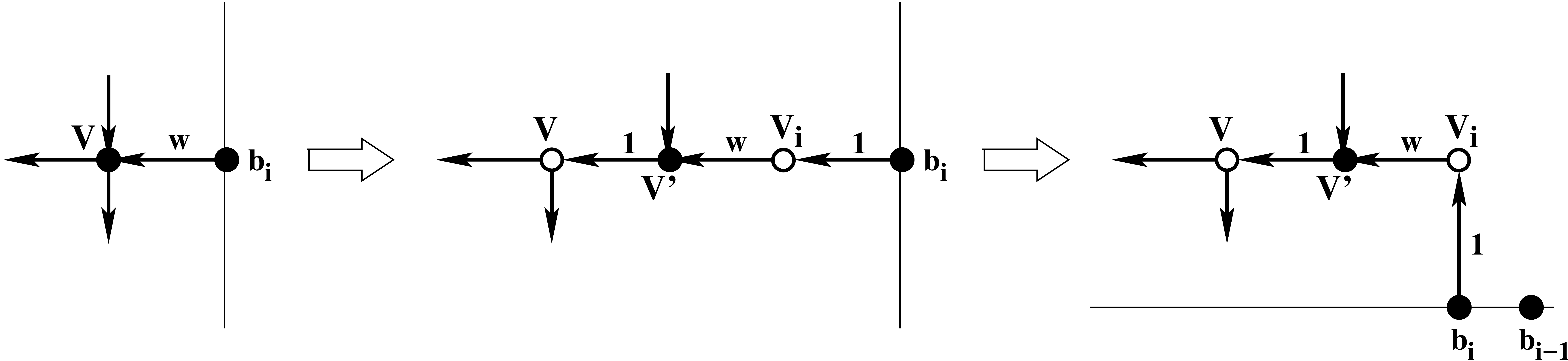}
\caption{\footnotesize{\sl Transformation of graph at the boundary source $b_i$.}}
\label{fig:boundary_vertex}
\end{figure}

\begin{definition}\label{def:can_Le}\textbf{The trivalent bipartite Le--network used to construct the curve $\Gamma({\mathcal G})$}
The acyclically oriented network associated to the Le--tableau $T$ is transformed into a {\bf perfect trivalent bipartite network ${\mathcal N}$} in the disk with the following rules:
\begin{enumerate}
\item If the box $B_{ij}$ of $T$ is filled with 1, the vertex $V_{ij}$ is transformed into a couple of one black vertex $V^{\prime}_{ij}$ and one white vertex $V_{ij}$ (see Figure \ref{fig:trivalent}[left]); following \cite{Pos}
the horizontal edge joining the black vertex $V_{ij}^{\prime} $ to the white vertex $V_{ij}$ has unit weight, while all other weights are unchanged (see Figure \ref{fig:trivalent}[right]);
\item All boundary vertices have black colour and degree one. Any isolated boundary source $b_i$ is joined by a vertical edge to a white vertex $V_i$. If the boundary source $b_i$ is not isolated, we add
a white vertex $V_i$ to the edge of weight $w$ starting at $b_i$, we assign unit weight to the edge joining $b_i$ to $V_i$ and weight $w$ to the other edge at $V_i$ (see Figure \ref{fig:boundary_vertex} middle);
\item All internal vertices corresponding to a given row $r$ in $T$, included $V_{i_r}$, lie on a common horizontal line;
\item The contour of the disk is continuously deformed in such a way that all of the boundary sources and boundary sinks lay on the same horizontal segment and the edge at each boundary vertex is vertical; in this process the positions of all internal vertices are left invariant (see Figure \ref{fig:boundary_vertex} right). 
\end{enumerate}
\end{definition}

In Figure \ref{fig:bipex1} we show the  bipartite Le--network for Example \ref{ex:gr492}.

\begin{figure}
\includegraphics[width=0.49\textwidth]{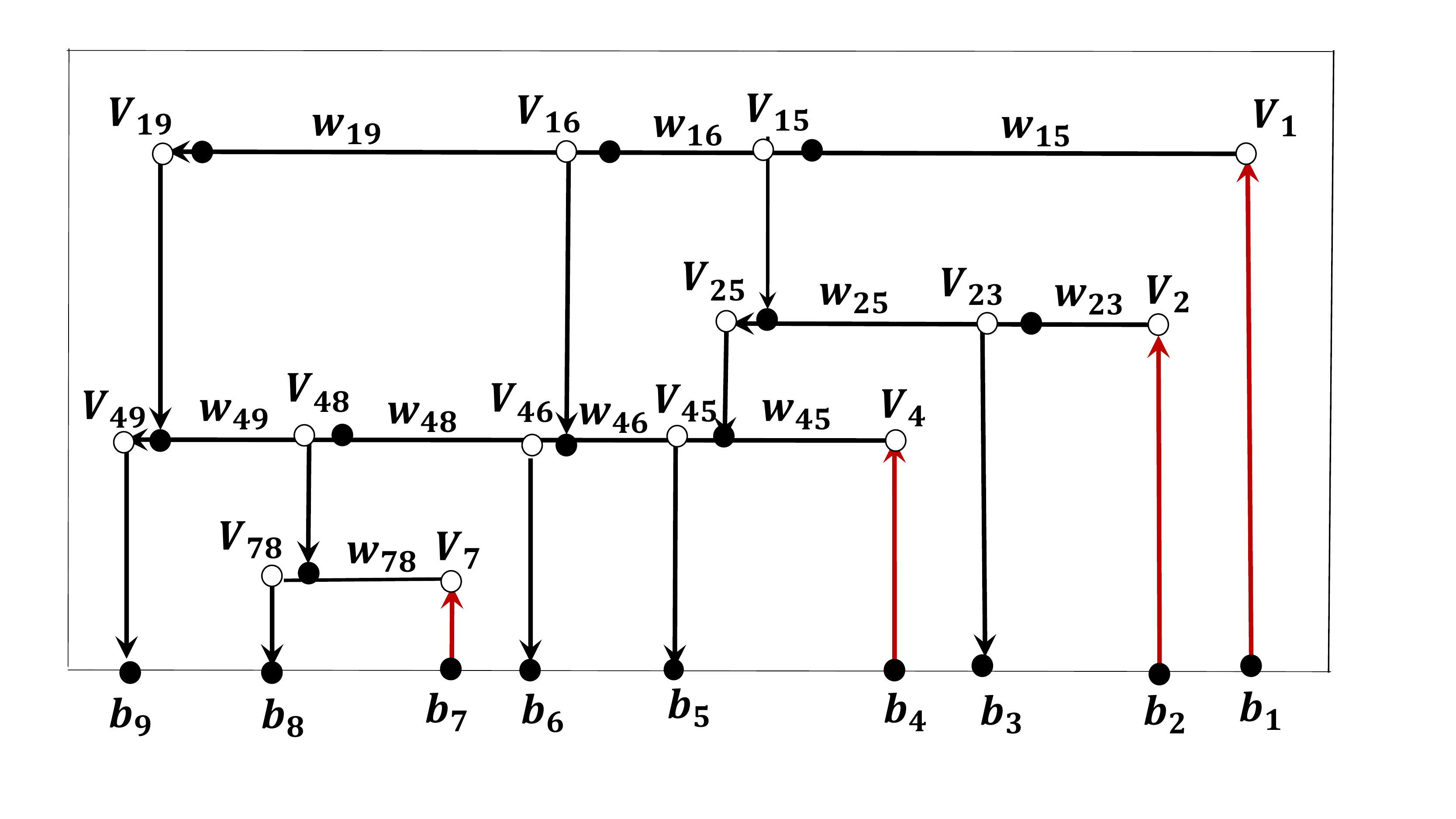}
\vspace{-.6 truecm}
\caption{\footnotesize{\sl The bipartite Le--network for Example \ref{ex:gr492} (see Figure \ref{fig:lediagex2}). The weights refer to the perfect orientation associated to the pivot base $[1,2,4,7]$ of the matroid: the vertical edges starting at the boundary sources $b_1,b_2,b_4$ and $b_7$ are oriented upwards, all other vertical edges are oriented downwards, while all horizontal edges are oriented from right to left.}}
\label{fig:bipex1}
\end{figure}

\begin{remark}
Let ${\mathcal G}$ be the bipartite Le--graph associated to the Le--diagram $D$. Then in ${\mathcal G}$ 
\begin{enumerate}
\item Each black vertex has {\bf at most} one vertical edge; 
\item Each white vertex has {\bf exactly} one vertical edge;
\item The total number of white vertices is $d+k$;
\item If $D$ is irreducible, then the total number of trivalent white vertices is $d-k$, while the total number of trivalent black vertices is $d-n+k$.
\end{enumerate}
\end{remark}

For any $r\in [k]$, we denote $N_{r}$ the number of boxes filled with 1 in the $r$--th row of the Le--diagram $D$. By construction we have 
\begin{equation}\label{eq:Nij}
d \equiv \sum\limits_{i_r\in I} N_{r}, \quad \quad \mbox{ with } N_{r} \equiv \# \, \{ \mbox{ boxes } B_{i_r j}  \mbox{ filled by 1, for }  j \in {\bar I}  \} 
\end{equation}
We also introduce an index to simplify the counting of boxes filled by ones. For any fixed $r\in [k]$, let $1\le j_1< j_2 \dots < j_{N_r}\le n$ be the non--pivot indexes of the boxes $B_{i_r j_s}$, $s\in {\hat N}_r$, filled by one in the $r$-th row. Then
for any $r\in [k]$, we define the index 
\begin{equation}\label{eq:chiindex}
\chi^{i_r}_{l} = \left\{ \begin{array}{ll} 
1 &\quad \mbox { if there exists } s\in [N_r] \mbox{ such that } l=j_s, \\
0 &\quad \mbox { if } B_{i_r l} \mbox { is filled by 0  or } l<i_r. \\
\end{array}\right.
\end{equation}

\bibliographystyle{alpha}

\end{document}